\pdfoutput=1
\RequirePackage{fix-cm}
\documentclass[11pt]{article}
\usepackage[T1]{fontenc}
\usepackage[utf8]{inputenc}
\usepackage[margin=1in]{geometry}
\usepackage{mathtools,amssymb,amsthm}
\usepackage[varbb]{newpxmath}
\usepackage{microtype}
\usepackage{etoolbox}
\usepackage{booktabs,array,tabularx}
\usepackage{float}
\usepackage{graphicx}
\usepackage{multicol}
\usepackage{caption}
\usepackage{dashrule}
\usepackage[inline]{enumitem}
\usepackage{acro}
\usepackage{xcolor}
\usepackage{tcolorbox}
\tcbuselibrary{skins,breakable}
\usepackage{needspace}
\usepackage{aliascnt}
\usepackage{textcase}
\definecolor{darkcyan}{RGB}{0,110,110}
\definecolor{brightblue}{RGB}{0,90,210}
\definecolor{darkdarkred}{RGB}{110,0,0}
\definecolor{revisionold}{RGB}{128,0,0}
\definecolor{revisionneworange}{RGB}{180,85,0}
\definecolor{revisionnewgreen}{RGB}{0,100,0}
\newenvironment{oldversion}
  {\par\begingroup\color{revisionold}}
  {\par\endgroup}

\newenvironment{newversion}
  {\par\begingroup\color{revisionneworange}}
  {\par\endgroup}

\usepackage[hypertexnames=false]{hyperref}
\definecolor{blueGrotto}{HTML}{059DC0}
\definecolor{royalBlue}{HTML}{057DCD}
\definecolor{navyBlue}{HTML}{0B579C}
\hypersetup{
  colorlinks=true,
  urlcolor=blueGrotto,
  linkcolor=royalBlue,
  citecolor=navyBlue
}
\newtagform{referenceblue}[\textcolor{royalBlue}]{\textcolor{black}{(}}{\textcolor{black}{)}}
\usetagform{referenceblue}

\usepackage[
  backref=true,
  backend=biber,
  natbib=true,
  style=alphabetic,
  sorting=alphabeticlabel,
  sortcites=true,
  minbibnames=3,
  maxbibnames=999,
  mincitenames=4,
  maxcitenames=4,
  minalphanames=4,
  maxalphanames=4
]{biblatex}
\DeclareSortingTemplate{alphabeticlabel}{
  \sort[final]{\field{labelalpha}}
  \sort{\field{year}}
  \sort{\field{title}}
}
\AtBeginRefsection{\GenRefcontextData{sorting=ynt}}
\AtEveryCite{\localrefcontext[sorting=ynt]}
\usepackage{xurl}
\usepackage[capitalize,nameinlink,nosort]{cleveref}
\crefname{equation}{Equation}{Equations}
\Crefname{equation}{Equation}{Equations}
\crefname{subequation}{Equation}{Equations}
\Crefname{subequation}{Equation}{Equations}
\crefname{theorem}{Theorem}{Theorems}
\crefname{lemma}{Lemma}{Lemmas}
\crefname{proposition}{Proposition}{Propositions}
\crefname{corollary}{Corollary}{Corollaries}
\crefname{definition}{Definition}{Definitions}
\crefname{assumption}{Assumption}{Assumptions}
\crefname{example}{Example}{Examples}
\crefname{remark}{Remark}{Remarks}
\newcommand{\appref}[1]{\hyperref[#1]{Appendix~\ref*{#1}}}

\newcolumntype{Y}{>{\raggedright\arraybackslash}X}
\setlist[enumerate]{label=\textup{(\roman*)},leftmargin=*,itemsep=2pt}
\setlist[itemize]{leftmargin=*,itemsep=2pt}
\newenvironment{compactdisplay}{%
  \begingroup
  \setlength{\abovedisplayskip}{3pt plus 1pt minus 1pt}%
  \setlength{\belowdisplayskip}{3pt plus 1pt minus 1pt}%
  \setlength{\abovedisplayshortskip}{1.5pt plus 1pt}%
  \setlength{\belowdisplayshortskip}{3pt plus 1pt minus 1pt}%
}{\endgroup}
\newcommand{\R}{\mathbb R}
\newcommand{\val}{\operatorname{Val}(A)}
\newcommand{\spnorm}[1]{\left\lVert #1\right\rVert_{\mathrm{sp}}}
\DeclareMathOperator*{\argmax}{arg\,max}
\DeclareMathOperator*{\argmin}{arg\,min}
\DeclareMathOperator{\br}{br}
\DeclareMathOperator{\supp}{supp}
\DeclareMathOperator{\dist}{dist}
\DeclareMathOperator{\Cov}{Cov}
\DeclareMathOperator{\Var}{Var}

\DeclareMathOperator{\Rew}{\mathfrak{Rew}}
\DeclareMathOperator{\Surplus}{\mathfrak{Surplus}}
\DeclareMathOperator{\vag}{VAG}
\DeclareMathOperator{\lag}{LAG}

\newcommand{\vagd}{\vag}
\newcommand{\lagd}{\lag}
\newcommand{\pto}{\xrightarrow{\mathbb P}}
\newcommand{\dto}{\xrightarrow{\mathrm d}}
\newcounter{noiseregime}

\crefname{noiseregime}{Regime}{Regimes}
\Crefname{noiseregime}{Regime}{Regimes}
\newcommand{\regimetag}[1]{\textcolor{darkdarkred}{Regime~(#1)}}
\newcommand{\regimeparagraph}[2]{%
  \par\smallskip\noindent
  \refstepcounter{noiseregime}\label{regime:#1}%
  \textbf{\regimetag{#1}: #2}\enspace}
\DeclareRobustCommand{\regime}[1]{\hyperref[regime:#1]{\regimetag{#1}}}
\DeclareRobustCommand{\regimeref}[1]{\hyperref[regime:#1]{\textcolor{darkdarkred}{(#1)}}}
\DeclareRobustCommand{\Regimes}{\textcolor{darkdarkred}{Regimes}}

\DeclareAcronym{ftrl}{
  short = FTRL,
  long = follow-the-regularized-leader
}

\newtheoremstyle{tintedresult}
  {0pt}{0pt}{\itshape}{0pt}
  {\color{brightblue}\bfseries}{.}{0.5em}
  {\thmname{#1}\thmnumber{ #2}\thmnote{ \mdseries(#3)}}
\theoremstyle{tintedresult}
\newtheorem{theorem}{Theorem}[section]
\newaliascnt{lemma}{theorem}
\newtheorem{lemma}[lemma]{Lemma}
\aliascntresetthe{lemma}
\newaliascnt{proposition}{theorem}
\newtheorem{proposition}[proposition]{Proposition}
\aliascntresetthe{proposition}
\newaliascnt{corollary}{theorem}
\newtheorem{corollary}[corollary]{Corollary}
\aliascntresetthe{corollary}
\tcbset{resultbox/.style={
  enhanced, breakable, lines before break=3, frame hidden, sharp corners,
  colback=brightblue!4,
  borderline west={1.2pt}{0pt}{brightblue},
  left=4pt, right=4pt, top=4pt, bottom=4pt, boxsep=0pt,
  before skip=7pt plus 2pt, after skip=7pt plus 2pt
}}
\tcolorboxenvironment{theorem}{resultbox}
\tcolorboxenvironment{lemma}{resultbox}
\tcolorboxenvironment{proposition}{resultbox}
\tcolorboxenvironment{corollary}{resultbox}
\newtcolorbox{proofstrategy}[1]{
  enhanced, breakable, colback=cyan!3, colframe=brightblue!65,
  coltext=red, boxrule=0.45pt, arc=2pt,
  left=7pt, right=7pt, top=6pt, bottom=6pt, boxsep=0pt,
  before skip=8pt plus 2pt, after skip=5pt plus 1pt,
  fontupper=\normalfont\small,
  before upper={\textbf{Proof strategy for #1.}\par\smallskip}
}
\BeforeBeginEnvironment{proofstrategy}{\Needspace{11\baselineskip}}
\newtheoremstyle{compactdefinition}
  {0pt}{0pt}{\normalfont}{0pt}{\bfseries}{.}{0.5em}{}
\theoremstyle{compactdefinition}
\newaliascnt{definition}{theorem}
\newtheorem{definition}[definition]{Definition}
\aliascntresetthe{definition}
\newaliascnt{assumption}{theorem}
\newtheorem{assumption}[assumption]{Assumption}
\aliascntresetthe{assumption}
\newaliascnt{example}{theorem}

\aliascntresetthe{example}
\colorlet{remarkcolor}{darkcyan}
\newtheoremstyle{tintedremark}
  {0pt}{0pt}{\itshape}{0pt}
  {\color{remarkcolor}\bfseries}{.}{0.5em}
  {\thmname{#1}\thmnumber{ #2}\thmnote{ \mdseries(#3)}}
\theoremstyle{tintedremark}
\newaliascnt{remark}{theorem}
\newtheorem{remark}[remark]{Remark}
\aliascntresetthe{remark}
\tcolorboxenvironment{remark}{
  enhanced, breakable, frame hidden, sharp corners,
  fontupper=\small,
  colback=remarkcolor!6,
  borderline west={1.2pt}{0pt}{remarkcolor},
  left=5pt, right=4pt, top=3pt, bottom=3pt, boxsep=0pt,
  before skip=6pt plus 2pt, after skip=6pt plus 2pt
}

\title{Noise Shapes the Surplus:\\
Strategic Play Against No-Regret Learners}
\date{}
\author{
Yiheng Su\\
University of Wisconsin--Madison\\
\texttt{ysu24@cs.wisc.edu}
\and
Emmanouil-Vasileios Vlatakis-Gkaragkounis\\
University of Wisconsin--Madison\\
\texttt{vlatakis@wisc.edu}
}
\hypersetup{
  pdftitle={Noise Shapes the Surplus: Strategic Play Against No-Regret Learners},
  pdfauthor={Yiheng Su and Emmanouil-Vasileios Vlatakis-Gkaragkounis}
}

\begin{document}
\raggedbottom
\pagenumbering{roman}
\maketitle
\bigskip
\bigskip
\thispagestyle{empty}
\begin{abstract}
Recent work has shown that,
in zero-sum games,
an optimizer can extract $\Omega(\sqrt{T})$ \emph{strategic surplus}---cumulative
payoff above the minimax benchmark---against noiseless no-regret
learners with a constant step size.
However,
in practice,
learners often act on perturbed payoff scores
and use decaying learning rates $\eta_t$.
Since existing extraction mechanisms rely on exact
scores,
a natural question is whether noise shields
the learner from exploitation and eliminates surplus
altogether.

We address this question for the widely
used class of \emph{Follow-the-Regularized-Leader} (FTRL) algorithms.
Under standard regularity conditions,
we establish three results,
none of which assumes independence of the
noise.
First,
under summability conditions,
passive play of a strategy fixed independently
of the horizon earns at most constant
surplus.
Second,
we design a new mechanism that exploits
the learner's noise-dependent response:
a clairvoyant optimizer who knows the update
rule and the limiting mean response secures
$\Omega(1/\eta_T)$ surplus whenever that mean is profitable.
Third,
even when that mean is not profitable,
active alternating play extracts $\Omega(\sum_{t<T}\eta_t)$ surplus by
exploiting the learner's movement rather than its
location,
under suitable responsiveness and stability conditions.

Finally,
at the canonical learning rate $\eta_t\propto t^{-1/2}$,
we recover the noiseless $\Omega(\sqrt T)$ rate for
mean-zero,
finite-variance perturbations---whether fresh with arbitrary temporal dependence
or accumulated from i.i.d.\ increments---provided the limiting mean
response is profitable.
\end{abstract}

\clearpage
\tableofcontents
\thispagestyle{empty}
\clearpage
\pagenumbering{arabic}

\section{Introduction}\label{sec:intro}

Consider a penalty kick:
the kicker chooses a side,
the goalkeeper dives,
and exactly one of them benefits.
This is matching pennies played for a
living---one of the few games in which
professionals have been found to use minimax
strategies \citep{chiappori2002testing,palacioshuerta2003professionals}.
Now suppose that the goalkeeper learns instead.
Before each kick,
he uses the kicker's record to lean
toward the side that has historically been
most rewarding to cover,
as prescribed by multiplicative weights or,
more generally,
\ac{ftrl}.
Because the update rule is public,
the kicker can anticipate his response.
She now plays against a \emph{rule} as
well as an opponent:
a kick to the left today may
matter both for its immediate payoff and
for how it changes the goalkeeper's behavior
tomorrow.

Game theory gives this asymmetry precise arithmetic.
If the goalkeeper's rule has sublinear regret---if
after $T$ kicks he performs asymptotically as
well as the best fixed side in
hindsight---then the kicker's cumulative payoff is at
most $T\val+o(T)$,
where $\val$ is the minimax value of
the game.
This ceiling comes from the zero-sum structure:
every additional unit earned by the kicker
is a unit lost by the goalkeeper,
so his regret directly bounds her advantage.
We call payoff above $T\val$ the \emph{strategic
surplus}.
A no-regret guarantee rules out linear surplus
but leaves open the size and mechanism
of the remaining sublinear gain.

The potential for exploitation is even greater
in general-sum games,
where the players' utilities need not offset
one another.
A no-regret guarantee still protects the learner
relative to her own benchmark,
but it places no comparable upper bound
on the optimizer's payoff.
An optimizer who anticipates the learner's behavior
can therefore obtain $\Omega(T)$ gains \citep{deng2019strategizing,mansour2022strategizing,brown2023learning}.
This possibility motivates work on Stackelberg-style benchmarks,
dynamic contracts,
and repeated auctions designed around learner behavior
\citep{guruganesh2024contracting,lin2025generalized,braverman2018selling,cai2023selling,rubinstein2024strategizing}.
Against this broader backdrop,
two-player zero-sum games pose a subtler extraction
problem:
linear gains are impossible,
so how large can the remaining sublinear
surplus be?
Only recently have two efficient mechanisms been
identified \citep{assos2024maximizing,su2026nocoin}.

The passive and active mechanisms of \citet{su2026nocoin}
exploit two distinct sources of surplus.
Under passive play,
the optimizer fixes a max--min strategy and
profits while the learner still assigns weight
to suboptimal responses.
For noiseless \ac{ftrl} with a constant learning
rate $\eta$,
this yields a transient surplus of order
$N_{\rm sub}/\eta$,
where $N_{\rm sub}$ counts responses that are strictly
suboptimal against that strategy.

The penalty-kick example illustrates why active play
matters.
In matching pennies,
both dives are best responses to the
kicker's equilibrium mixed strategy,
so $N_{\rm sub}=0$ and fixed max--min play earns
no surplus.
By alternating her strategy,
the kicker can instead profit from the
goalkeeper's adjustments,
even though neither dive is suboptimal at
equilibrium.
Under suitable conditions on the game and
regularizer,
this movement yields $\Omega(\eta T)$ surplus.

\noindent\textbf{\emph{Uncertainty}: an overlooked complication.}
Both analyses,
nevertheless,
rely on two idealizations:
$(i)$ \emph{exact payoff feedback} and $(ii)$ \emph{a
constant learning rate}.
Neither is generally warranted in practice.
A goalkeeper combines video evidence with potentially
inaccurate reports about a kicker's preferences \citep{murta2022goalkeepers};
an online advertising platform estimates click-through rates
from binary feedback \citep{zhang2024online};
and algorithms for large imperfect-information games estimate
counterfactual values through sampling \citep{lanctot2009montecarlo}.
In each case,
the learner updates from noisy cumulative-payoff estimates.
At the same time,
anytime implementations of \ac{ftrl} commonly use decaying
learning rates such as $\eta_t\propto t^{-\rho}$ \citep{cesabianchi2006prediction,orabona2019modern}.

\noindent\textbf{Technical challenges.}
Noise and learning-rate decay disrupt the two
extraction mechanisms in different ways.
Noise makes it difficult to maintain exact
score ties and introduces randomness into the
learner's selected mixed strategy.
A decaying learning rate changes the learner's
sensitivity over time and may render its
transient mass on suboptimal actions summable.
From a control perspective \citep{deng2019strategizing},
the optimizer must now steer a stochastic,
time-varying dynamical system rather than the deterministic,
time-homogeneous system generated by exact feedback and
a constant learning rate.
It is therefore unclear whether uncertainty protects
the learner or merely requires an adaptation
of the manipulation mechanism and rate of
surplus extraction.

\noindent These challenges lead to two central
questions:\vspace{0.125em}\par
\noindent
{%
\begin{minipage}{0.94\linewidth}
\begin{itemize}[leftmargin=2.5em,before=\raggedright]
\item[{\hyperref[sec:learner-dynamics-maxmin]{\textcolor{darkcyan}{\textbf{(Q1)}}}}]
\emph{Can feedback noise protect a no-regret learner
from a learner-aware optimizer?}
\item[{\hyperref[sec:exploitation]{\textcolor{darkcyan}{\textbf{(Q2)}}}}]
\emph{Which mechanisms still allow a learner-aware optimizer
to extract strategic surplus under noisy feedback,
and at what rates?}
\end{itemize}
\end{minipage}}

\noindent\textbf{Our model.}
To study these questions,
we consider a learner who applies \ac{ftrl}
with learning rate $\eta_t$ to a cumulative
payoff vector perturbed by a general stochastic
process $(Z_t)_{t\geq 1}$.
The noise may be dependent across actions
and over time.
The optimizer never observes its realization;
whether knowledge of its law is required
depends on the guarantee under consideration.

This formulation captures several practically relevant sources
of uncertainty.
For example,
$Z_t$ may aggregate errors in scouting records,
click-value estimates \citep{zhang2024online},
or sampled counterfactual values \citep{lanctot2009montecarlo}.
A single sampled trajectory can perturb several
counterfactual values at once,
creating dependence across actions,
while delayed observations and Markovian sampling can
create dependence over time
\citep{joulani2013delayed,mou2022markovian}.
Thus,
the model accommodates estimates derived from shared
observations or from a common underlying trajectory.

\noindent\textbf{Our results at a glance.}
Our analysis is twofold: we characterize how noise shapes the
learner's response and how an optimizer turns that response into surplus.

\noindent$\star$\textsc{Learner's side}
\begin{itemize}[leftmargin=0.85em,labelsep=0.35em,itemsep=1pt,topsep=0pt,parsep=0pt,partopsep=0pt]
\item \textbf{Localization and selection.}
Our key technical insight explains how
the learner approaches the best-response face
of a fixed max--min strategy.
Within that face,
the relative noise leads to either
a deterministic mixed strategy chosen by
the regularizer,
a random mixed strategy shaped by
the noise and regularizer,
or a random nearly pure action
(\cref{subsec:main-noise-regimes}).
\end{itemize}

\noindent$\star$\textsc{Optimizer's side}
\begin{itemize}[leftmargin=0.85em,labelsep=0.35em,itemsep=1pt,topsep=0pt,parsep=0pt,partopsep=0pt]
\item \textbf{Limits of passive play.}
Under summable response and regular noise tails\footnote{For
\(\eta_t\propto t^{-\rho}\) with \(0<\rho<1\),
these conditions hold for entropy and
Euclidean regularization under centered finite-variance noise,
including Gaussian,
Laplace,
and bounded uniform laws,
either fresh with arbitrary temporal dependence
or accumulated from i.i.d.\ increments.
For Tsallis-\(q\),
also require \(\rho<q\)
(\cref{sec:noise-examples}).},
any constant optimizer strategy independent of
\(T\) earns at most constant surplus
(\cref{thm:main-normal-bounded-surplus}).
For noiseless \ac{ftrl} with constant rate \(\eta\), passive
play yields a transient of order \(N_{\rm sub}/\eta\).
Under our conditions, a decaying rate still leaves the total
surplus bounded, even as \(T, 1/\eta_T\to\infty\).
\vspace{0.5em}

{\it This limitation motivates the analysis and design of two active mechanisms:}
\item \textbf{Surplus from selection.}
A clairvoyant
optimizer knows the learner's update
rule and learning-rate schedule
and can accurately estimate its limiting
mean response under the noise
process.
We call this mean \emph{profitable} when
some optimizer strategy earns more than
the minimax value against it.
Given an estimate accurate enough to
identify such a strategy,
we construct small perturbations toward it
from a max--min strategy that secure
\(\Omega(1/\eta_T)\) surplus without knowing \(T\)
(\cref{thm:main-selection-exploitation}).
Unlike passive extraction from transiently suboptimal actions,
this mechanism exploits the response selected within
the best-response face.
\item \textbf{Surplus from movement.}
Even when the limiting mean is unprofitable
or the optimizer cannot estimate it accurately,
we can exploit the learner's movement.
Under responsiveness and stability conditions, an active
alternating play around a max--min strategy
earns \(\Omega(\sum_{t<T}\eta_t)\) surplus
(\cref{thm:main-alternating-exploitation}).
For cumulative noise with \(\eta_t\propto t^{-\rho}\),
nearly pure responses can weaken this
mechanism when \(\rho<1/2\).
\end{itemize}
\begin{remark}[Connection to regret]
For a constant learning rate $\eta$,
the classical \ac{ftrl} analysis gives a regret
bound of the form
\(
\operatorname{Reg}_L(T)\lesssim {1/\eta}+\eta T.
\)
The two terms have different origins:
$1/\eta$ is the initial regularization cost,
whereas $\eta T$ accumulates the stability cost of
successive updates.
Our extraction surplus mechanisms aim to  complementarily saturate
these same scales:
Selection targets the inverse-rate term---$1/\eta_T$ for a
decaying schedule---while movement targets the cumulative-rate term
$\sum_{t<T}\eta_t$,
which reduces to $\eta T$ for a constant
rate.
\end{remark}
\begin{remark}[Beyond passive extraction]
The selection mechanism is also of independent
interest in the noiseless,
constant-rate setting.
Whenever the selected best-response strategy is profitable,
it yields surplus of order $1/\eta$ even
when every learner action is a best
response to the optimizer's max--min strategy.
By contrast,
the passive mechanism of \citet{su2026nocoin} yields surplus
of order $N_{\rm sub}/\eta$,
which vanishes when there are no suboptimal
responses.
\end{remark}
To our knowledge,
this is the first study of strategic
surplus in a model that combines noisy
cumulative payoff scores with a learning rate
that decays during the interaction.

\subsection{Related Work}\label{subsec:related-work}
\noindent\textcolor{black}{$\blacktriangleright$}\ \textbf{\emph{Strategic manipulation under uncertainty.}}
In repeated contracting, \citet{guruganesh2024contracting} exploit
a learning agent's predictable movement through its action space
using a \emph{free-fall} contract. Their analysis covers full and
bandit feedback and shows that uncertainty about the horizon can
reduce the principal's advantage over a static contract.
We study how uncertainty in the learner's payoff scores affects
surplus in zero-sum games, where the optimizer controls its own
play and never observes the realized noise.

\noindent\textcolor{black}{$\blacktriangleright$}\ \textbf{\emph{Manipulation of stochastic feedback.}}
Reward-poisoning attacks steer stochastic bandit algorithms toward
a target arm at low corruption cost
\citep{jun2018attacks,liu2019poisoning,zuo2024attacks}.
Strategic-arm models let the arms manipulate observed rewards
\citep{braverman2019strategic,feng2020strategic}, while attacks on
linear bandits depend on arm geometry and correlations
\citep{wang2022attackable}.
These models give the strategic party control over feedback.
Our optimizer acts through the game's payoffs; the score noise
is exogenous and hidden.

\noindent\textcolor{black}{$\blacktriangleright$}\ \textbf{\emph{Stochastic regularized learning.}}
Under noisy feedback, regularized learning can eliminate dominated
strategies, with rates shaped by regularizer geometry
\citep{bravo2017robustness,giannou2021rate}.
Mixed equilibria can nevertheless be stochastically unstable
\citep{giannou2021survival}, and persistent uncertainty can favor
extreme strategies \citep{cauvin2025impact}.
These results describe the learner's long-run behavior.
To quantify an optimizer's surplus, we also need rates of approach,
the selected response, and its sensitivity to intervention.

\noindent\textbf{Technical questions.}
These two questions connect response dynamics to surplus and guide the
reader through the remainder of the paper:
\begin{itemize}[leftmargin=2.8em,labelsep=0.45em,itemsep=3pt,
  topsep=3pt,parsep=0pt,partopsep=0pt,before=\raggedright]
\item[{\hyperref[sec:learner-dynamics-maxmin]{\textcolor{darkcyan}{\textbf{(T1)}}}}]
\emph{How fast does the learner approach the best-response face,
and which response does noise select within it?}
\item[{\hyperref[sec:exploitation]{\textcolor{darkcyan}{\textbf{(T2)}}}}]
\emph{When do small score changes still move the learner,
and how much surplus can that movement yield?}
\end{itemize}

\noindent\textbf{Organization.}
\Cref{sec:preliminaries} introduces the model and assumptions.
\Cref{sec:learner-dynamics-maxmin,sec:exploitation} answer (T1) and (T2),
respectively, and \cref{sec:noise-examples} treats concrete noise models.

\section{Preliminaries}
\label{sec:preliminaries}

We study a finite zero-sum game with optimizer payoff matrix
\(A\in\mathbb R^{n\times m}\) and learner payoff matrix \(B=-A\).
At round \(t=0,\ldots,T-1\), the players choose \(x_t\in\Delta_n\)
and \(y_t\in\Delta_m\), and the optimizer earns
\(u_o(x_t,y_t)=x_t^\top Ay_t\).
Here \(\Delta_d\) is the probability simplex in \(\mathbb R^d\)
and \(e_j\) its \(j\)-th pure strategy.
We write \(\val:=\max_{x\in\Delta_n}\min_{y\in\Delta_m}x^\top Ay\)
for the one-shot value and
\(\br(x):=\argmin_{j\in[m]}(A^\top x)_j\)
for the learner's pure best responses to \(x\).

Rock--paper--scissors makes the distinction geometric.
Under uniform max--min play, all three learner actions are tied. The regimes differ only in where the learner settles inside
the same response simplex.
\Cref{fig:rps-three-triangles} previews this geometry before the formal
response limits in \cref{sec:learner-dynamics-maxmin}.

\begin{figure*}[t]
    \centering
    \includegraphics[width=\textwidth]{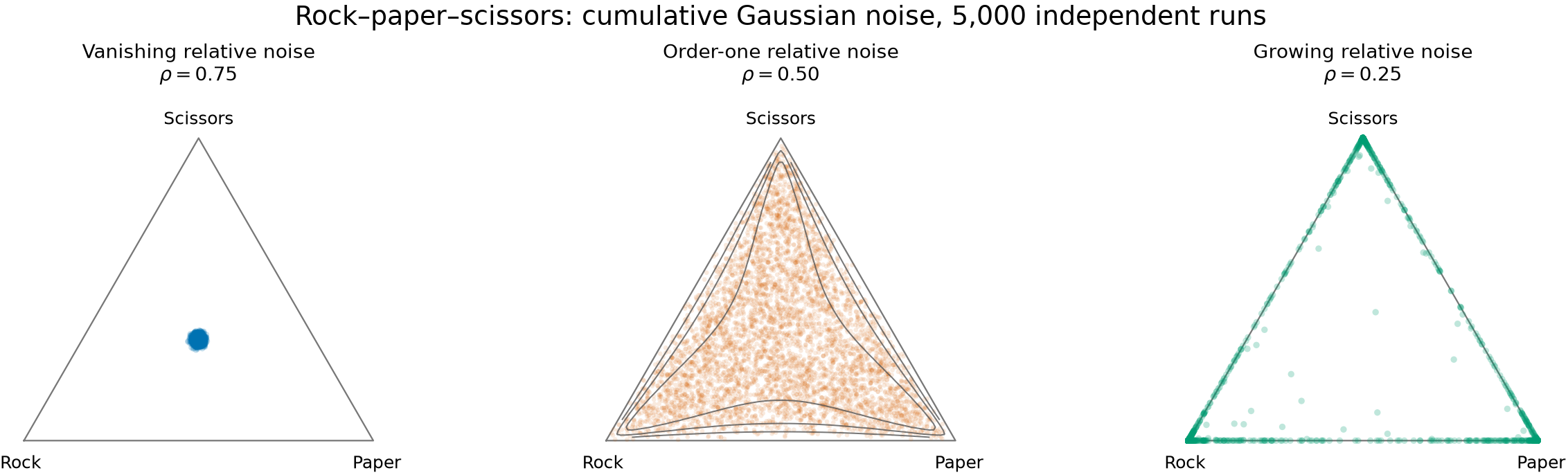}
    \caption{Terminal learner responses in \Regimes~\regimeref{I}--\regimeref{III}.
    Each panel shows \(5000\) independent runs at \(T=10^6\)
    with negative entropy regularization, repeated max--min play,
    cumulative Gaussian noise, and \(\eta_t=(t+1)^{-\rho}\).
    From left to right, responses concentrate near the uniform mixed
    strategy, remain dispersed, or approach the simplex vertices.
    Gray contours show the limiting logistic-normal density.}
    \label{fig:rps-three-triangles}
\end{figure*}

\paragraph{The optimizer.}
Write \(x_{<T}:=(x_0,\ldots,x_{T-1})\) for a trajectory,
\(x_{\times T}:=(x,\ldots,x)\) for constant play, and
\(\bar x_T:=T^{-1}\sum_{t<T}x_t\) for the time average.
The optimizer knows \(A\) and the learner's update rule.
Its play is \emph{oblivious}: the trajectory is fixed independently
of the realized noise, which it never observes.

Expected reward is \(\Rew(x_{<T}):=\mathbb E\sum_{t<T}x_t^\top Ay_t\),
and surplus above the game value is
\begin{compactdisplay}
\begin{equation}
  \Surplus(x_{<T}):=\Rew(x_{<T})-T\val.
  \label{eq:discrete-surplus}\tag{Surplus}
\end{equation}
\end{compactdisplay}
Expectations are over the noise; the trajectory encodes the horizon.
The learner's \emph{expected external regret} is its best fixed
payoff in hindsight minus its expected payoff. Since \(B=-A\),
\begin{compactdisplay}
\begin{equation}
  \operatorname{Reg}_L(x_{<T})=\Rew(x_{<T})-T\min_{y\in\Delta_m}\bar x_T^\top Ay.
  \label{eq:main-learner-regret}\tag{regret}
\end{equation}
\end{compactdisplay}
\paragraph{The learner.}
Noisy Follow-the-Regularized-Leader (FTRL) uses the following deterministic step sizes.
\begin{assumption}[Step sizes]
\label{ass:step-sizes}
\begin{compactdisplay}
\begin{align}
  \label{eqn:step-size-assumption}
  \eta_t > 0, \quad \eta_{t+1} \le \eta_t, \quad t \eta_t \to \infty.
\end{align}
\end{compactdisplay}
\end{assumption}
Its cumulative payoff scores are defined as
\begin{compactdisplay}
\begin{equation*}
  S_t:=\sum_{\tau<t}B^\top x_\tau,\qquad S_0=0.
  \label{eq:cumulative-payoff}\tag{Cumulative Payoff}
\end{equation*}
\end{compactdisplay}
\begin{definition}[Regularizer]
\label{def:regularizer}
A \emph{regularizer} is a continuous, strictly convex function
\(h:\Delta_m\to\mathbb R\).
Strict convexity defines a unique choice map
\(Q_h:\mathbb R^m\to\Delta_m\):
\begin{compactdisplay}
\begin{equation*}
    Q_h(u)
    :=
    \argmax_{y\in\Delta_m}
    \bigl\{\langle u,y\rangle-h(y)\bigr\}.
    \tag{Choice Map}
    \label{eqn:choice-map}
\end{equation*}
\end{compactdisplay}
\end{definition}
\emph{Noisy FTRL.} At round \(t\),
a random perturbation \(Z_t \in \R^m \) is added to
the cumulative score,
and the learner
plays
\begin{compactdisplay}
\begin{equation*}
    y_t
    =
    Q_h\!\left(\eta_t(S_t+Z_t)\right),
    \tag{Noisy FTRL}
    \label{eq:ftrl}
\end{equation*}
\end{compactdisplay}
Let \(\Xi_t = \eta_t Z_t \) be the scaled noise seen
by the choice map,
while \(r_t = t\eta_t\) is the scale of the
accumulated deterministic score.
We allow arbitrary dependence across coordinates and
across rounds.

Only relative scores affect the learner's response. We therefore
measure noise differences by the span seminorm
\(\spnorm{u}:=\max_j u_j-\min_j u_j\).
\begin{assumption}[Score noise]
\label{def:step-size-noise-model}
\label{ass:normal}
Our model uses almost
surely finite score perturbations \(Z_t\in\mathbb R^m\) satisfying,
as \(t\to\infty\), $\spnorm{Z_t} / t \pto0.$
\end{assumption}
By \cref{ass:step-sizes}, \(r_t\to\infty\). Since
\(\spnorm{\Xi_t}/r_t=\spnorm{Z_t}/t\), \cref{ass:normal} gives
\(\spnorm{\Xi_t}=o_{\mathbb P}(r_t)\).
\begin{definition}[Separable regularizer]
\label{def:separable-regularizer}
A regularizer \(h\) is \emph{separable} if
\(h(y)=\sum_{j=1}^m \theta(y_j)\),
where the kernel \(\theta:[0,1]\to\mathbb R\) is continuous and strictly convex,
belongs to \(C^2((0,1])\),
and satisfies $
\theta''(u)>0\text{ for all }u\in(0,1].
$
\end{definition}

For a separable regularizer,
define the truncated inverse derivative
$\phi:=\operatorname{trunc}_{[0,1]}\{(\theta')^{-1}\}.$
If \(y=Q_h(u)\),
the KKT conditions give a scalar \(\lambda\)
such that
\begin{compactdisplay}
\begin{equation*}
\bigl\{y_j=\phi(u_j+\lambda),\;
\textstyle\sum_j y_j=1\bigr\}.
\tag{Sep.\ Choice Map}
\label{eq:separable-choice-map}
\end{equation*}
\end{compactdisplay}
On every positive coordinate,
$\theta'(y_j)=u_j+\lambda,$
while an inactive coordinate satisfies \(u_j+\lambda\leq\theta'(0^+)\).

\begin{definition}
\label{def:steep}
A separable regularizer with kernel \(\theta\) is
\emph{steep} if
$\theta'(0^+)=-\infty,$
and \emph{non-steep} otherwise.
\end{definition}
\paragraph{Examples.}
We focus on three standard separable regularizers:
\begin{itemize}[label=$\triangleright$,topsep=0pt,partopsep=0pt,
  itemsep=1pt,parsep=0pt]
\item \emph{Negative entropy}~\citep{auer1995gambling}.
For \(\theta(u)=u\log u\),
the choice map is the softmax
\(Q_h(v)_j=\exp(v_j)/\sum_{\ell}\exp(v_\ell)\),
and
\(1/\theta''(u)=u\).
\item \emph{Euclidean}~\citep{zinkevich2003online}.
For \(\theta(u)=u^2/2\),
\(Q_h(v)\) is the Euclidean projection of \(v\)
onto \(\Delta_m\),
and \(1/\theta''(u)=1\).
\item \emph{Negative Tsallis}~\citep{abernethy2015fighting}.
For \(q\in(0,1)\),
let
\(\theta(u)=(u^q-u)/(q-1)\).
The response \(y=Q_h(v)\) is characterized by
\(q y_j^{q-1}=1+(q-1)(v_j+\lambda)\),
where \(\lambda\) normalizes
\(\sum_j y_j=1\),
and \(1/\theta''(u)=u^{2-q}/q\).
\end{itemize}

\noindent\textbf{Noise models.}
For nonempty \(K\subseteq[m]\), write \(u_K=(u_j)_{j\in K}\).
We classify relative noise on \(K\), which will be \(\br(x^\star)\)
in \cref{sec:learner-dynamics-maxmin}.
\regimeparagraph{I}{vanishing noise.}
\(\spnorm{\Xi_{t,K}}\pto0\).
\regimeparagraph{II}{finite noise.}
For almost surely finite scalar shifts \(c_t\),
\(\Xi_{t,K}-c_t\mathbf1\dto\Xi_{\infty,K}\),
with an almost surely finite limit.
\regimeparagraph{III}{growing noise.}
For deterministic \(a_t>0\) with \(a_t\to\infty\) and almost surely
finite scalar shifts \(c_t\),
\((\Xi_{t,K}-c_t\mathbf1)/a_t\dto\Xi_{\infty,K}\).
The limit is almost surely finite with an almost surely unique
largest coordinate.

These conditions describe relative noise within \(K\).
The model in \cref{def:step-size-noise-model} controls localization
to the best-response face.



\section{Learner Dynamics under Max--Min Play}
\label{sec:learner-dynamics-maxmin}
Before asking how the optimizer can extract surplus, we identify where
it can afford to play on average.
The next identity is a budget equation: learner regret is the available
budget, while any loss caused by moving the optimizer's average away
from max--min play is deducted from that budget.

\begin{proposition}
\label{prop:must-be-average-maxmin}
For every optimizer's \(x_{<T}\),
\[
\begin{aligned}
  \Surplus(x_{<T})
    &=\operatorname{Reg}_L(x_{<T})-T\left(\val-\min_{y\in\Delta_m}\bar x_T^\top Ay\right).
\end{aligned}
\]
\end{proposition}
The proof is deferred to \appref{app:maxmin-regret}.
This identity leaves little room for the optimizer's average strategy.
If \(\bar x_T\) loses a fixed \(\varepsilon>0\) relative to the game
value, the optimizer pays \(\varepsilon T\), which sublinear regret
cannot offset.
Therefore, if the learner has sublinear expected external regret, every trajectory with nonnegative surplus up to lower-order
terms must satisfy
$  \val-\min_{y\in\Delta_m}\bar x_T^\top Ay=o(1).$

In words, the optimizer must remain asymptotically max--min \emph{on
average}, although it may still make profitable round-by-round
deviations around that baseline.
The formal consequence is proved in \cref{cor:average-maxmin}.

We now freeze the optimizer at a max--min strategy \(x^\star\) and use
the resulting learner response as our reference process.
Let \(K=\br(x^\star)\) be the learner actions tied at the game value,
and let \(\delta=A^\top x^\star-\val\mathbf1\) record the payoff gaps:
\(\delta_j=0\) on \(K\) and \(\delta_j>0\) outside \(K\).
Thus \(\Delta_K=\{y\in\Delta_m:\supp(y)\subseteq K\}\) is the
best-response face, while repeated play gives
\(S_t=-tA^\top x^\star\).

Since common score shifts leave \(Q_h\) unchanged,\footnote{Adding
\(c\mathbf1\) to \(u\) adds \(c\) to every feasible objective, so
\(Q_h(u+c\mathbf1)=Q_h(u)\).}
the learner responds to the \emph{reference score} as
\begin{equation}
  y_t=Q_h(-r_t\delta+\Xi_t).
  \label{eq:main-fixed-response}
\end{equation}
Equation~\eqref{eq:main-fixed-response} separates two geometrically
different effects.
The deterministic term \(-r_t\delta\) acts \emph{normally} to the face
and removes mass from inferior actions; the relative noise
\(\Xi_{t,K}\) acts \emph{tangentially} along the face and determines
how the surviving mass is split among tied best responses.
We analyze these effects in that order.
\paragraph{Normal dynamics: localization to the face.}
The step-size and noise model ensures \(r_t\to\infty\) and
\(\spnorm{\Xi_t}=o_{\mathbb P}(r_t)\).
Thus each positive payoff gap dominates noise on the
accumulated-score scale.

\begin{proposition}
\label{prop:main-localization}
\(\sum_{j\notin K}y_{t,j}\pto0\) for every regularizer \(h\).
\end{proposition}
Each off-face action has a score disadvantage \(r_t\delta_j\),
while noise differences are \(o_{\mathbb P}(r_t)\).
Comparing with an interior mixed strategy close to a best response
makes its finite regularization cost negligible on that scale.
The proof is deferred to \cref{lem:normal} of \appref{app:normal-localization}.

\paragraph{Tangential dynamics: reduction to the face.}
We define the restricted choice map by
\begin{equation}\label{eq:restricted-choice}
  Q_{h,K}(u_K)
  :=\argmax_{y\in\Delta_K}
  \{\langle u_K,y_K\rangle-h(y)\},
\end{equation}
Coordinates outside \(K\) are zero.
Continuity on the compact face gives existence, and strict convexity
gives uniqueness.
\begin{proposition}[Reduction to best-response face]
\label{prop:main-face-reduction}
The response satisfies
\(\|y_t-Q_{h,K}(\Xi_{t,K})\|_1\pto0\).
\end{proposition}
The proof repairs the vanishing off-face mass and uses a uniform
strict-convexity gap to compare the two responses
(\cref{lem:reduction}, \appref{app:face-reduction}).
Both propositions also give convergence in expectation because
the respective quantities are bounded by one and two.
It remains to apply the noise models from \cref{sec:preliminaries}
to the restricted map.

\subsection{Response limits in the three noise regimes}
\label{subsec:main-noise-regimes}

After localization, only one comparison remains: the size of the
relative noise along \(\Delta_K\) versus the unit-scale regularization
penalty.
Vanishing noise lets the regularizer choose; finite noise and
regularization choose jointly; growing noise overwhelms regularization
and selects the largest noisy coordinate.
The theorem below characterizes the response limit in these three regimes.

\begin{theorem}[Response limits]
\label{thm:main-noise-regimes}
Let \(K=\br(x^\star)\).
For the corresponding noise model:
\noindent\regime{I}.
\(y_t\pto y_K:=\argmin_{y\in\Delta_K}h(y)\).
For a separable \(h\), \(y_K=|K|^{-1}\sum_{j\in K}e_j\).

\noindent\regime{II}.
\(y_t\dto Y_\infty:=Q_{h,K}(\Xi_{\infty,K})\).

\noindent\regime{III}.
\(y_t\dto e_{J_\infty}\), with \(J_\infty\) the almost surely
unique maximizer of \(\Xi_{\infty,K}\).

The mean responses converge respectively to \(y_K\),
\(\mathbb E[Y_\infty]\), and
\(\sum_{j\in K}\Pr(J_\infty=j)e_j\).
\end{theorem}

\emph{Proof idea.}
Face reduction transfers the problem to \(Q_{h,K}\).
In \Regimes~\regimeref{I} and~\regimeref{II}, continuity transfers
the relative score limit to the response. In \regime{III}, the
growing score gap at the unique winning coordinate overwhelms
the bounded regularizer. Bounded simplex coordinates give convergence
of means. \Cref{prop:vanishing,cor:uniform-separable-selection,prop:finite,thm:diverging}
provide the proof of the theorem in the different regimes.

The three panels of \cref{fig:rps-three-triangles} are therefore a
geometric dictionary for the theorem: concentration at the regularizer's
center in \regime{I}, persistent dispersion in \regime{II}, and
concentration near vertices in \regime{III}.
For exploitation, the realized randomness is only part of the story:
expected payoff is governed by the mean of the limiting response.
The experimental setup and trajectory plots appear in
\appref{app:rps-noise-regimes}.

\section{Exploiting the Learner's Response}
\label{sec:exploitation}

The previous section tells us where the learner's response lies.
Surplus can be harvested at three different stages of that response.
\emph{Normal exploitation} earns from mass that has not yet reached the
best-response face. \emph{Selection exploitation} earns from where the
learner settles on that face. \emph{Movement exploitation} earns
from how the learner moves along the face after a small score change.
These mechanisms correspond to delay, location, and sensitivity,
respectively, and we study them in that order.

\subsection{Normal Exploitation: Incomplete Localization}
\label{subsec:normal-exploitation}

Normal exploitation is the area under the learner's localization tail.
When the optimizer repeatedly plays a fixed max--min strategy
\(x^\star\), every unit of probability on an inferior action
\(j\notin K\) earns exactly its gap \(\delta_j\):
\[
\Surplus((x^\star)_{\times T})
=
\sum_{t=0}^{T-1}\sum_{j\notin K}
\delta_j\,\mathbb E[y_{t,j}].
\]
The optimizer therefore profits only while localization is incomplete;
the slower the off-face mass disappears, the larger the accumulated
normal surplus.

With noiseless, constant-rate FTRL, the gap \(\eta t\delta_j\)
suppresses an inferior action on a time scale proportional to
\(1/\eta\).
This gives \(\Theta(N_{\rm sub}/\eta)\) surplus for
\(N_{\rm sub}=m-|K|\), up to gap- and regularizer-dependent constants
\citep{su2026nocoin}.
With a decaying rate, the clock becomes \(r_t=t\eta_t\), so the full
response tail must be summed.

To bound this response series, write \(k=|K|\) and,
when \(N_{\rm sub}>0\), define the smallest and largest positive
payoff gaps \(\delta_{\min}=\min_{j\notin K}\delta_j\) and
\(\delta_{\max}=\max_{j\notin K}\delta_j\).

\begin{proposition}[Bounds on normal surplus]
\label{prop:normal-exploitation-informal}
Let $h$ be separable with kernel $\theta$
and truncated inverse derivative
$\phi$,
and fix a max--min strategy $x^\star$.
If $N_{\rm sub}=0$,
then
$\Surplus((x^\star)_{\times T})=0$ for every $T$.
Otherwise,
for any positive learning rate and any
noise process,
write
\(\mathcal S_T^\star:=\Surplus((x^\star)_{\times T})\).
\begin{align*}
\mathcal S_T^\star &\ge N_{\rm sub}\delta_{\min}\sum_{t<T}\!
  \mathbb E\,\phi\!\bigl(\theta'(\tfrac1m)-\delta_{\max}r_t-\spnorm{\Xi_t}\bigr),\\
\mathcal S_T^\star &\le N_{\rm sub}\delta_{\max}\sum_{t<T}\!
  \mathbb E\,\phi\!\bigl(\theta'(\tfrac1k)-\delta_{\min}r_t+\spnorm{\Xi_t}\bigr).
\end{align*}
\end{proposition}
In \appref{sec:passive}, we obtain general fixed-strategy bounds
by summing the payoff-weighted coordinate estimates in
\cref{lem:pointwise-noisy-ftrl-bound}. Their max--min specialization
gives the inequalities above (\cref{prop:normal-surplus-bounds}).

Localization alone makes the per-round gain vanish.
To bound cumulative gain from a fixed strategy, we use summability of both the response tail determined by the regularizer and learning-rate schedule (Eq. \eqref{eqn:summable-response-tail}) and the probabilities of large-noise events (Eq. \eqref{eqn:summable-noise-tail}).

\begin{theorem}
\label{thm:main-normal-bounded-surplus}
Let \(h\) be separable with kernel \(\theta\) and truncated
inverse derivative \(\phi\). Suppose the optimizer is playing a constant strategy \(\hat x\).
For this strategy, write \(K=\br(\hat x)\), \(k=|K|\),
\(N_{\rm sub}=m-k\), and
\(\delta_j=(A^\top\hat x)_j-\min_\ell(A^\top\hat x)_\ell\).
If \(N_{\rm sub}>0\), let \(\delta_{\min}=\min_{j\notin K}\delta_j\)
and suppose that for some \(\alpha\in(0,1)\),
\begin{align}
  \label{eqn:summable-response-tail}
  \sum_{t=0}^{\infty}
\phi\!\left(\theta'(1/k)-\alpha\delta_{\min}r_t\right)&<\infty,\\
  \label{eqn:summable-noise-tail}\sum_{t=0}^{\infty}\Pr\!\left(
\spnorm{\Xi_t}>(1-\alpha)\delta_{\min}r_t
\right)&<\infty.
\end{align}
Then there exists a finite \(C_{\hat x}\)
independent of \(T\)
such that
\[
[\Surplus(\hat x_{\times T})]_+
\le C_{\hat x}
\qquad\text{for every }T.
\]
\end{theorem}

The first condition is equivalent to
\cref{ass:fixed-response-summability}.
The second is implied by \cref{ass:passive-linear-tail}.
The bound is proved in
\cref{thm:unknown-horizon-pointwise}.
\appref{app:normal-standard-regularizer-example} gives the proof and criteria for
standard
regularizers.

\begin{corollary}
\label{cor:main-normal-maxmin-bounded-surplus}
Under the conditions of \cref{thm:main-normal-bounded-surplus},
if \(\hat x=x^\star\) is max--min, then
\[
0\le\Surplus((x^\star)_{\times T})\le C_{x^\star}
\qquad\text{for every }T.
\]
\end{corollary}


Hence $\eta_T\to0$ does not by itself recover
the
$\Theta(1/\eta_T)$ surplus suggested by the constant-rate theory:
under these
conditions,
normal surplus remains $O_{x^\star}(1)$.
\begin{remark}
\Cref{thm:main-normal-bounded-surplus} bounds the positive surplus
from repeating a fixed \(\hat x\), with a constant \(C_{\hat x}\).
It does not rule out growing surplus
from a time-varying $x_{<T}$.
\end{remark}

\subsection{Selection Exploitation: A Profitable Limiting Mean Response}
\label{subsec:main-new-selection}

Localization answers \emph{which} actions survive, but not \emph{where}
the learner settles on the surviving face.
Let \(y_t^\star=Q_h(-r_t\delta+\Xi_t)\) be the response to repeated
\(x^\star\), and suppose
\(\mathbb E[y_t^\star]\to\mu\in\Delta_K\).
Although every point of \(\Delta_K\) is a best response to
\(x^\star\), the selected mean \(\mu\) need not be minimax against all
optimizer strategies.
Its exploitable gap is
$
g_{\rm mean}(x^\star;\mu)
  :=\max_{x\in\Delta_n}x^\top A\mu-\val.
$

When this gap is positive, a best response \(x^\dagger\) would profit
against \(\mu\); the difficulty is that playing \(x^\dagger\) also
changes future learner scores and may destroy \(\mu\).
Selection exploitation therefore faces a preserve-while-extracting
tradeoff.

We move only partway toward \(x^\dagger\):
\(\widehat x_t=(1-\varepsilon_t)x^\star+\varepsilon_t x^\dagger\).
Let \(a=A^\top(x^\dagger-x^\star)\) be the payoff direction and
\(b_t:=\sum_{\tau<t}\varepsilon_\tau\) the cumulative intervention.
Each past deviation adds \(-\varepsilon_\tau a\) to the scores, so
\(y_t=Q_h(-r_t\delta+\Xi_t-\eta_tb_ta)\) see Equation \eqref{eq:new-selection-score} in \appref{app:effect-past-deviations}.

The response change is governed by \(\eta_tb_t\).
Uniform stability of the choice map ensures that keeping this quantity
small preserves a positive fraction of
\(\langle a,\mu\rangle=g_{\rm mean}\), even at large random scores. See \cref{lem:new-selection-uniform-stability} in \appref{app:uniform-stability}.
As \(\eta_t\) decreases, \(b_t\) can grow while its effect stays bounded.
We therefore spend the one-round increase of the reciprocal learning rate:
\begin{equation}
  \varepsilon_t
    =c\left(\frac1{\eta_{t+1}}-\frac1{\eta_t}\right).
\label{eq:old-main-new-selection-weights}
\end{equation}
The sum telescopes to
\(b_t=c(1/\eta_t-1/\eta_0)\),
so
\(\eta_tb_t=c(1-\eta_t/\eta_0)\le c\).
The model's nonincreasing learning rates ensure
\(\varepsilon_t\ge0\). For this construction, we additionally
require bounded reciprocal increments, which allow us to choose
a fixed \(c>0\) small enough that \(\varepsilon_t\le1\)
for every \(t\).
This gives a strategy trajectory using the
known deterministic
schedule,
without knowing the horizon or observing the
noise.

\begin{theorem}
\label{thm:main-selection-exploitation}
Suppose
\(\mathbb E[y_t^\star]\to\mu\in\Delta_K\) with
\(g_{\rm mean}>0\).
Assume additionally that
\(\sup_t(1/\eta_{t+1}-1/\eta_t)<\infty\).
For every sufficiently small fixed \(c>0\),
the sequence above
satisfies,
for some finite \(C\ge0\) independent of \(T\),
\[
  \Surplus(\widehat x_{<T})
       \ge\frac{c\,g_{\rm mean}}{2\eta_T}-C
       \quad \text{for all } T\ge1.
\]
If \(\eta_T\to0\),
its expected surplus is
\(\Omega(1/\eta_T)\).
\end{theorem}

The payoff decomposition \eqref{eq:new-selection-payoff} gives
\(\widehat x_t^\top Ay_t-\val\ge\varepsilon_t\langle a,y_t\rangle\).
Convergence of \(\mathbb E[y_t^\star]\) to \(\mu\), together with
uniform stability, ensures
\(\mathbb E\langle a,y_t\rangle\ge g_{\rm mean}/2\) after some
time \(t_0\).
Each later round therefore earns at least
\(\varepsilon_tg_{\rm mean}/2\) in expectation.
Summing yields a gain proportional to \(b_T\), while the initial
rounds cost only a constant.
\Cref{thm:new-selection-every-horizon} gives the proof for every $T$.

The profitability condition concerns the limiting mean and applies
in all three response regimes (\cref{cor:new-selection-regimes}).
For \(\eta_t=\eta_0(t+1)^{-\rho}\), \(0<\rho<1\), deviations decay
as \(t^{\rho-1}\) and the surplus guarantee grows as \(T^\rho\)
(\cref{cor:new-selection-power-law}).
When \(g_{\rm mean}=0\), we instead turn to the learner's movement.

\subsection{Movement Exploitation: Alternating Play}
\label{sec:movement-exploitation}

Selection exploitation uses where the learner sits, while movement
exploitation uses how the learner reacts.
It can work even when the limiting mean response is minimax.
The optimizer plays a push--pull pair \((x',x'')\): the first play
tilts the learner's relative scores, the second profits from the induced
response, and the pair averages back to the max--min baseline
\(x^\star\).

\begin{assumption}
  \label{ass:main-new-geometry}
  There exist distinct \(i,j\in \br(x^\star)\) and \(\ell\in\supp(x^\star)\)
  with \(A_{\ell i}\ne A_{\ell j}\).
\end{assumption}
Following \citet{assos2024maximizing} and \citet{su2026nocoin},
this permits a pair \(x',x''\in\Delta_n\) with
\((x'+x'')/2=x^\star\) and \(v'=A^\top x'\) nonconstant on \(K\).
\appref{app:new-movement-construction} gives the construction.
Write \(v=A^\top x^\star=\val\mathbf1+\delta\).

We allow an initial baseline period so that the movement gain
can dominate the errors described below.
The trajectory \(x^{\rm alt}_{<T}\) plays \(x^\star\) for the
first \(\min\{2s_0,T\}\) rounds, then repeats \((x',x'')\),
using \(x^\star\) for any unmatched final round.
Since \(A^\top x''=2v-v'\), pair \(s\) earns surplus
\begin{align*}
    x'^\top Ay_{2s}+x''^\top Ay_{2s+1}-2\val =2\langle\delta,y_{2s+1}\rangle
          +\langle v',y_{2s}-y_{2s+1}\rangle.
\end{align*}
The first term is nonnegative, so we focus on the movement payoff
in the second. Each completed pair returns the cumulative scores
to the baseline path, giving
\(S_{2s}=-2sv\) and \(S_{2s+1}=-2sv-v'\).

To isolate the response to \(x'\), we separate it from changes in
the learning rate and noise. Couple
\(\widehat\Xi_s^0\stackrel d=\Xi_{2s}\) and
\(\widehat\Xi_s^1\stackrel d=\Xi_{2s+1}\).
Under oblivious play, expected payoffs depend only on these marginals,
so the coupling need not match the actual temporal dependence.
Holding the noise fixed, connect the two deterministic scores using
\(\bar\eta_s(r)=(1-r)\eta_{2s}+r\eta_{2s+1}\), \(r\in[0,1]\),
and \(y_s(r;\xi)=Q_h(u_s(r;\xi))\), where
\[
  u_s(r;\xi)=-2s\bar\eta_s(r)v-r\eta_{2s+1}v'+\xi.
\]
The endpoints compare the two rounds at fixed noise; we account
for the noise change at \(r=1\) separately.

For separable \(h\), inverse curvature measures response sensitivity.
Set
\(w_{s,r}(a)=1/\theta''(y_{s,a}(r;\widehat\Xi_s^0))\) on positive
coordinates and zero otherwise. Let
\(W_{s,r}=\sum_a w_{s,r}(a)>0\) and
\(\pi_{s,r}(a)=w_{s,r}(a)/W_{s,r}\).
With \(I\sim\pi_{s,r}\), the resulting gain and error terms are
\begin{align*}
  \mathcal K_s&=\mathbb E\int_0^1
    W_{s,r}\Var_{\pi_{s,r}}(v'_I)\,dr,\\
  \mathcal D_s&=\mathbb E\int_0^1
    W_{s,r}|\Cov_{\pi_{s,r}}(v'_I,v_I)|\,dr,\\
  \mathcal N_s&=\mathbb E\bigl|\bigl\langle v',
    Q_h(u_s(1;\widehat\Xi_s^0))-Q_h(u_s(1;\widehat\Xi_s^1))\bigr\rangle\bigr|.
\end{align*}
The variance in \(v'\) measures responsiveness to the new play,
the covariance with \(v\) controls learning-rate drift, and
\(\mathcal N_s\) accounts for the endpoint noise replacement.
Integrating the response derivative from
\cref{lem:new-movement-derivative} gives
\begin{align*}
  \mathbb E\langle v',y_{2s}-y_{2s+1}\rangle \ge \eta_{2s+1}\mathcal K_s
      -2s|\eta_{2s}-\eta_{2s+1}|\mathcal D_s -\mathcal N_s.
\end{align*}
Thus the movement gain must outweigh both errors.

\begin{theorem}
\label{thm:main-alternating-exploitation}
Under the game geometry above, suppose \(h\) is separable with
\(\sup_{0<u\le1}1/\theta''(u)<\infty\).
Assume there exist \(\kappa>0\) and couplings of adjacent
noise marginals such that
\[
\liminf_{s\to\infty}\mathcal K_s\ge\kappa, \qquad
  2s|\eta_{2s}-\eta_{2s+1}|\mathcal D_s+\mathcal N_s
    =o(\eta_{2s+1}).
\]
There is a finite \(s_0\), such that
for every \(T\ge1\),
\[
  \Surplus(x^{\rm alt}_{<T})
    \ge\frac\kappa2
      \sum_{s=s_0}^{\lfloor T/2\rfloor-1}\eta_{2s+1}.
\]
\end{theorem}
Each sufficiently late pair earns at least
\(\kappa\eta_{2s+1}/2\). Baseline and unmatched rounds have
nonnegative surplus. The derivation and proof appear in
\appref{app:new-movement-pair}--\ref{app:new-movement-surplus}.
For \(\eta_t=\eta_0(t+1)^{-\rho}\), \(0<\rho<1\),
the bound is \(\Omega(T^{1-\rho})\) whenever these conditions hold.

\paragraph{When the hypotheses hold.}
The lower bound on \(\mathcal K_s\) asks for useful response
variation with positive probability. For instance, keeping the two
distinguished actions away from zero ensures a positive variance
contribution. The geometry provides their payoff difference; the
regularizer and noise determine whether the learner keeps enough
weight on both. For drift, \(v\) is constant on \(K\), so the
covariance with \(v\) vanishes for a response supported on that face.
The sufficient conditions in \appref{app:new-movement-sufficient}
also control the inverse-curvature weights outside the face, making
this drift negligible after localization.

The noise error concerns the closeness of adjacent marginal laws.
For example, under fresh noise with common marginal
\(z\sim\mathcal D\), one can couple the scaled noises as
\(\widehat\Xi_s^0=\eta_{2s}z\) and
\(\widehat\Xi_s^1=\eta_{2s+1}z\), regardless of their actual temporal
dependence. Under a power-law schedule and finite variance, the
expected span of their difference is \(O(s^{-\rho-1})\), which is
\(o(\eta_{2s+1})\). Payoff stability then controls
\(\mathcal N_s\). Cumulative noise requires a separate check of
this condition, carried out for the models in
\appref{app:movement-noise-column}.

Growing relative noise can instead concentrate the response near
vertices and remove a uniform positive bound on \(\mathcal K_s\).
The matching-pennies example in \appref{app:movement-below-half}
shows the consequence for entropy, Euclidean, and negative Tsallis
regularization: with cumulative i.i.d.\ finite-variance noise of
positive-definite covariance and \(0<\rho<1/2\), alternating play
earns \(\Theta(\sqrt T)\), while every oblivious trajectory earns
at most \(O(\sqrt T)\). Thus the
\(\Omega(T^{1-\rho})\) guarantee depends on persistent
responsiveness as well as the ability to construct an alternating
pair. The next section verifies these conditions for concrete noise
models.


\section{Noise Models and Exploitation Rates}
\label{sec:noise-examples}

We now apply the response limits and exploitation guarantees to concrete
feedback models. \Cref{tab:main-noise-examples} brings the conclusions
together: its rows describe the noise and limiting reference response,
while its columns show the surplus available from normal, selection,
and movement exploitation.

\paragraph{Noise models and common assumptions.}
Fix a max--min strategy \(x^\star\), let \(K=\br(x^\star)\), and use
\[
  \eta_t=\eta_0(t+1)^{-\rho},\qquad
  \eta_0>0,\qquad 0<\rho<1.
\]
The regularizer is negative entropy, Euclidean, or negative
Tsallis-\(q\), \(0<q<1\). The game, regularizer, noise law, and
chosen baseline \(x^\star\) are fixed as \(T\) varies. The common
perturbation law satisfies
\begin{equation*}
  z\sim\mathcal D,\qquad \mathbb E[z]=0,\qquad
  \mathbb E\lVert z\rVert_2^2<\infty.
  \tag{FV}\label{eq:finite-variance-noise}
\end{equation*}
This includes centered Gaussian, bounded, Laplace, and Student-\(t\)
noise with more than two degrees of freedom.

In the \emph{fresh} model, \(Z_t=z_t\), where each \(z_t\) has
marginal law \(\mathcal D\); arbitrary temporal dependence is allowed.
In the \emph{cumulative} model,
\(Z_t=\sum_{\tau<t}z_\tau\), with i.i.d.\ increments of law
\(\mathcal D\). For the cumulative rows, assume additionally
\(\Sigma:=\Cov(z)\succ0\), and let \(G\sim\mathcal N(0,\Sigma)\).
This nondegeneracy excludes ties between limiting Gaussian coordinates.

\paragraph{Reading the table.}
The response column concerns repeated play of \(x^\star\), with
\(y_t^\star=Q_h(-r_t\delta+\Xi_t)\). For the three regularizers,
write \(y_K=|K|^{-1}\sum_{j\in K}e_j\). The normal column uses
the regularizer-dependent scale
\[
  \mathcal N_{h,\rho}(T):=
  \begin{cases}
    1,&h\text{ is entropy or Euclidean},\\
    1,&h\text{ is Tsallis-}q,\ \rho<q,\\
    \log(T+1),&h\text{ is Tsallis-}q,\ \rho=q,\\
    T^{(\rho-q)/(1-q)},&h\text{ is Tsallis-}q,\ \rho>q.
  \end{cases}
\]
Normal entries are upper bounds for a fixed baseline strategy. Selection
and movement entries are lower bounds for the corresponding active
strategies. Every selection entry requires
\(g_{\rm mean}(x^\star;\mu)>0\), where \(\mu\) is the mean of the
limiting response in that row. Every movement entry requires the
game geometry in \cref{ass:main-new-geometry}. The column headings link
to the derivations in the appendix.

\begin{table}[!htbp]
\centering
\small
\setlength{\tabcolsep}{4pt}
\renewcommand{\arraystretch}{1.2}
\begin{tabularx}{\linewidth}{@{}>{\raggedright\arraybackslash}p{0.285\linewidth}>{\raggedright\arraybackslash}p{0.225\linewidth}*{3}{>{\centering\arraybackslash}X}@{}}
\toprule
Noise model
& \hyperref[app:cumulative-finite-variance]{Limiting response}
& \hyperref[app:normal-noise-column]{Normal}
& \hyperref[app:selection-noise-column]{Selection}
& \hyperref[app:movement-noise-column]{Movement}\\
\midrule
Noiseless,
\(0<\rho<1\)
& Uniform on \(K\)
& \(O(\mathcal N_{h,\rho}(T))\)
& \(\Omega(T^\rho)\)
& \(\Omega(T^{1-\rho})\)\\
\addlinespace
Fresh finite-variance,\newline \(0<\rho<1\)
& Uniform on \(K\)
& \(O(\mathcal N_{h,\rho}(T))\)
& \(\Omega(T^\rho)\)
& \(\Omega(T^{1-\rho})\)\\
\addlinespace
Cumulative i.i.d.
finite-variance,
\(\rho>1/2\)
& Uniform on \(K\)
& \(O(\mathcal N_{h,\rho}(T))\)
& \(\Omega(T^\rho)\)
& \(\Omega(T^{1-\rho})^\dagger\)\\
\addlinespace
Cumulative i.i.d.
finite-variance,
\(\rho=1/2\)
& \(Q_{h,K}(\eta_0G_K)\)
& \(O(\mathcal N_{h,\rho}(T))\)
& \(\Omega(\sqrt T)\)
& \(\Omega(\sqrt T)^\dagger\) \\
\addlinespace
Cumulative i.i.d.
finite-variance,
\(0<\rho<1/2\)
& \(e_{\argmax_{j\in K}G_j}\)
& \(O(\mathcal N_{h,\rho}(T))\)
& \(\Omega(T^\rho)\)
& (\hyperref[app:movement-below-half]{\(\star\)})\\
\bottomrule
\end{tabularx}
\caption{Limiting reference responses and expected surplus bounds under the
setup in this section. Normal entries are upper bounds for repeated
\(x^\star\); selection and movement entries are lower bounds for their
respective constructions. Selection requires a profitable limiting mean,
and movement requires \cref{ass:main-new-geometry}.
\(\dagger\): the additional coupling condition stated below.
\(\star\): the nominal \(\Omega(T^{1-\rho})\) movement rate need not
hold; matching pennies has optimal order \(\sqrt T\)
(\appref{app:movement-below-half}). Uniform-response limits hold in
probability; random-response limits hold in distribution.}
\label{tab:main-noise-examples}
\end{table}

\paragraph{Why cumulative noise has a threshold.}
Fresh perturbations vanish after scaling:
\(\Xi_t=\eta_tz_t\pto0\). Consequently, the response converges to
\(y_K\) for every \(0<\rho<1\), just as in the noiseless model.
Cumulative perturbations instead have fluctuation scale
\[
  \eta_t\sqrt t\sim\eta_0t^{1/2-\rho}.
\]
For \(\rho>1/2\), the scaled noise vanishes and the same uniform
response is selected. At \(\rho=1/2\), the central limit theorem
gives \(\Xi_t\dto\eta_0G\), so the response retains the random
mixed limit \(Q_{h,K}(\eta_0G_K)\). For \(0<\rho<1/2\), relative
noise grows and the response approaches the vertex with the largest
Gaussian score. These are \Regimes~\regimeref{I}--\regimeref{III}.
They describe the response under the baseline; the exploitation
columns describe what the optimizer can earn by its chosen strategy.

\newpage
\paragraph{Normal and selection surplus.}
Finite variance makes the large-noise contribution summable for both
models. The remaining normal surplus depends on the regularizer's
response tail: entropy gives a summable exponential tail in \(r_t\),
Euclidean regularization truncates the corresponding response bound,
and Tsallis gives a polynomial tail. Summing it yields the three
Tsallis cases above, including the logarithmic boundary at \(\rho=q\)
(\appref{app:normal-noise-column}).

Selection uses the mean response rather than its realized dispersion.
The response limits in the table give convergence of means because
simplex coordinates are bounded. Whenever the resulting mean is
profitable, \cref{thm:main-selection-exploitation} gives
\(\Omega(1/\eta_T)=\Omega(T^\rho)\) in every row. The noise model
changes which mean is selected, and therefore whether the profitability
condition holds. In particular, a uniform response on \(K\) need not
be minimax, and centered noise alone does not imply a uniform limiting
mean in the two random-response regimes
(\appref{app:selection-noise-column}).

\paragraph{Movement surplus and the two qualifications.}
For noiseless and fresh finite-variance noise, the alternating
construction earns \(\Omega(T^{1-\rho})\) throughout \(0<\rho<1\).
For cumulative noise with \(\rho\ge1/2\), the same rate follows
provided the adjacent scaled noise marginals admit couplings satisfying
\[
  \mathbb E\spnorm{\widehat\Xi_s^0-\widehat\Xi_s^1}
      =o(\eta_{2s+1}).
\]
This is the condition marked by \(\dagger\). Finite variance alone
does not guarantee it. \appref{app:dagger-noise-models} verifies it for
centered Gaussian increments and the specified Gaussian scale-mixture
or independent-coordinate models of symmetric Laplace and
Student-\(t\) increments with more than two degrees of freedom.

For the \(\star\) row, growing relative noise can make the response
nearly pure and less sensitive to alternating play. In matching pennies
with cumulative noise and any of the three regularizers, alternating
play earns \(\Theta(\sqrt T)\), and every oblivious trajectory earns
at most \(O(\sqrt T)\) (\appref{app:movement-below-half}). Here the
limiting mean is minimax, so selection has no positive gap. Since
\(\sqrt T=o(T^{1-\rho})\) for \(0<\rho<1/2\), this example rules
out a general \(\Omega(T^{1-\rho})\) movement guarantee in this row.

\section{Conclusion}
\label{sec:conclusion}
Noisy feedback does not generally eliminate the advantage of
anticipating a learner's update rule. For FTRL with decaying learning
rates, active play can earn growing surplus even when summability
conditions bound the gains from every constant strategy.
A profitable limiting mean response supports
a new selection mechanism, yielding \(\Omega(1/\eta_T)\) surplus
without knowing the horizon or observing the realized noise. Even
when that mean is minimax, alternating play can extract
\(\Omega(\sum_{t<T}\eta_t)\) under suitable game geometry,
responsiveness, and stability conditions.

These mechanisms connect surplus to the inverse-rate and cumulative-rate
scales in classical FTRL regret bounds. Both yield \(\Omega(\sqrt T)\)
at \(\eta_t\propto t^{-1/2}\) under their respective conditions.
Noise can make the learner react less to the optimizer's actions,
reducing the gains from alternating play. We therefore examine
whether the optimizer can profit against the learner's average strategy
or induce and profit from changes in that strategy. The appendix
provides detailed proofs, additional examples, and analysis of
the noise models and surplus rates.

\clearpage
\section*{AI Use Statement}
During the preparation of this work, we used GPT-5.6 Sol,
GPT-6 Sol, and GPT-6.1 Sol primarily to polish the writing and
improve the exposition. GPT-6 Sol also assisted with the simulation code.
GPT-Astra and GPT-6 Pro also helped clarify
proof strategies and organize the presentation, building on the
authors' research notes and manuscript. In selection exploitation,
the authors first developed the argument for a known horizon, and
GPT-6 Pro assisted in extending it to an unknown horizon. GPT-6 Pro
also helped analyze a matching-pennies example showing
that alternative surplus cannot improve beyond the \(\sqrt{T}\) scale
in that setting. The authors take full responsibility for the final
content, including all AI-assisted material, and for verifying the
statements, proofs, and citations before submission.

\printbibliography[heading=bibintoc]
\clearpage
\appendix

\section{A Primer in Convex Analysis, Regularized Learning, and Probability}
\label{app:preliminaries}

\noindent\textit{A short roadmap. Think of this primer as a trail with
three stops: relative scores, regularized responses, and random limits.
Start with \cref{app:preliminary-notation,app:preliminary-noisy-ftrl}
for the game and score conventions; continue through
\cref{app:preliminary-convex-analysis,app:preliminary-kkt,app:preliminary-kernel-examples}
to see how convexity turns scores into strategies; finish with
\cref{app:preliminary-probability} for the probability tools.
Already comfortable with a stop? Follow the references straight to
the tool you need. Definitions from the main text are used here
without being repeated.}

\subsection{Game notation and relative scores}
\label{app:preliminary-notation}

We use the game, simplex, best-response set, and value from
\cref{sec:preliminaries}, and the face and span seminorm from
\cref{sec:learner-dynamics-maxmin}. Write \([d]=\{1,\ldots,d\}\),
\(\langle u,v\rangle=u^\top v\), and \(\mathbf1\) for the all-ones
vector in the relevant dimension.

\begin{definition}[Support]
\label{def:preliminary-support}
The support of \(p\in\Delta_d\) is
\(\supp(p):=\{j\in[d]:p_j>0\}\), the set of actions assigned
positive probability.
\end{definition}
\begin{definition}[Max--min and min--max strategies]
\label{def:preliminary-minimax-strategies}
An optimizer strategy \(x^\star\in\Delta_n\) is max--min if
\(\min_{y\in\Delta_m}x^{\star\top}Ay=\val\).
A learner strategy \(y^\star\in\Delta_m\) is min--max if
\(\max_{x\in\Delta_n}x^\top Ay^\star=\val\).
\end{definition}

For \(u\in\mathbb R^m\), \(u_K=(u_j)_{j\in K}\) restricts \(u\) to
the coordinates in \(K\). The relative interior of the face is
\(\operatorname{ri}(\Delta_K)=\{y\in\Delta_K:y_j>0\text{ for every }j\in K\}\).
Mixed best responses are exactly the strategies supported on
\(\br(x)\): a weighted average of the coordinates of \(A^\top x\)
attains their minimum precisely when it assigns mass only to
minimizing coordinates.

The finite-game minimax theorem gives
\(\val=\min_{y\in\Delta_m}\max_{x\in\Delta_n}x^\top Ay\).
Thus \(x^{\star\top}Ay\ge\val\) for all \(y\), while
\(x^\top Ay^\star\le\val\) for all \(x\). For \(K=\br(x^\star)\),
\(y\in\Delta_K\) implies \(x^{\star\top}Ay=\val\).
Common score shifts leave the choice map unchanged, as shown in
\cref{sec:learner-dynamics-maxmin}. The span seminorm discards these
irrelevant shifts.

\begin{remark}
  The span seminorm is a classical tool in Markov decision process theory, including the analysis of value iteration. See \citet{puterman1994markov}. It measures variation across coordinates while disregarding a common additive constant. Motivated by this use in dynamic programming, we adopt it here to measure relative score perturbations. Indeed, the FTRL choice map satisfies \(Q_h(u+c\mathbf1)=Q_h(u)\), so only differences between score coordinates can affect the learner’s response. The span seminorm measures precisely the largest such difference.
\end{remark}

This is the score range after disregarding common shifts.
The following lemma makes this comparison precise.

\begin{lemma}
\label{lem:preliminary-span-centering}
  For every \(u\in\mathbb R^m\),
\(\inf_{c\in\mathbb R}\|u-c\mathbf1\|_\infty =\frac12\|u\|_{\mathrm{sp}}\).
The infimum is attained at
$c^\star=\frac{\max_j u_j+\min_j u_j}{2}.$
Moreover, for every \(p,q\in\Delta_m\),
\[
|\langle u,p-q\rangle|
\le \frac12\|u\|_{\mathrm{sp}}\|p-q\|_1
\le \|u\|_{\mathrm{sp}}.
\]
\end{lemma}
\begin{proof}
  Subtracting a common scalar preserves the difference between the largest and smallest coordinates. For any \(c\in\mathbb R\), these two centered coordinates therefore give
\(\|u\|_{\mathrm{sp}}\le 2\|u-c\mathbf1\|_\infty\).
At \(c=c^\star\), every centered coordinate lies between \(-\|u\|_{\mathrm{sp}}/2\) and \(\|u\|_{\mathrm{sp}}/2\), and the extreme coordinates attain these endpoints. Thus equality holds, proving the centering identity.
Since \(p\) and \(q\) are probability vectors, \(\langle\mathbf1,p-q\rangle=0\). Subtracting \(c^\star\mathbf1\) and applying Hölder’s inequality yields
\[
  |\langle u,p-q\rangle| =|\langle u-c^\star\mathbf1,p-q\rangle| \le \|u-c^\star\mathbf1\|_\infty\|p-q\|_1 =\frac12\|u\|_{\mathrm{sp}}\|p-q\|_1.
\]
The final bound follows from \(\|p-q\|_1\le2\).
\end{proof}

\subsection{Regularized learning and the noisy score}
\label{app:preliminary-noisy-ftrl}

Use the cumulative score, scaled noise, and deterministic score scale
from \cref{sec:preliminaries}. The update \eqref{eq:ftrl} can be
written as \(y_t=Q_h(\eta_tS_t+\Xi_t)\).

Because \(\eta_t>0\), multiplying the objective by
\(1/\eta_t\) gives the equivalent FTRL formulation
\[
  y_t=\argmax_{y\in\Delta_m}
       \left\{\sum_{\tau<t}x_\tau^\top By
                    +\langle Z_t,y\rangle-\frac{h(y)}{\eta_t}\right\}.
\]
The current learning rate rescales the entire
cumulative score.
In particular,
the deterministic part of the score entering \(Q_h\) is
\(\eta_t\sum_{\tau<t}B^\top x_\tau\),
rather than a sum with a
different learning rate attached to each past
payoff.

This convention explains the schedule error in
the movement argument.
Writing \(u_t=\eta_tS_t+\Xi_t\) for the actual choice-map input,
the identity \(S_{t+1}=S_t+B^\top x_t\) yields
\begin{equation*}
  u_{t+1}-u_t
    =(\eta_{t+1}-\eta_t)S_t
        +\eta_{t+1}B^\top x_t+\Xi_{t+1}-\Xi_t.
\end{equation*}
The response can change because past payoffs
are rescaled,
a new
payoff is added,
or the scaled perturbation changes.
These are the
three changes treated separately in
\appref{app:new-movement-pair}.

Under constant max--min play, use \(K\), \(\delta\), and the
reference-score representation \eqref{eq:main-fixed-response}.
For \(t\ge1\),
\(\frac{\spnorm{\Xi_t}}{r_t}=\frac{\spnorm{Z_t}}{t}\).
The normal comparison measures cumulative noise against
cumulative
payoff gaps.
The tangential comparison instead measures
\(\Xi_{t,K}\) against the fixed regularizer on \(\Delta_K\).
At \(t=0\),
the update is \(y_0=Q_h(\Xi_0)\); no division by
\(r_0=0\) is needed.

Fresh and cumulative perturbations are specified in
\cref{sec:noise-examples}.
Conditions on their laws are stated where needed, without imposing independence
on the general process. The optimizer trajectories are deterministic
and fixed independently of the realized noise, as in the main text.
The expected external-regret benchmark and its relation to surplus
are given in \appref{app:maxmin-regret}; the sublinear-regret
condition there is an additional requirement on
the learner's
performance under the prescribed schedule and noise.

\subsection{Convex duality and response stability}
\label{app:preliminary-convex-analysis}

The choice map is the optimizer of a regularized linear objective.
Its convex conjugate records the value of that objective and makes
the dependence of the response on scores explicit.

\begin{definition}[Convex conjugate]
\label{def:preliminary-conjugate-choice}
Let \(h\) satisfy \cref{def:regularizer}, and extend it by
\(+\infty\) outside \(\Delta_m\). Its convex conjugate is
\begin{equation}
  h^*(u):=\sup_{y\in\mathbb R^m}\{\langle u,y\rangle-h(y)\}
         =\max_{y\in\Delta_m}\{\langle u,y\rangle-h(y)\}.
  \label{eq:preliminary-conjugate}
\end{equation}
The maximizer is the choice map from \eqref{eqn:choice-map}.
For this extension, the effective domain is
\(\operatorname{dom}(h):=\{y\in\mathbb R^m:h(y)<+\infty\}=\Delta_m\).
\end{definition}

Continuity of \(h\) on the compact simplex gives a finite
minimum. For every \(u\), the regularized linear objective is also
continuous on that simplex, so the displayed maximum is finite
and attained.

\begin{lemma}[Choice map and conjugate]
\label{lem:preliminary-choice-conjugate}
For any regularizer satisfying \cref{def:regularizer},
\(h^*\) is continuously differentiable on \(\mathbb R^m\),
and
\(\nabla h^*(u)=Q_h(u)\).
Moreover,
\(h(y)+h^*(u)-\langle u,y\rangle\ge0\) for
\(y\in\operatorname{dom}(h)\),
with equality exactly when
\(y=Q_h(u)\).
\end{lemma}

\begin{proof}
Strict convexity makes the maximizer in
\eqref{eq:preliminary-conjugate} unique.
The subdifferential of the
conjugate consists of the maximizing strategies,
so it is the
singleton \(\{Q_h(u)\}\).
A finite convex function with a singleton
subdifferential everywhere is continuously differentiable,
giving
the gradient identity.
The remaining assertion follows by comparing
the objective at \(y\) with its maximum.
Equality holds precisely
at the unique maximizer.
\end{proof}

Thus the response is continuous under the
general assumptions.
A uniform lower bound on curvature gives
a Lipschitz estimate for
how far it moves.
We use the following definitions,
which also apply when
\(h\) is not differentiable at the boundary
of the simplex.

\begin{definition}[Dual norm]
\label{def:preliminary-dual-norm}
For a norm \(\|\cdot\|\) on \(\mathbb R^m\), its dual norm is
\(\|u\|_*:=\sup_{\|z\|\le1}\langle u,z\rangle\).
\end{definition}

\begin{definition}[Strong convexity]
\label{def:preliminary-curvature}
A convex function \(h\) is \(\alpha\)-strongly convex with respect
to \(\|\cdot\|\), for \(\alpha>0\), if for every
\(p,q\in\operatorname{dom}(h)\) and \(s\in[0,1]\),
\(h((1-s)p+sq) \le(1-s)h(p)+sh(q) -\frac\alpha2s(1-s)\|p-q\|^2\).
\end{definition}

\begin{definition}[Smoothness]
\label{def:preliminary-smoothness}
A differentiable convex function \(f:\mathbb R^m\to\mathbb R\)
is \(\beta\)-smooth with respect to the dual
norm if
\[
  f(u)\le f(v)+\langle\nabla f(v),u-v\rangle
                       +\frac\beta2\|u-v\|_*^2
  \qquad(u,v\in\mathbb R^m).
\]
\end{definition}
\begin{definition}[Bregman divergence]
\label{def:preliminary-bregman}
For a convex function \(f\), finite at \(p\) and differentiable at
\(q\), its Bregman divergence is
\(D_f(p,q):=f(p)-f(q)-\langle\nabla f(q),p-q\rangle\).
\end{definition}

Strong convexity supplies a uniform quadratic cost for changing the
strategy. Smoothness controls first-order error in the score potential.
The standard conjugacy relation connects the two
estimates~\citep{orabona2019modern}. The Bregman divergence records
this first-order error.

It is nonnegative by convexity. For \(f=h\), we use this expression
only where \(h\) is finite at \(p\) and differentiable relative to
the simplex at \(q\). The inner product is independent of the
choice of a gradient representative because \(p-q\) sums to zero.
For \(f=h^*\), \cref{lem:preliminary-choice-conjugate} gives
the globally defined score divergence
\(D_{h^*}(u,v)=h^*(u)-h^*(v)-\langle Q_h(v),u-v\rangle\).
This definition avoids evaluating a primal gradient
at a boundary
response,
which is important for steep regularizers.

\begin{proposition}[Quantitative stability under strong convexity]
\label{prop:preliminary-strong-stability}
If \(h\) satisfies \cref{def:regularizer} and is
\(\alpha\)-strongly convex with respect to \(\|\cdot\|\),
then,
for all \(u,v\in\mathbb R^m\),
\begin{align*}
\|Q_h(u)-Q_h(v)\|&\le\frac1\alpha\|u-v\|_*,\\
0\le D_{h^*}(u,v)&\le\frac1{2\alpha}\|u-v\|_*^2.
\end{align*}
In particular,
\(h^*\) is \(1/\alpha\)-smooth.
If strong convexity
is with respect to \(\|\cdot\|_1\),
shift invariance sharpens the
first bound to
\(\|Q_h(u)-Q_h(v)\|_1 \le\frac1{2\alpha}\spnorm{u-v}\).
\end{proposition}

\begin{proof}
The argument compares each response with the
other response in its
own optimization problem.
Adding the comparisons cancels the
regularizer values and relates the strategy distance
to the score
difference.

Let \(p=Q_h(u)\) and \(q=Q_h(v)\).
Strong convexity and
optimality give
\[
  h(q)-\langle u,q\rangle \ge h(p)-\langle u,p\rangle+\frac\alpha2\|p-q\|^2, \qquad h(p)-\langle v,p\rangle \ge h(q)-\langle v,q\rangle+\frac\alpha2\|p-q\|^2.
\]
Adding and using the dual-norm inequality yields
\(\alpha\|p-q\|^2 \le\langle u-v,p-q\rangle \le\|u-v\|_*\|p-q\|\).
If the responses differ,
divide by \(\|p-q\|\); otherwise the
Lipschitz bound holds directly.

To bound the potential error,
integrate its gradient along the score
segment from \(v\) to \(u\):
\[
  D_{h^*}(u,v) =\int_0^1\langle Q_h(v+s(u-v))-Q_h(v),u-v\rangle\,ds \le\int_0^1\frac{s}{\alpha}\|u-v\|_*^2\,ds =\frac1{2\alpha}\|u-v\|_*^2.
\]
This is the asserted smoothness bound.
For the \(\ell_1\) statement,
apply the Lipschitz bound to \(u\) and
\(v+c\mathbf1\),
which have
the same responses as \(u\) and \(v\).
Minimizing their
\(\ell_\infty\) distance over \(c\) gives half the
span.
\end{proof}


\subsection{Separable choice maps: KKT conditions and the boundary}
\label{app:preliminary-kkt}

For the coordinate bounds used in normal
exploitation,
we need to
turn relative scores into individual probabilities.
Separability
reduces this problem to one-dimensional maximization,
with a common
multiplier enforcing the simplex constraint.

Let \(h(y)=\sum_{j=1}^m\theta(y_j)\) satisfy
\cref{def:separable-regularizer}. We first introduce the
scalar optimization problem that determines each coordinate.

\begin{definition}[Kernel conjugate]
\label{def:preliminary-kernel-conjugate}
The scalar convex conjugate of \(\theta\) on \([0,1]\) is
\(\theta^*(z):=\max_{0\le p\le1}\{zp-\theta(p)\}\).
\end{definition}

This scalar conjugate optimizes one coordinate on \([0,1]\);
the simplex conjugate \(h^*\) additionally couples all coordinates
through their sum. Strict convexity of \(\theta\) makes the scalar
maximizer unique. The following map records that maximizing
probability, including when it is an endpoint.

The truncated inverse derivative from \cref{sec:preliminaries}
also has the scalar representation
\(\phi(z)=\argmax_{0\le p\le1}\{zp-\theta(p)\}\).
In particular,
\begin{equation*}
  \phi(z)=
  \begin{cases}
    0,&z\le\theta'(0^+)\text{ and }\theta'(0^+)>-\infty,\\
    (\theta')^{-1}(z),&\theta'(0^+)<z<\theta'(1),\\
    1,&z\ge\theta'(1).
  \end{cases}
\end{equation*}

The scalar conjugate is differentiable and satisfies
\((\theta^*)'(z)=\phi(z)\).
The lower branch is absent for a steep kernel. The upper branch is
needed because each coordinate is at most one.

Fix \(u\in\mathbb R^m\).
Using the sign convention for \(\lambda\)
in the main text,
the Lagrangian for the normalization constraint is
\[
  \mathcal L(y,\lambda)
    =\sum_{j=1}^m\bigl[(u_j+\lambda)y_j-\theta(y_j)\bigr]-\lambda,
    \qquad y\in[0,1]^m.
\]
For fixed \(\lambda\),
each coordinate maximizes its scalar
objective at \(\phi(u_j+\lambda)\).
Choose \(\lambda\) so these
coordinates sum to one.
Continuity and monotonicity of \(\phi\),
together with its limits zero and one,
ensure that such a multiplier
exists.
The resulting vector is feasible,
and maximizes the
Lagrangian over the box.
Since the Lagrangian equals the original
objective on the simplex,
it is the unique response:
\begin{equation}
  Q_h(u)_j=\phi(u_j+\lambda),
  \qquad \sum_{j=1}^m\phi(u_j+\lambda)=1.
  \label{eq:preliminary-kkt-choice}
\end{equation}

For coordinates strictly between zero and one,
and for inactive
coordinates,
the corresponding KKT conditions read
\[
  0<y_j<1\quad\Longrightarrow\quad \theta'(y_j)=u_j+\lambda, \qquad y_j=0\quad\Longrightarrow\quad u_j+\lambda\le\theta'(0^+).
\]
At \(y_j=1\),
the condition is
\(u_j+\lambda\ge\theta'(1)\).
If the response is a vertex
\(e_j\),
one may choose \(\lambda=\theta'(1)-u_j\),
so the
active-coordinate equality used in the main text
holds there as well.
The multiplier itself can be nonunique when
truncation produces a
vertex; all admissible multipliers give the same
unique response.

In particular,
if two coordinates are positive,
subtracting their
KKT equalities eliminates the normalization multiplier:
\(\theta'(y_i)-\theta'(y_j)=u_i-u_j\).
This is the local comparison behind the
pointwise noisy-FTRL bounds
in \appref{sec:passive}.

\paragraph{Steepness and finite boundary values.}
Whether a finite relative score can eliminate an action depends on
the boundary slope, as in \cref{def:steep}.

For a steep kernel, \(\phi(z)>0\) at every
finite score, so \(Q_h(u)\) has full support. Non-steep kernels
have a finite cutoff and can assign exact zero probability to an
action~\citep{mertikopoulos2016learning}. This cutoff is a property
of the current response: later changes in relative scores can return
an action to the support.

Steepness does not require an infinite value
of \(\theta(0)\).
Negative entropy and negative Tsallis regularization both
have an
infinite boundary slope but finite,
continuous boundary values.
They therefore belong to the model in \cref{def:regularizer}
and can be restricted to \(\Delta_K\), even though their full-simplex
responses remain interior at every finite score.
The restricted
map is a separate optimization problem on
the face,
as defined in
\cref{eq:restricted-choice}.

\subsection{Standard kernels and their response maps}
\label{app:preliminary-kernel-examples}

We give the scalar formulas for the
three examples in the main text.
The inverse derivative controls probability tails in
normal
exploitation,
while inverse curvature controls response movement.

\paragraph{Negative entropy.}
For the kernel in \cref{sec:preliminaries}, take the continuous
extension \(\theta(0)=0\).

Its derivatives and boundary slope are
\[
  \theta'(p)=1+\log p,\qquad
  \theta''(p)=\frac1p,\qquad \theta'(0^+)=-\infty.
\]
Scalar maximization gives
\[
  \theta^*(z)=
  \begin{cases}
    \exp(z-1),&z\le1,\\
    z,&z\ge1,
  \end{cases}
  \qquad
  \phi(z)=\min\{1,\exp(z-1)\}.
\]
Normalization in \eqref{eq:preliminary-kkt-choice} then yields
\[
  Q_h(u)_j=\frac{\exp(u_j)}{\sum_{\ell=1}^m\exp(u_\ell)},
  \qquad h^*(u)=\log\!\left(\sum_{\ell=1}^m\exp(u_\ell)\right).
\]
Substituting \(u=\eta_tS_t+\Xi_t\) gives exponential weights on
the noisy scores~\citep{auer1995gambling}.
Its inverse curvature is
\(1/\theta''(p)=p\),
so responsiveness decreases as a coordinate's
probability approaches zero.

\paragraph{Euclidean regularization.}

Its derivatives and boundary slope are
\[
  \theta'(p)=p,\qquad \theta''(p)=1,\qquad \theta'(0^+)=0.
\]
The scalar conjugate and inverse derivative are
\[
  \theta^*(z)=
  \begin{cases}
    0,&z\le0,\\
    z^2/2,&0\le z\le1,\\
    z-1/2,&z\ge1,
  \end{cases}
  \qquad \phi(z)=\min\{1,\max\{0,z\}\}.
\]
Completing the square in the simplex objective
gives
\[
  Q_h(u)=\argmin_{y\in\Delta_m}\frac12\|y-u\|_2^2,
  \qquad Q_h(u)_j=\max\{0,u_j+\lambda\},
\]
where \(\lambda\) normalizes the coordinates.
This is Euclidean
projection onto the simplex~\citep{mertikopoulos2016learning}.
The kernel is non-steep,
and its inverse curvature is one on
every
positive coordinate.
Inactive coordinates are assigned zero
susceptibility in the movement analysis.

\paragraph{Negative Tsallis regularization.}
For the kernel in \cref{sec:preliminaries}, \(0<q<1\) and
\(\theta(0)=0\).

Its derivatives and boundary slope are
\[
  \theta'(p)=\frac{q p^{q-1}-1}{q-1},\qquad
  \theta''(p)=q p^{q-2},\qquad \theta'(0^+)=-\infty.
\]
The scalar formulas are
\[
\theta^*(z)=
  \begin{cases}
    \left(\dfrac{1+(q-1)z}{q}\right)^{q/(q-1)},&z\le1,\\
    z,&z\ge1,
  \end{cases}
  \qquad
  \phi(z)=
  \begin{cases}
    \left(\dfrac{1+(q-1)z}{q}\right)^{1/(q-1)},&z\le1,\\
    1,&z\ge1.
  \end{cases}
\]
The base of each power is positive
on the stated branch.
Thus \(Q_h(u)_j=\phi(u_j+\lambda)\),
with \(\lambda\) chosen
to enforce normalization.
This kernel gives the standard Tsallis
family~\citep{abernethy2015fighting}; it converges to negative entropy
as \(q\uparrow1\).
Its inverse curvature is
\(1/\theta''(p)=p^{2-q}/q\),
which also vanishes near zero.

All three kernels are finite, continuous, and strictly convex
on \([0,1]\), so their regularizers satisfy \cref{def:regularizer}.
Their inverse-curvature constants in
\cref{ass:new-movement-regularity} are \(C_\theta=1\), \(1\), and
\(1/q\), respectively. Their common kernel across coordinates
also implies that the regularizer minimizer on \(\Delta_K\) is
uniform on \(K\), as proved in
\cref{cor:uniform-separable-selection}.

\subsection{Probability laws,
mean responses,
and couplings}
\label{app:preliminary-probability}

The response analysis distinguishes convergence of a realized
strategy from convergence of its distribution.

\begin{definition}[Convergence in probability]
\label{def:preliminary-convergence-probability}
A sequence of random vectors \(U_t\) converges in probability to
\(U\), written \(U_t\pto U\), if for every \(\varepsilon>0\),
\(\Pr(\|U_t-U\|>\varepsilon)\longrightarrow0\).
\end{definition}

\begin{definition}[Convergence in distribution]
\label{def:preliminary-convergence-distribution}
A sequence of random vectors \(U_t\) converges in distribution to
\(U\), written \(U_t\dto U\), if
\(\mathbb E[f(U_t)]\to\mathbb E[f(U)]\) for every bounded
continuous function \(f\) on the ambient Euclidean space.
\end{definition}

\begin{definition}[Equality in distribution]
\label{def:preliminary-equality-distribution}
Random vectors \(U\) and \(V\) are equal in distribution, written
\(U\stackrel d=V\), if they have the same probability law.
\end{definition}

In finite dimensions, the particular norm used to define convergence in
probability does not affect the conclusion.

Responses take values in a compact simplex,
so their convergence
also controls the means without moment assumptions
on the raw noise.
For example,
if \(Y_t,Y\in\Delta_m\) and \(Y_t\pto Y\),
then
\(\mathbb E\|Y_t-Y\|_1 \le\varepsilon+2\Pr(\|Y_t-Y\|_1>\varepsilon)\)
for every \(\varepsilon>0\).
Taking the limit superior bounds the
expectation by \(\varepsilon\),
and then letting
\(\varepsilon\downarrow0\) gives
\(\mathbb E\|Y_t-Y\|_1\to0\).
If instead \(Y_t\dto Y\),
boundedness and continuity of the coordinate functions
on the
simplex give \(\mathbb E[Y_t]\to\mathbb E[Y]\).
In this case
the limiting response can remain random.
These are the two ways
the response results supply a limiting mean
for selection
exploitation.

The symmetry condition used to identify some of these means is
exchangeability.

\begin{definition}[Exchangeability]
\label{def:preliminary-exchangeability}
A random vector \(U_K\in\mathbb R^K\) has an exchangeable law if
for every permutation \(\sigma\) of \(K\),
\((U_{\sigma(j)})_{j\in K}\stackrel d=U_K\).
\end{definition}

For the symmetric separable regularizers above, permuting the scores
permutes the restricted response. An exchangeable score law therefore
gives an exchangeable response law and a uniform mean on \(K\).
Centering the noise alone does not imply this permutation symmetry,
as discussed in \appref{app:cumulative-finite-variance}.

To compare consecutive noise laws in movement exploitation, we use
couplings.

\begin{definition}[Coupling]
\label{def:preliminary-coupling}
A coupling of two laws is a joint law with those marginals.
Here \((\widehat\Xi_s^0,\widehat\Xi_s^1)\) couples adjacent noise
laws when \(\widehat\Xi_s^0\stackrel d=\Xi_{2s}\) and
\(\widehat\Xi_s^1\stackrel d=\Xi_{2s+1}\).
\end{definition}

The coupling need not reproduce temporal dependence.
For deterministic optimizer play, \(S_t\) is deterministic and
expected round payoffs depend only on the noise marginals.
For bounded measurable \(f_0,f_1\), every coupling satisfies
\(\mathbb E[f_0(\widehat\Xi_s^0)-f_1(\widehat\Xi_s^1)]
=\mathbb E[f_0(\Xi_{2s})]-\mathbb E[f_1(\Xi_{2s+1})]\).
Thus the noise-error bound in \appref{app:new-movement-pair}
may depend on the coupling, while the expected payoff difference does not.

\section{Proofs for Learner Dynamics under Max--Min Play}
\label{app:normal}

This appendix provides the proofs of
\cref{prop:must-be-average-maxmin,prop:main-localization,prop:main-face-reduction,thm:main-noise-regimes}
from \cref{sec:learner-dynamics-maxmin}.
We use the notation and noise regimes defined in the main text.
\appref{app:rps-noise-regimes} records the experimental details behind
\cref{fig:rps-three-triangles} and the accompanying trajectory plots.

\subsection{Regret decomposition and average max--min play}
\label{app:maxmin-regret}

The identity in \cref{prop:must-be-average-maxmin} follows by writing
the fixed-comparator payoff in terms of \(\bar x_T\).
We include the calculation and then prove the asymptotic consequence
used immediately after that proposition.

\begin{proposition}[Restatement of \cref{prop:must-be-average-maxmin}]
\label{prop:appendix-regret-decomposition}
For every optimizer trajectory \(x_{<T}\),
\[
  \Surplus(x_{<T})
    =\operatorname{Reg}_L(x_{<T})
       -T\left(\val-\min_{y\in\Delta_m}\bar x_T^\top Ay\right).
\]
\end{proposition}

\begin{proof}
For each fixed learner strategy \(y\), linearity gives
\[
  \sum_{t=0}^{T-1}x_t^\top Ay=T\bar x_T^\top Ay.
\]
Thus the minimum fixed-comparator payoff is
\(T\min_{y\in\Delta_m}\bar x_T^\top Ay\).
Substituting the regret identity \eqref{eq:main-learner-regret} into
\(\Surplus(x_{<T})=\Rew(x_{<T})-T\val\) gives
\begin{align}
  \Surplus(x_{<T})
    &=\operatorname{Reg}_L(x_{<T})
       +T\min_{y\in\Delta_m}\bar x_T^\top Ay-T\val
       \notag\\
    &=\operatorname{Reg}_L(x_{<T})
       -T\left(\val-\min_{y\in\Delta_m}\bar x_T^\top Ay\right).
       \label{eq:regret-surplus-decomposition:b}
\end{align}
The quantity in parentheses is nonnegative by the definition of
\(\val\), and it vanishes exactly when \(\bar x_T\) is max--min.
\end{proof}

\begin{corollary}[Sublinear regret requires asymptotically max--min average play]
\label{cor:average-maxmin}
Suppose the learner has sublinear expected external regret
\(\operatorname{Reg}_L(x_{<T})=o(T)\).
If \(\Surplus(x_{<T})\ge-o(T)\), then
\begin{equation*}
  \val-\min_{y\in\Delta_m}\bar x_T^\top Ay=o(1).
\end{equation*}
\end{corollary}

\begin{proof}
Rearrange \eqref{eq:regret-surplus-decomposition:b} and divide by \(T\):
\[
  0\le\val-\min_{y\in\Delta_m}\bar x_T^\top Ay
    =\frac{\operatorname{Reg}_L(x_{<T})}{T}
       -\frac{\Surplus(x_{<T})}{T}
    \le o(1).
\]
The last inequality uses the assumed regret bound and lower bound
on surplus, proving the claim.
\end{proof}

\subsection{Proof of localization to the best-response face}
\label{app:normal-localization}

For this and the remaining response proofs, fix the max--min strategy
\(x^\star\) from \cref{sec:learner-dynamics-maxmin} and retain
\(K=\br(x^\star)\) and
\(\delta=A^\top x^\star-\val\mathbf1\).
The reference response \eqref{eq:main-fixed-response} is
\begin{equation}
  y_t=Q_h(-r_t\delta+\Xi_t).
  \label{eq:fixed-response}
\end{equation}
The estimates below use the following consequences of
\cref{ass:step-sizes,def:step-size-noise-model}:
\begin{equation}
  r_t\longrightarrow\infty,
  \qquad
  \frac{\spnorm{\Xi_t}}{r_t}
    =\frac{\spnorm{Z_t}}{t}\pto0\qquad(t\ge1).
  \label{eq:normal-condition}
\end{equation}
In particular, the noise need not be bounded on the scale seen by
the choice map. Only its ratio to \(r_t\) must vanish.

\begin{lemma}[Restatement of \cref{prop:main-localization}]
\label{lem:normal}
Let \(h\) satisfy \cref{def:regularizer}. Then
\begin{equation*}
  \sum_{j\notin K}y_{t,j}\pto0.
\end{equation*}
\end{lemma}

\begin{proof}
If \(K=[m]\), the sum is identically zero.
Otherwise, let
\(\delta_{\min}=\min_{j\notin K}\delta_j>0\) and
\(h_{\min}=\min_{\Delta_m}h\), which is finite by continuity
and compactness.
We bound the off-face mass by comparing \(y_t\) with a strategy
whose payoff loss is arbitrarily small and whose regularization
cost is fixed as \(t\) grows.

Fix \(i\in K\), put \(p=m^{-1}\mathbf1\), and, for \(0<\alpha<1\),
take
\[
  p_\alpha=(1-\alpha)e_i+\alpha p,
  \qquad
  \langle\delta,p_\alpha\rangle
    =\alpha\langle\delta,p\rangle.
\]
This comparator is interior and has finite regularization cost.
Using an interior comparator also gives the extension in
\cref{rem:localization-weaker-regularizer}.
Optimality of \(y_t\) in \eqref{eq:fixed-response} implies
\[
  r_t\langle\delta,y_t\rangle
    \le r_t\alpha\langle\delta,p\rangle
       +\langle\Xi_t,y_t-p_\alpha\rangle
       +h(p_\alpha)-h(y_t).
\]
The noise term is at most \(\spnorm{\Xi_t}\), since it is the
difference of two averages of the same coordinates.
Moreover, \(\delta\) is zero on \(K\) and at least
\(\delta_{\min}\) outside \(K\).
Dividing by \(r_t>0\) therefore gives
\[
  \delta_{\min}\sum_{j\notin K}y_{t,j}
    \le \alpha\langle\delta,p\rangle
       +\frac{\spnorm{\Xi_t}}{r_t}
       +\frac{h(p_\alpha)-h_{\min}}{r_t}.
\]
The three terms are, respectively, the comparator's payoff loss,
the relative noise, and the fixed regularization cost divided by
the growing score scale.

Given \(\varepsilon>0\), first choose \(\alpha\) so that the first
term is at most \(\varepsilon\delta_{\min}/3\).
With that comparator fixed, take \(t\) large enough that the last
term is at most \(\varepsilon\delta_{\min}/3\).
Then \eqref{eq:normal-condition} yields
\[
  \Pr\!\left(\sum_{j\notin K}y_{t,j}>\varepsilon\right)
    \le
    \Pr\!\left(\frac{\spnorm{\Xi_t}}{r_t}
                  >\frac{\varepsilon\delta_{\min}}3\right)
    \longrightarrow0.
\]
This proves localization. The mass is bounded by one, so its
expectation also tends to zero.
\end{proof}

Equivalently, localization gives convergence in distance to the face:
\[
  \dist_1(y_t,\Delta_K)=2\sum_{j\notin K}y_{t,j}.
\]
Indeed, any point on the face removes all mass outside \(K\) and
adds the same amount within \(K\), giving the lower bound.
Moving that mass to any one action in \(K\) attains it.

\begin{remark}[Localization under weaker regularizer conditions]
\label{rem:localization-weaker-regularizer}
The proof of \cref{lem:normal} also applies when
\(h:\Delta_m\to\mathbb R\cup\{+\infty\}\) is proper and lower
semicontinuous, strictly convex on its effective domain
\(\{y\in\Delta_m:h(y)<+\infty\}\), and finite on
\(\operatorname{ri}(\Delta_m)=\{y\in\Delta_m:y_j>0\ \forall j\}\).
Compactness and lower semicontinuity still give a finite
\(h_{\min}\), and each interior comparator \(p_\alpha\) has finite
regularization cost. These are the regularizer properties used in
the comparison, so this localization argument does not require
finite or continuous boundary values. The remaining results use
the regularizer class in \cref{def:regularizer}.
\end{remark}

\subsection{Proof of reduction to the best-response face}
\label{sec:regimes}
\label{app:face-reduction}

We prove \cref{prop:main-face-reduction} using the restricted map
\(Q_{h,K}\) defined in \eqref{eq:restricted-choice}.
Two properties are needed: continuity for the finite-score limits,
and a comparison between the full and restricted responses that
remains valid when relative scores grow.
The continuity argument is given first.

\begin{remark}[Regularity on the face]
\label{rem:finite-regularizer-examples}
\Cref{def:regularizer} supplies finite, continuous, strictly convex
values on every face, including its boundary.
Negative entropy, Euclidean regularization, and negative Tsallis
regularization with \(q\in(0,1)\) all satisfy this condition.
It is also part of the regularizer definitions used by
\citet{mertikopoulos2016learning},
\citealt{mertikopoulos2019learning}, and \citet{su2026nocoin}.
\end{remark}

\begin{lemma}
\label{lem:continuity}
\(Q_{h,K}\) is continuous and satisfies
\(Q_{h,K}(u_K+c\mathbf1)=Q_{h,K}(u_K)\) for every
\(u_K\in\mathbb R^K\) and \(c\in\mathbb R\).
\end{lemma}

\begin{proof}
A common shift adds \(c\) to every feasible objective, proving
shift invariance.
For continuity, let \(u_K^{(\ell)}\to u_K\) and write
\(q_\ell=Q_{h,K}(u_K^{(\ell)})\).
Compactness of \(\Delta_K\) gives a convergent further subsequence
of every subsequence of \((q_\ell)\).
Let \(\bar q\) be the limit of such a further subsequence.

We identify \(\bar q\) by passing the optimality inequality to the
limit. For every \(z\in\Delta_K\),
\[
  \langle u_K^{(\ell)},(q_\ell)_K\rangle-h(q_\ell)
    \ge \langle u_K^{(\ell)},z_K\rangle-h(z).
\]
Continuity of \(h\) gives, along the chosen subsequence,
\[
  \langle u_K,\bar q_K\rangle-h(\bar q)
    \ge \langle u_K,z_K\rangle-h(z).
\]
Thus \(\bar q\) maximizes the limiting objective, whose maximizer
is unique by strict convexity. Hence \(\bar q=Q_{h,K}(u_K)\).
If the full sequence failed to converge to this point, a subsequence
at a fixed positive distance would have a further convergent
subsequence with a different limit. This contradiction proves
continuity.
\end{proof}

For face reduction, continuity alone is insufficient: the scores
\(\Xi_{t,K}\) need not converge or remain bounded.
The next proof instead controls the objective loss incurred by
moving the off-face mass onto \(\Delta_K\).
A strict-convexity bound independent of the scores then controls
the distance between the two optimizers.

\begin{lemma}[Restatement of \cref{prop:main-face-reduction}]
\label{lem:reduction}
\begin{equation}
  \|y_t-Q_{h,K}(\Xi_{t,K})\|_1\pto0.
  \label{eq:reduction}
\end{equation}
\end{lemma}

\begin{proof}
When \(K=[m]\), the two responses coincide.
Otherwise, set
\[
  u_t=-r_t\delta+\Xi_t,
  \qquad F_t(y)=\langle u_t,y\rangle-h(y),
  \qquad q_t=Q_{h,K}(\Xi_{t,K}).
\]
Because \(\delta\) vanishes on \(K\), the vectors \(y_t\) and
\(q_t\) maximize the same objective \(F_t\) on \(\Delta_m\) and
\(\Delta_K\), respectively.

\paragraph{Repairing the off-face mass.}
To compare these two feasible sets, choose
\(i_t\in\argmax_{i\in K}\Xi_{t,i}\), breaking ties by index,
and move all off-face mass to \(i_t\).
Writing \(M_t=\sum_{j\notin K}y_{t,j}\), the repaired vector is
\[
  \widetilde y_t
    =y_t-\sum_{j\notin K}y_{t,j}e_j+M_te_{i_t}
    \in\Delta_K,
  \qquad
  \|\widetilde y_t-y_t\|_1=2M_t\pto0,
\]
where the convergence follows from \cref{lem:normal}.
We must also control the linear score change, since a small mass
could otherwise carry a large score.
Let \(\delta_{\min}=\min_{j\notin K}\delta_j\) and consider
\[
  E_t=\{\spnorm{\Xi_t}\le r_t\delta_{\min}/2\}.
\]
By \eqref{eq:normal-condition}, \(\Pr(E_t)\to1\).
On this event, each \(j\notin K\) satisfies
\[
  u_{t,i_t}-u_{t,j}
    \ge r_t\delta_{\min}-\spnorm{\Xi_t}
    \ge r_t\delta_{\min}/2>0.
\]
Thus the repair does not decrease the linear score:
\(\langle u_t,\widetilde y_t-y_t\rangle\ge0\) on \(E_t\).
The factor \(1/2\) leaves a positive score advantage after
accounting for the noise difference.

\paragraph{Controlling the objective loss.}
Optimality of \(y_t\) on the full simplex and of \(q_t\) on the
face gives, on \(E_t\),
\[
  0\le F_t(y_t)-F_t(q_t)
    \le F_t(y_t)-F_t(\widetilde y_t)
    \le |h(\widetilde y_t)-h(y_t)|.
\]
The final bound uses the nonnegative linear score change.
Uniform continuity of \(h\) on \(\Delta_m\), together with
\(\|\widetilde y_t-y_t\|_1\pto0\), makes this bound vanish
in probability. Since \(\Pr(E_t^c)\to0\), we obtain
\(F_t(y_t)-F_t(q_t)\pto0\).

\paragraph{Converting objective loss to distance.}
To make this last step uniform in \(u_t\), fix
\(0<\varepsilon\le2\) and define
\[
  c_\varepsilon
    :=\min_{\substack{v,w\in\Delta_m\\\|v-w\|_1\ge\varepsilon}}
       \left\{\frac{h(v)+h(w)}2-h\!\left(\frac{v+w}2\right)\right\}.
\]
The minimum is attained on a compact set, and strict convexity
makes its value positive. It measures the regularizer's midpoint
gap between strategies at distance at least \(\varepsilon\).
If \(\|y_t-q_t\|_1\ge\varepsilon\), the linear terms cancel at
the midpoint, yielding
\[
  F_t\!\left(\frac{y_t+q_t}2\right)
    \ge\frac{F_t(y_t)+F_t(q_t)}2+c_\varepsilon.
\]
The midpoint is feasible for the full problem, so optimality of
\(y_t\) implies \(F_t(y_t)-F_t(q_t)\ge2c_\varepsilon\).
Consequently,
\[
  \Pr(\|y_t-q_t\|_1\ge\varepsilon)
    \le\Pr(F_t(y_t)-F_t(q_t)\ge2c_\varepsilon)
    \longrightarrow0.
\]
This proves \eqref{eq:reduction}. Since the distance is bounded
by two, its expectation tends to zero as well.
\end{proof}

The remaining three subsections prove the respective cases of
\cref{thm:main-noise-regimes}. Each applies the noise condition
from \cref{sec:preliminaries} to \(Q_{h,K}\), then uses
\cref{lem:reduction} to transfer the limit to \(y_t\).

\Needspace{14\baselineskip}
\subsection{Proof of the \texorpdfstring{\regime{I}}{Regime I} limit}
\label{app:vanishing-noise}

\begin{proposition}[Vanishing-noise case of \cref{thm:main-noise-regimes}]
\label{prop:vanishing}
Suppose \regime{I} holds on \(K\). Then
\begin{equation*}
  y_t\pto y_K=\argmin_{y\in\Delta_K}h(y),
  \qquad \mathbb E\|y_t-y_K\|_1\longrightarrow0.
\end{equation*}
\end{proposition}

\begin{proof}
The regime assumption controls score differences, so first remove
one common coordinate. Fix \(i\in K\) and set
\(v_{t,j}=\Xi_{t,j}-\Xi_{t,i}\) for \(j\in K\).
Then \(\|v_t\|_\infty\le\spnorm{\Xi_{t,K}}\pto0\).
Shift invariance and continuity from \cref{lem:continuity} give
\[
  Q_{h,K}(\Xi_{t,K})=Q_{h,K}(v_t)
    \pto Q_{h,K}(0)=y_K.
\]
The zero-score maximizer minimizes \(h\) on \(\Delta_K\), which
identifies the limit.
By \cref{lem:reduction} and the triangle inequality,
\[
  \|y_t-y_K\|_1
    \le\|y_t-Q_{h,K}(\Xi_{t,K})\|_1
       +\|Q_{h,K}(\Xi_{t,K})-y_K\|_1\pto0.
\]
The distance is bounded by two, so its expectation also vanishes.
\end{proof}

\begin{corollary}[Separable-regularizer case of \cref{thm:main-noise-regimes}]
\label{cor:uniform-separable-selection}
If \(h\) is separable as in \cref{def:separable-regularizer}, then
\[
  y_K=\frac1{|K|}\sum_{j\in K}e_j.
\]
\end{corollary}

\begin{proof}
Every permutation of the coordinates in \(K\) preserves
\(\Delta_K\) and \(h\), because all coordinates use the same
kernel \(\theta\). It therefore maps \(y_K\) to another minimizer.
Strict convexity makes the minimizer unique, so \(y_K\) is invariant
under every such permutation. Its coordinates on \(K\) must be
equal and sum to one; its coordinates outside \(K\) are zero.
This gives the displayed formula.
\end{proof}

\Needspace{14\baselineskip}
\subsection{Proof of the \texorpdfstring{\regime{II}}{Regime II} limit}
\label{app:finite-noise}

\begin{proposition}[Finite-noise case of \cref{thm:main-noise-regimes}]
\label{prop:finite}
Suppose \regime{II} holds on \(K\) with limit \(\Xi_{\infty,K}\).
Then
\begin{equation}
  y_t\dto Y_\infty:=Q_{h,K}(\Xi_{\infty,K}).
  \label{eq:finite-limit}
\end{equation}
Moreover,
\begin{equation}
  \mathbb E[y_t]\longrightarrow\mathbb E[Y_\infty],
  \qquad \Cov(y_t)\longrightarrow\Cov(Y_\infty).
  \label{eq:moment-limit}
\end{equation}
\end{proposition}

\begin{proof}
Let \(c_t\) be the scalar shifts in \regime{II}.
Shift invariance and continuity of \(Q_{h,K}\) give, by the
continuous mapping theorem,
\[
  Q_{h,K}(\Xi_{t,K})
    =Q_{h,K}(\Xi_{t,K}-c_t\mathbf1)
    \dto Q_{h,K}(\Xi_{\infty,K})=Y_\infty.
\]
Face reduction gives \(y_t-Q_{h,K}(\Xi_{t,K})\pto0\).
Slutsky's theorem therefore yields \eqref{eq:finite-limit}; no
independence between these two quantities is required.

To obtain the moments, the functions \(y\mapsto y_i\) and
\(y\mapsto y_i y_j\) are bounded and continuous on the simplex.
Their expectations consequently converge:
\[
  \mathbb E[y_{t,i}]\longrightarrow\mathbb E[Y_{\infty,i}],
  \qquad
  \mathbb E[y_{t,i}y_{t,j}]
    \longrightarrow\mathbb E[Y_{\infty,i}Y_{\infty,j}].
\]
Substituting these limits into
\(\Cov(y_t)_{ij}=\mathbb E[y_{t,i}y_{t,j}]
-\mathbb E[y_{t,i}]\mathbb E[y_{t,j}]\)
proves \eqref{eq:moment-limit}.
\end{proof}

\Needspace{14\baselineskip}
\subsection{Proof of the \texorpdfstring{\regime{III}}{Regime III} limit}
\label{app:growing-noise}

\begin{theorem}[Growing-noise case of \cref{thm:main-noise-regimes}]
\label{thm:diverging}
Suppose \regime{III} holds on \(K\), with scale \(a_t\), shifts
\(c_t\), and limit \(\Xi_{\infty,K}\). Define
\[
  J_\infty=\argmax_{j\in K}\Xi_{\infty,j},
  \qquad J_t=\min\operatorname*{arg\,max}_{j\in K}\Xi_{t,j}.
\]
Then
\[
  y_t\dto e_{J_\infty},
  \qquad \mathbb E\|y_t-e_{J_t}\|_1\longrightarrow0.
\]
In particular,
\(\mathbb E[y_t]\longrightarrow
\sum_{j\in K}\Pr(J_\infty=j)e_j\).
\end{theorem}

\begin{proof}
Use the normalized scores from \regime{III} and write the restricted
response as
\[
  \widehat\Xi_{t,K}=\frac{\Xi_{t,K}-c_t\mathbf1}{a_t},
  \qquad
  q_t=Q_{h,K}(\Xi_{t,K})=Q_{h,K}(a_t\widehat\Xi_{t,K}).
\]
Centering and dividing by a positive scalar preserve the maximizing
indices, so \(J_t\) also maximizes \(\widehat\Xi_{t,K}\).
If \(|K|=1\), the restricted response is always the sole vertex
of \(\Delta_K\), and all conclusions follow from
\cref{lem:reduction}. Hence assume \(|K|\ge2\).

\paragraph{Bounding the mass away from the current winner.}
We compare \(q_t\) with \(e_{J_t}\). Any regularization saving is
bounded by
\(H_h=\max_{\Delta_m}h-\min_{\Delta_m}h<\infty\).
Optimality of \(q_t\) in the restricted objective gives
\begin{equation*}
  0\le\widehat\Xi_{t,J_t}
         -\sum_{j\in K}\widehat\Xi_{t,j}q_{t,j}
    \le\frac{h(e_{J_t})-h(q_t)}{a_t}
    \le\frac{H_h}{a_t}.
\end{equation*}
Thus the normalized score loss relative to the winning vertex is
at most \(H_h/a_t\).
To translate score loss into misplaced mass, let \(\Gamma_t\) be
the difference between the largest and second-largest coordinates
of \(\widehat\Xi_{t,K}\), with \(\Gamma_t=0\) when the maximum
is tied. Every unit of mass away from \(J_t\) loses at least
\(\Gamma_t\), so
\[
  \widehat\Xi_{t,J_t}
       -\sum_{j\in K}\widehat\Xi_{t,j}q_{t,j}
    =\sum_{j\in K}q_{t,j}
       (\widehat\Xi_{t,J_t}-\widehat\Xi_{t,j})
    \ge\Gamma_t(1-q_{t,J_t}).
\]
Since \(\|q_t-e_{J_t}\|_1=2(1-q_{t,J_t})\), for every fixed
\(\gamma>0\) we have, on \(\{\Gamma_t>\gamma\}\),
\begin{equation}
  \|q_t-e_{J_t}\|_1\le\frac{2H_h}{a_t\gamma}.
  \label{eq:current-winner-bound}
\end{equation}
This controls the response whenever the normalized top scores
are separated. It remains to bound the probability of a near tie.

\paragraph{Excluding near ties in probability.}
Let \(\Gamma_\infty\) be the corresponding top-score gap of
\(\Xi_{\infty,K}\).
The largest and second-largest coordinates are continuous functions
of a finite score vector, so
\(\Gamma_t\dto\Gamma_\infty\).
The unique-maximum assumption gives
\(\Gamma_\infty>0\) almost surely.
For fixed \(\varepsilon,\gamma>0\), take \(t\) large enough that
\(2H_h/(a_t\gamma)<\varepsilon\).
Then \eqref{eq:current-winner-bound} and the Portmanteau theorem yield
\[
  \limsup_{t\to\infty}
    \Pr(\|q_t-e_{J_t}\|_1>\varepsilon)
    \le\limsup_{t\to\infty}\Pr(\Gamma_t\le\gamma)
    \le\Pr(\Gamma_\infty\le\gamma).
\]
Letting \(\gamma\downarrow0\) makes the right-hand side vanish.
Therefore \(\|q_t-e_{J_t}\|_1\pto0\), and boundedness by two
also gives convergence in expectation.

\paragraph{Transferring the response and mean limits.}
Face reduction and the triangle inequality give
\[
  \|y_t-e_{J_t}\|_1
    \le\|y_t-q_t\|_1+\|q_t-e_{J_t}\|_1\pto0.
\]
Both terms on the right also vanish in expectation, proving the
claimed approximation by the current winning vertex.
The map from a score vector to its winning vertex is continuous
where the largest coordinate is unique.
The continuous mapping theorem thus gives
\(e_{J_t}\dto e_{J_\infty}\).
Together with the vanishing response distance, Slutsky's theorem
proves \(y_t\dto e_{J_\infty}\).
Finally, boundedness and continuity of each coordinate on the
simplex give
\[
  \mathbb E[y_{t,j}]\longrightarrow
    \mathbb E[(e_{J_\infty})_j]=\Pr(J_\infty=j)
    \qquad(j\in K),
\]
with limit zero outside \(K\).
\end{proof}

This completes the proof of \cref{thm:main-noise-regimes}.
The growing-noise argument controls the response at a late round;
it does not assert convergence along a sample path to a fixed
vertex, since \(J_t\) may keep changing.

\subsection{Details of the numerical illustration}
\label{app:rps-noise-regimes}

This subsection gives the setup and noise-condition checks for
\cref{fig:rps-three-triangles}. The supplementary
\cref{fig:rps-noise-regimes} includes sample trajectories and
pointwise empirical percentiles.

\paragraph{Experimental setting.}
The optimizer plays \(x^\star=(1/3,1/3,1/3)\) at every round in
rock--paper--scissors, with payoff matrix
\[
  A=\begin{pmatrix}0&-1&1\\1&0&-1\\-1&1&0\end{pmatrix}.
\]
Here \(A^\top x^\star=0\), \(\val=0\), \(K=[3]\), and
\(S_t=0\). The learner uses negative entropy
\(h(y)=\sum_{j=1}^3y_j\log y_j\), with \(0\log0=0\), and
\[
  \eta_t=(t+1)^{-\rho},\qquad
  Z_t=\sum_{s=0}^{t-1}z_s,\qquad
  z_s\overset{\mathrm{iid}}{\sim}\mathcal N(0,I_3),\qquad Z_0=0.
\]
Thus \(y_t=\operatorname{softmax}(\eta_tZ_t)\), starting from
\(\bar y=(1/3,1/3,1/3)\).
We simulate \(5000\) independent runs through \(T=10^6\) for
\(\rho\in\{0.75,0.50,0.25\}\).
Within each run, the three schedules use the same Gaussian
increments, so their learning rates are the only changing parameter.
The simulations were run on a Mac equipped with an Apple M1 chip.

\paragraph{Verification of the hypotheses.}
Since \(Z_t\sim\mathcal N(0,tI_3)\), for \(t\ge1\) we have
the exact marginal representation
\[
  \Xi_t\overset d=a_tG,
  \qquad G\sim\mathcal N(0,I_3),
  \qquad a_t=\frac{\sqrt t}{(t+1)^\rho}\sim t^{1/2-\rho}.
\]
For every \(0<\rho<1\), this gives
\[
  r_t=\frac{t}{(t+1)^\rho}\longrightarrow\infty,
  \qquad
  \frac{\spnorm{\Xi_t}}{r_t}
    \overset d=\frac{\spnorm{G}}{\sqrt t}\pto0.
\]
Hence \eqref{eq:normal-condition} holds for all three schedules.
The relative-noise conditions and response limits are as follows.
\begin{itemize}
  \item For \(\rho=0.75\), \(a_t\to0\), so
  \(\spnorm{\Xi_t}\overset d=a_t\spnorm{G}\pto0\).
  This is \regime{I};
  \cref{prop:vanishing,cor:uniform-separable-selection} give
  \(y_t\pto\bar y\).

  \item For \(\rho=0.50\),
  \(a_t=\sqrt{t/(t+1)}\to1\), so \(\Xi_t\dto G\).
  This is \regime{II} with \(c_t=0\), and
  \cref{prop:finite} gives
  \(y_t\dto\operatorname{softmax}(G)\).

  \item For \(\rho=0.25\), \(a_t\to\infty\) and
  \(\Xi_t/a_t\overset d=G\).
  For each \(i\ne j\), \(G_i-G_j\sim\mathcal N(0,2)\), so
  a tie has probability zero. A finite union over pairs shows that
  the largest coordinate is almost surely unique.
  Thus \regime{III} holds with \(c_t=0\), and
  \cref{thm:diverging} gives
  \(y_t\dto e_{\argmax_jG_j}\).
\end{itemize}
The threshold \(\rho=1/2\) follows from the exponent in
\(a_t\sim t^{1/2-\rho}\).
At \(T=10^6\), the three noise scales are approximately
\(0.03162\), \(1\), and \(31.62277\).
These checks use marginal distributions only, so the dependence
across rounds created by cumulative noise is permitted.

\begin{figure}[!t]
    \centering
    \includegraphics[width=\textwidth]{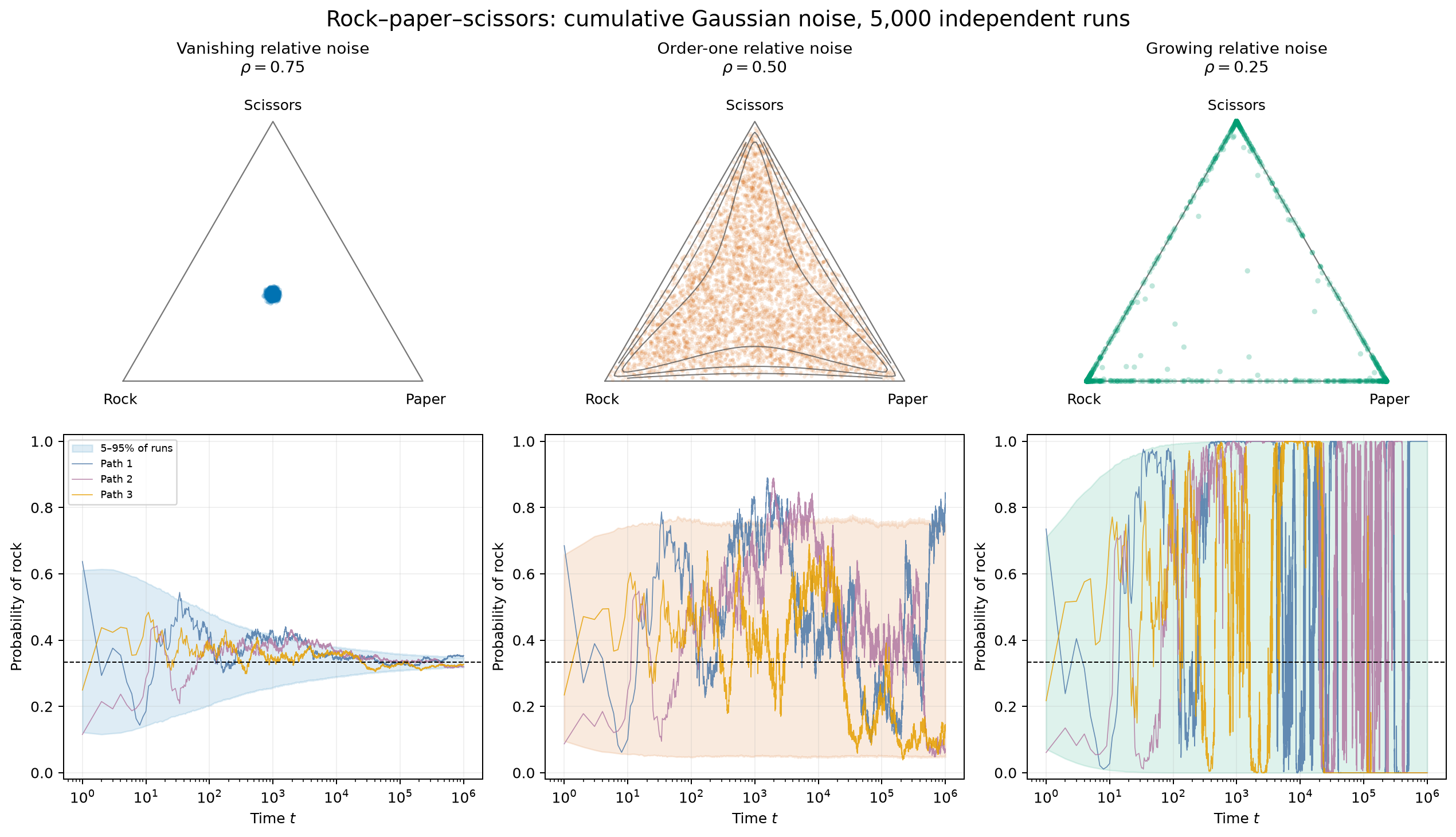}
    \caption{Three relative-noise regimes in rock--paper--scissors,
    with \(\rho=0.75\),
    \(0.50\),
    and \(0.25\) from left to right.
    Top: terminal responses from \(5000\) runs at
    \(T=10^6\); gray contours show the limiting logistic-normal
    density.
    Bottom: rock probabilities along three sample trajectories,
    using the same runs and colors across
    panels.
    Bands show pointwise empirical 5th--95th percentiles; dashed
    lines mark \(1/3\).}
    \label{fig:rps-noise-regimes}
\vspace{-1.5em}
\end{figure}

\paragraph{Supplementary measurements and trajectory plots.}
The trajectories in the bottom row of \cref{fig:rps-noise-regimes}
were chosen independently of their outcomes; matching colors show
the same runs across schedules.
The bands are pointwise empirical 5th--95th percentiles across all
\(5000\) runs and describe the spread of individual responses.
For \regime{I}, the mean terminal \(\ell_1\)-distance from
\(\bar y\) is \(0.02050\).
For \regime{II}, the sample variance of the terminal rock probability
is \(0.04790\), compared with the limiting variance \(0.04941\)
evaluated by Gaussian quadrature.
For \regime{III}, \(4410/5000=88.20\%\) of terminal responses
satisfy \(\max_j y_{T,j}\ge0.99\); the trajectories may still switch
between near-pure actions, consistently with the marginal limit
proved above.

In this experiment, the Gaussian law is invariant under coordinate
permutations, and softmax commutes with those permutations.
The three response coordinates therefore have equal expectations
and sum to one, giving \(\mathbb E[y_t]=\bar y\) at every time
for all three schedules. This verifies that the different response
distributions in the figures have the same mean.

\section{Proofs for Normal Exploitation: Incomplete Localization}
\label{sec:passive}

This appendix proves
\cref{prop:normal-exploitation-informal,thm:main-normal-bounded-surplus,cor:main-normal-maxmin-bounded-surplus}
from \cref{subsec:normal-exploitation} and derives the criteria for
standard regularizers cited there.
The coordinate estimates are stated for an arbitrary constant
strategy \(\hat x\), as required by the bounded-surplus theorem.
Specializing to a max--min strategy gives the normal-surplus bounds.

Throughout, \(h\) is separable as in
\cref{def:separable-regularizer}, with kernel \(\theta\) and
truncated inverse derivative \(\phi\).
The learning rates \(\eta_t>0\) are deterministic, and the noise
may have arbitrary dependence across coordinates and rounds.
Fix \(\hat x\in\Delta_n\) independently of the horizon and repeat
it at every round.

For a strategy that need not be max--min, we measure payoff gaps
from its own best-response payoff. Write
\[
  v=A^\top\hat x,\qquad u_o^\ast=\min_jv_j,\qquad
  K=\br(\hat x),\qquad k=|K|,
\]
and let \(\delta_j=v_j-u_o^\ast\) and \(N_{\rm sub}=m-k\).
Thus \(\delta_j=0\) on \(K\) and \(\delta_j>0\) outside \(K\).
When \(\hat x=x^\star\), we have \(u_o^\ast=\val\), recovering
the gaps in \cref{subsec:normal-exploitation}.
Since \(S_t=-tv\), removing the common score shift gives
\[
  y_t=Q_h(-r_t\delta+\Xi_t),
  \qquad r_t=t\eta_t,\qquad \Xi_t=\eta_tZ_t.
\]

To pass from response bounds to surplus, separate the cost of the
fixed strategy from the gain due to the learner's off-face mass.
The identity
\(\hat x^\top Ay_t=u_o^\ast+\sum_{j\notin K}\delta_jy_{t,j}\)
gives, after taking expectations and summing,
\begin{equation}
  \Surplus(\hat x_{\times T})
    =\vagd(\hat x_{\times T},T)+\lagd(\hat x_{\times T},T),
  \label{eq:fixed-surplus-decomposition-discrete}
\end{equation}
where, following the VAG/LAG notation of
\citet[Definition~2]{su2026nocoin},
\[
  \begin{aligned}
    \vagd(\hat x_{\times T},T)
      &:=T(u_o^\ast-\val)\le0,\\
    \lagd(\hat x_{\times T},T)
      &:=\sum_{t=0}^{T-1}\sum_{j\notin K}
           \delta_j\mathbb E[y_{t,j}]\ge0.
  \end{aligned}
\]
The value gap \(\vagd\) vanishes exactly for max--min play.
The learner-approximation gain \(\lagd\) therefore bounds surplus
from above for every fixed strategy and equals surplus for
\(\hat x=x^\star\).
If \(N_{\rm sub}=0\), then \(\lagd=0\) and surplus is exactly
\(T(u_o^\ast-\val)\).
Whenever \(N_{\rm sub}>0\), we use
\[
  \delta_{\min}=\min_{j\notin K}\delta_j,
  \qquad \delta_{\max}=\max_{j\notin K}\delta_j.
\]

\subsection{Proof of the normal-surplus bounds}

To prove \cref{prop:normal-exploitation-informal}, we first bound
each learner coordinate using the KKT representation from
\appref{app:preliminary-kkt}.
Summing these bounds with their payoff gaps will give the desired
surplus estimates.

\begin{lemma}[Pointwise noisy FTRL bounds under fixed play]
\label{lem:pointwise-noisy-ftrl-bound}
For every round \(t\ge0\), every \(j\in[m]\), and every
realization of \(\Xi_t\),
\begin{subequations}
\begin{align}
  y_{t,j}
    &\ge\phi\!\left(\theta'(1/m)-r_t\delta_j
                              -\spnorm{\Xi_t}\right),
      \label{eq:pointwise-noisy-ftrl-bound:a}\\
  y_{t,j}
    &\le\phi\!\left(\theta'(1/k)-r_t\delta_j
                              +\spnorm{\Xi_t}\right).
      \label{eq:pointwise-noisy-ftrl-bound:b}
\end{align}
\end{subequations}
No distributional assumption on the noise is required.
\end{lemma}

\begin{proof}
Fix \(t\) and a noise realization, and put
\(u_t=-r_t\delta+\Xi_t\).
Choose the common multiplier \(\lambda_t\) in
\eqref{eq:preliminary-kkt-choice} so that
\[
  y_{t,j}=\phi(u_{t,j}+\lambda_t),\qquad
  \begin{cases}
    u_{t,j}+\lambda_t=\theta'(y_{t,j}),&y_{t,j}>0,\\
    u_{t,j}+\lambda_t\le\theta'(0^+),&y_{t,j}=0.
  \end{cases}
\]
Such a choice is available even when the response is a vertex,
as explained in \appref{app:preliminary-kkt}.
We bound \(\lambda_t\) from above and below by choosing a
coordinate whose probability is controlled by normalization.

\paragraph{Upper coordinate bound.}
Since \(\sum_{i\in K}y_{t,i}\le1\), some \(i\in K\) has
\(y_{t,i}\le1/k\).
If this coordinate is positive, monotonicity of \(\theta'\) gives
\(u_{t,i}+\lambda_t\le\theta'(1/k)\).
If it is zero, the inactive-coordinate inequality gives the same
conclusion. Thus \(\lambda_t\le\theta'(1/k)-u_{t,i}\).
Because \(\delta_i=0\), for every \(j\) we obtain
\[
  u_{t,j}+\lambda_t
    \le\theta'(1/k)+u_{t,j}-u_{t,i}
    \le\theta'(1/k)-r_t\delta_j+\spnorm{\Xi_t}.
\]
Applying the nondecreasing map \(\phi\) proves
\eqref{eq:pointwise-noisy-ftrl-bound:b}, including at the boundary.

\paragraph{Lower coordinate bound.}
Some coordinate \(\ell\in[m]\) has \(y_{t,\ell}\ge1/m>0\).
Its KKT equality implies
\[
  \lambda_t=\theta'(y_{t,\ell})-u_{t,\ell}
    \ge\theta'(1/m)-u_{t,\ell}.
\]
Here \(\delta_\ell\ge0\), so
\(u_{t,j}-u_{t,\ell}
=-r_t\delta_j+r_t\delta_\ell+\Xi_{t,j}-\Xi_{t,\ell}
\ge-r_t\delta_j-\spnorm{\Xi_t}\).
Consequently,
\[
  u_{t,j}+\lambda_t
    \ge\theta'(1/m)-r_t\delta_j-\spnorm{\Xi_t}.
\]
Applying \(\phi\) now proves
\eqref{eq:pointwise-noisy-ftrl-bound:a}.
Both comparisons hold for each realization and each round,
including \(t=0\).
\end{proof}

The next proposition sums these coordinate estimates for
\(\lagd\). We retain the general fixed-strategy form because it
is also used in the bounded-surplus argument below.

\begin{proposition}[Bounds on the learner-approximation gain]
\label{prop:fixed-lag-bounds}
If \(N_{\rm sub}=0\), then
\(\lagd(\hat x_{\times T},T)=0\) for every \(T\).
Otherwise, for every integer \(T\ge1\),
\begin{subequations}
\begin{align}
  \lagd(\hat x_{\times T},T)
    &\ge N_{\rm sub}\delta_{\min}\sum_{t=0}^{T-1}
       \mathbb E\!\left[\phi\!\left(
         \theta'(1/m)-\delta_{\max}r_t-\spnorm{\Xi_t}
       \right)\right],\label{eq:normal-surplus-lower}\\
  \lagd(\hat x_{\times T},T)
    &\le N_{\rm sub}\delta_{\max}\sum_{t=0}^{T-1}
       \mathbb E\!\left[\phi\!\left(
         \theta'(1/k)-\delta_{\min}r_t+\spnorm{\Xi_t}
       \right)\right].\label{eq:normal-surplus-upper}
\end{align}
\end{subequations}
No distributional assumption on the noise is required.
\end{proposition}

\begin{proof}
If \(N_{\rm sub}=0\), the sum defining \(\lagd\) is empty.
Otherwise, fix \(j\notin K\).
For the lower bound, \cref{lem:pointwise-noisy-ftrl-bound} gives
\[
  \delta_jy_{t,j}
    \ge\delta_j\phi\!\left(
         \theta'(1/m)-r_t\delta_j-\spnorm{\Xi_t}\right)
    \ge\delta_{\min}\phi\!\left(
         \theta'(1/m)-r_t\delta_{\max}-\spnorm{\Xi_t}\right).
\]
The second inequality uses two monotonicities: increasing the gap
inside \(\phi\) decreases its argument, and decreasing the
nonnegative prefactor decreases the product.
For the upper bound, the reverse choices yield
\[
  \delta_jy_{t,j}
    \le\delta_j\phi\!\left(
         \theta'(1/k)-r_t\delta_j+\spnorm{\Xi_t}\right)
    \le\delta_{\max}\phi\!\left(
         \theta'(1/k)-r_t\delta_{\min}+\spnorm{\Xi_t}\right).
\]
There are \(N_{\rm sub}\) positive-gap coordinates.
Taking expectations and summing these inequalities over those
coordinates and \(t<T\) gives the two bounds on \(\lagd\).
\end{proof}

\begin{corollary}[\cref{prop:normal-exploitation-informal} as a special case]
\label{prop:normal-surplus-bounds}
For a fixed max--min strategy \(x^\star\),
\begin{equation}
  \Surplus((x^\star)_{\times T})
    =\lagd((x^\star)_{\times T},T).
  \label{eq:normal-surplus-identity}
\end{equation}
Consequently, \(\Surplus((x^\star)_{\times T})=0\) if
\(N_{\rm sub}=0\).
Otherwise, \eqref{eq:normal-surplus-lower} and
\eqref{eq:normal-surplus-upper} hold with
\(\lagd(\hat x_{\times T},T)\) replaced by
\(\Surplus((x^\star)_{\times T})\).
\end{corollary}

\begin{proof}
For \(\hat x=x^\star\), max--min optimality gives
\(u_o^\ast=\val\). Hence \(\vagd=0\) in
\eqref{eq:fixed-surplus-decomposition-discrete}, proving
\eqref{eq:normal-surplus-identity}.
Substituting this identity into \cref{prop:fixed-lag-bounds}
gives exactly the two inequalities and the zero-gain case in
\cref{prop:normal-exploitation-informal}.
\end{proof}

\subsection{Separating the response series from the noise}

To use the summability conditions in
\cref{thm:main-normal-bounded-surplus}, we need to remove the
random noise from the argument of \(\phi\).
The next estimate does this for each coordinate by retaining a
fixed fraction of its deterministic gap on a suitable event and
charging the complementary event through its probability.
The lower estimate records the corresponding loss from
exceptional noise events.

\begin{proposition}[Direct bounds on the learner-approximation gain]
\label{thm:fixed-noisy-lag}
For every fixed \(\hat x\), every \(\alpha\in(0,1)\), every
\(\beta>0\), and every \(T\ge1\),
\begin{align}
  \lagd(\hat x_{\times T},T)
    \le\sum_{j\notin K}\delta_j\sum_{t=0}^{T-1}
       \phi\!\left(\theta'(1/k)-\alpha\delta_jr_t\right)+\sum_{j\notin K}\delta_j\sum_{t=0}^{T-1}
       \Pr\!\left(\spnorm{\Xi_t}>(1-\alpha)\delta_jr_t\right),
       \label{eq:fixed-lag-master-upper}\\
  \lagd(\hat x_{\times T},T)
    \ge\sum_{j\notin K}\delta_j\sum_{t=0}^{T-1}
       \phi\!\left(\theta'(1/m)-(1+\beta)\delta_jr_t\right)-\sum_{j\notin K}\delta_j\sum_{t=0}^{T-1}
       \Pr\!\left(\spnorm{\Xi_t}>\beta\delta_jr_t\right).
       \notag
\end{align}
\end{proposition}

\begin{proof}
Fix \(j\notin K\).
For the upper bound, retaining at least the fraction \(\alpha\)
of its gap requires
\(-r_t\delta_j+\spnorm{\Xi_t}\le-\alpha\delta_jr_t\).
This motivates the event
\[
  E_{t,j}:=\{\spnorm{\Xi_t}\le(1-\alpha)\delta_jr_t\}.
\]
On \(E_{t,j}\), \eqref{eq:pointwise-noisy-ftrl-bound:b} gives
\(y_{t,j}\le\phi(\theta'(1/k)-\alpha\delta_jr_t)\).
On \(E_{t,j}^c\), use \(y_{t,j}\le1\).
Splitting the expectation therefore yields
\[
  \begin{aligned}
    \mathbb E[y_{t,j}]
      &\le\phi\!\left(\theta'(1/k)-\alpha\delta_jr_t\right)
            \Pr(E_{t,j})+\Pr(E_{t,j}^c)\\
      &\le\phi\!\left(\theta'(1/k)-\alpha\delta_jr_t\right)
            +\Pr(E_{t,j}^c).
  \end{aligned}
\]
The first term bounds the response when the required fraction of
the gap survives, and the second bounds the contribution of the
exceptional event.

For the lower bound, limit the additional disadvantage caused by
noise to \(\beta\delta_jr_t\), using
\[
  F_{t,j}:=\{\spnorm{\Xi_t}\le\beta\delta_jr_t\}.
\]
On this event, \eqref{eq:pointwise-noisy-ftrl-bound:a} gives
\(y_{t,j}\ge\phi(\theta'(1/m)-(1+\beta)\delta_jr_t)\).
On the complement, the coordinate is nonnegative.
Since \(0\le\phi\le1\),
\[
  \begin{aligned}
    \mathbb E[y_{t,j}]
      &\ge\phi\!\left(\theta'(1/m)-(1+\beta)\delta_jr_t\right)
            \Pr(F_{t,j})\\
      &\ge\phi\!\left(\theta'(1/m)-(1+\beta)\delta_jr_t\right)
            -\Pr(F_{t,j}^c).
  \end{aligned}
\]
Multiply the upper and lower expected-coordinate bounds by
\(\delta_j\), then sum over \(j\notin K\) and \(t<T\).
This gives the proposition. If \(N_{\rm sub}=0\), both sides
are zero by empty summation.
Only marginal event probabilities enter the argument, so no
temporal independence is needed.
\end{proof}

\subsection{Proof of bounded surplus for a fixed strategy}

We now prove \cref{thm:main-normal-bounded-surplus}.
Using one event for all positive-gap coordinates lets us charge
an exceptional round by \(\delta_{\max}\), since their total
probability is at most one.
For \(N_{\rm sub}>0\) and \(\alpha\in(0,1)\), define
\begin{equation*}
  E_t:=\{\spnorm{\Xi_t}\le(1-\alpha)\delta_{\min}r_t\}
      =\{\spnorm{Z_t}\le(1-\alpha)\delta_{\min}t\}.
\end{equation*}
On \(E_t\), noise removes at most a fraction \(1-\alpha\)
of the smallest positive gap, and therefore of every positive gap.
The equality of the two event descriptions uses only
\(\Xi_t=\eta_tZ_t\), \(r_t=t\eta_t\), and \(\eta_t>0\).

We first state the summability conditions in the form used for
the more detailed fixed-strategy bound.
The response condition sums the payoff-weighted coordinate
bounds on \(E_t\).
\begin{assumption}[Summable response series]
\label{ass:fixed-response-summability}
For the fixed strategy \(\hat x\), there exists
\(\alpha\in(0,1)\) such that
\begin{equation*}
  \sum_{j\notin K}\delta_j\sum_{t=0}^{\infty}
    \phi\!\left(\theta'(1/k)-\alpha\delta_jr_t\right)<\infty.
\end{equation*}
\end{assumption}

The following noise condition is uniform over fixed positive
thresholds, so it can be applied to different choices of \(\hat x\).
\begin{assumption}[Summable noise tails]
\label{ass:passive-linear-tail}
For every constant \(c>0\),
\begin{equation*}
  \sum_{t=1}^{\infty}\Pr\!\left(\spnorm{Z_t}>ct\right)<\infty.
\end{equation*}
\end{assumption}
Equivalently, \(\sum_{t=1}^{\infty}
\Pr(\spnorm{\Xi_t}>cr_t)<\infty\) for every \(c>0\).
After proving the bound under these assumptions, we verify it
directly under the strategy-specific conditions in the main text.

\begin{theorem}[Bounded gain for a fixed strategy]
\label{thm:unknown-horizon-pointwise}
Suppose the optimizer plays a fixed strategy \(\hat x_{\times T}\).
If \(N_{\rm sub}=0\), then
\(\lagd(\hat x_{\times T},T)=0\) for every \(T\).
Otherwise, suppose
\cref{ass:fixed-response-summability,ass:passive-linear-tail} hold,
and choose \(\alpha\) as in \cref{ass:fixed-response-summability}.
Then a finite constant \(C_{\hat x}\), independent of \(T\),
satisfies
\(0\le\lagd(\hat x_{\times T},T)\le C_{\hat x}\)
for every \(T\ge1\). One may take
\begin{equation}
  C_{\hat x}
    :=\sum_{j\notin K}\delta_j\sum_{t=0}^{\infty}
        \phi\!\left(\theta'(1/k)-\alpha\delta_jr_t\right)
      +\delta_{\max}\sum_{t=0}^{\infty}\Pr(E_t^c).
  \label{eq:fixed-lag-constant}
\end{equation}
With \(C_{\hat x}=0\) in the case \(N_{\rm sub}=0\), both cases give
\begin{equation}
  \Surplus(\hat x_{\times T})
    =\begin{cases}
       O_{\hat x}(1),&\hat x\text{ is max--min},\\
       -(\val-u_o^\ast)T+O_{\hat x}(1),
         &\hat x\text{ is not max--min},
     \end{cases}
  \label{eq:fixed-strategy-dichotomy}
\end{equation}
and \([\Surplus(\hat x_{\times T})]_+\le C_{\hat x}\)
for every \(T\ge1\).
\end{theorem}

\begin{proof}
If \(N_{\rm sub}=0\), the gain is zero and
\(\Surplus(\hat x_{\times T})=T(u_o^\ast-\val)\le0\), proving
all claims with \(C_{\hat x}=0\).
Hence assume \(N_{\rm sub}>0\).
Taking \(c=(1-\alpha)\delta_{\min}>0\) in
\cref{ass:passive-linear-tail} gives
\(\sum_{t=0}^{\infty}\Pr(E_t^c)<\infty\); the extra
\(t=0\) term is at most one.

\paragraph{Controlling each round on the common event.}
For every \(j\notin K\), on \(E_t\) we have
\[
  -r_t\delta_j+\spnorm{\Xi_t}
    \le-r_t\delta_j+(1-\alpha)\delta_{\min}r_t
    \le-\alpha\delta_jr_t.
\]
The final inequality uses \(\delta_{\min}\le\delta_j\), showing
that at least the fraction \(\alpha\) of each individual gap
survives. The upper coordinate estimate
\eqref{eq:pointwise-noisy-ftrl-bound:b} therefore gives
\[
  \sum_{j\notin K}\delta_jy_{t,j}
    \le\sum_{j\notin K}\delta_j
      \phi\!\left(\theta'(1/k)-\alpha\delta_jr_t\right)
    \qquad\text{on }E_t.
\]
On \(E_t^c\), normalization gives the uniform bound
\[
  0\le\sum_{j\notin K}\delta_jy_{t,j}
    \le\delta_{\max}\sum_{j\notin K}y_{t,j}
    \le\delta_{\max}.
\]
Thus one exceptional round costs at most \(\delta_{\max}\),
regardless of how many coordinates lie outside \(K\).

\paragraph{Summing the expected gain.}
Split the expectation according to \(E_t\) and its complement:
\[
  \mathbb E\!\left[\sum_{j\notin K}\delta_jy_{t,j}\right]
    \le\sum_{j\notin K}\delta_j
      \phi\!\left(\theta'(1/k)-\alpha\delta_jr_t\right)
      \Pr(E_t)+\delta_{\max}\Pr(E_t^c).
\]
Using \(\Pr(E_t)\le1\) and summing over \(t<T\) bounds
\(\lagd(\hat x_{\times T},T)\) by the corresponding partial
sums in \eqref{eq:fixed-lag-constant}.
The response assumption makes the first infinite sum finite,
and the noise assumption makes the second finite.
All terms are nonnegative, so the same constant bounds every
partial sum and hence every horizon.

\paragraph{Returning to surplus.}
The decomposition \eqref{eq:fixed-surplus-decomposition-discrete}
now reads
\[
  \Surplus(\hat x_{\times T})
    =-(\val-u_o^\ast)T+\lagd(\hat x_{\times T},T).
\]
For max--min \(\hat x\), the first term is zero, giving
\(0\le\Surplus(\hat x_{\times T})\le C_{\hat x}\).
For non-max--min \(\hat x\), the coefficient
\(\val-u_o^\ast\) is positive and the gain term is bounded,
which proves the second case of \eqref{eq:fixed-strategy-dichotomy}.
Since the first term is always nonpositive, the positive part of
surplus is at most \(C_{\hat x}\) in either case.
\end{proof}

\begin{proof}[Proof of \cref{thm:main-normal-bounded-surplus,cor:main-normal-maxmin-bounded-surplus}]
If \(N_{\rm sub}=0\), the zero-gain case just proved applies.
Otherwise, fix the \(\alpha\) in the two main-text summability
conditions. We first relate the response condition
\eqref{eqn:summable-response-tail} to
\cref{ass:fixed-response-summability}.
For each \(t\), nonnegativity and monotonicity of \(\phi\) give
\[
  \begin{aligned}
    \delta_{\min}\phi\!\left(
      \theta'(1/k)-\alpha\delta_{\min}r_t\right)
      &\le\sum_{j\notin K}\delta_j
        \phi\!\left(\theta'(1/k)-\alpha\delta_jr_t\right),\\
    \sum_{j\notin K}\delta_j
      \phi\!\left(\theta'(1/k)-\alpha\delta_jr_t\right)
      &\le\left(\sum_{j\notin K}\delta_j\right)
        \phi\!\left(\theta'(1/k)-\alpha\delta_{\min}r_t\right).
  \end{aligned}
\]
The first inequality retains a coordinate attaining
\(\delta_{\min}\); the second replaces every gap inside
\(\phi\) by \(\delta_{\min}\).
Since the prefactors are positive constants, summing over \(t\)
proves the stated equivalence of the response conditions.

The main-text noise condition \eqref{eqn:summable-noise-tail}
is exactly \(\sum_{t=0}^{\infty}\Pr(E_t^c)<\infty\).
The preceding proof uses the stronger
\cref{ass:passive-linear-tail} only to obtain this series.
Thus both terms in \eqref{eq:fixed-lag-constant} are finite under
the main-text hypotheses, and the same round-by-round estimate
gives
\[
  [\Surplus(\hat x_{\times T})]_+
    \le\lagd(\hat x_{\times T},T)\le C_{\hat x}
    \qquad\text{for every }T.
\]
This proves \cref{thm:main-normal-bounded-surplus}.
For \(\hat x=x^\star\), \eqref{eq:normal-surplus-identity}
identifies surplus with the nonnegative gain, proving
\cref{cor:main-normal-maxmin-bounded-surplus}.
\end{proof}


\subsection{Derivation of the standard-regularizer bounds}
\label{app:normal-standard-regularizer-example}

We now verify the standard-regularizer criteria cited in
\cref{subsec:normal-exploitation}.
The partial-sum bound \eqref{eq:fixed-lag-master-upper} also
gives rates when the full response series is not summable,
as needed for the normal-surplus rates in
\cref{sec:noise-examples}.
Use the power-law schedule
\begin{equation}
  \eta_t=\eta_0(t+1)^{-\rho},
  \qquad \eta_0>0,\qquad 0<\rho<1,
  \label{eq:polynomial-schedule}
\end{equation}
so \(r_t=t\eta_t=\Theta(t^{1-\rho})\) for \(t\ge1\).
Under \cref{ass:passive-linear-tail}, the noise contribution
is bounded independently of the horizon; the proof below
evaluates the remaining response sums for each kernel.

\begin{corollary}
\label{cor:fixed-standard-schedule-criteria}
Under \eqref{eq:polynomial-schedule} and
\cref{ass:passive-linear-tail}, every optimizer strategy
\(\hat x\) fixed independently of \(T\) satisfies
\(\lagd(\hat x_{\times T},T)=O_{\hat x}(1)\)
for negative entropy and Euclidean regularization.
For negative Tsallis regularization with \(q\in(0,1)\),
\begin{equation}
  \lagd(\hat x_{\times T},T)=
  \begin{cases}
    O_{\hat x}(1),&\rho<q,\\
    O_{\hat x}(\log T),&\rho=q,\\
    O_{\hat x}(T^{(\rho-q)/(1-q)}),&\rho>q.
  \end{cases}
  \label{eq:tsallis-polynomial-lag-rates}
\end{equation}
For max--min \(\hat x=x^\star\), these are also upper bounds on
\(\Surplus((x^\star)_{\times T})\).
For a fixed non-max--min \(\hat x\),
\(\Surplus(\hat x_{\times T})
=-(\val-u_o^\ast)T+o(T)=-\Theta(T)\).
If \(N_{\rm sub}=0\), then \(\lagd=0\) for every \(T\).
\end{corollary}

\begin{proof}
The case \(N_{\rm sub}=0\) follows from the empty sum defining
\(\lagd\). Suppose \(N_{\rm sub}>0\).
Apply \eqref{eq:fixed-lag-master-upper} with \(\alpha=1/2\),
retaining half of each deterministic gap.
For every \(j\notin K\), \cref{ass:passive-linear-tail} gives
\[
  \sum_{t=0}^{\infty}
    \Pr\!\left(\spnorm{\Xi_t}>\delta_jr_t/2\right)<\infty,
\]
where the single \(t=0\) term is at most one.
It therefore remains to estimate, for each positive gap,
\[
  \sum_{t=0}^{T-1}
    \phi\!\left(\theta'(1/k)-\delta_jr_t/2\right).
\]
We derive the scalar response for each kernel before summing it.
In the inversions below, \(a\in(0,1]\) and \(s\ge0\).

\paragraph{Negative entropy.}
Since \(\theta'(u)=1+\log u\), solving
\(\theta'(u)=\theta'(a)-s\) gives \(u=ae^{-s}\).
This lies in \((0,a]\), so truncation does not change it.
Hence
\[
  \phi\!\left(\theta'(1/k)-\delta_jr_t/2\right)
    =\frac1k e^{-\delta_jr_t/2}.
\]
Because \(r_t=\Theta(t^{1-\rho})\) with \(1-\rho>0\),
we have \(r_t/\log(t+1)\to\infty\).
For all sufficiently large \(t\), therefore,
\(\delta_jr_t/2\ge2\log(t+1)\), and the response is at most
\(k^{-1}(t+1)^{-2}\).
Its infinite sum is finite. Summing over the finitely many
positive gaps and adding the bounded noise term proves
\(\lagd=O_{\hat x}(1)\).

\paragraph{Euclidean regularization.}
Here \(\theta'(u)=u\), so the truncated inverse gives
\(\phi(\theta'(a)-s)=[a-s]_+\).
The response term becomes
\[
  \phi\!\left(\theta'(1/k)-\delta_jr_t/2\right)
    =\left[\frac1k-\frac{\delta_jr_t}{2}\right]_+.
\]
Since \(r_t\to\infty\), this term is zero for all sufficiently
large \(t\), once \(\delta_jr_t/2\ge1/k\).
Only finitely many response terms contribute for each fixed
gap. Adding the bounded noise term again gives
\(\lagd=O_{\hat x}(1)\).

\paragraph{Negative Tsallis regularization.}
The derivative is
\(\theta'(u)=(qu^{q-1}-1)/(q-1)\).
Solving \(\theta'(u)=\theta'(a)-s\) yields
\[
  u^{q-1}=a^{q-1}+\frac{1-q}{q}s,
  \qquad
  u=a\left(1+\frac{1-q}{q}a^{1-q}s\right)^{-1/(1-q)}.
\]
This solution lies in \((0,a]\), so it is also the truncated
inverse response. Taking \(a=1/k\) and \(s=\delta_jr_t/2\),
we obtain
\[
  \phi\!\left(\theta'(1/k)-\delta_jr_t/2\right)
    =\frac1k\left(
      1+\frac{1-q}{2q}k^{q-1}\delta_jr_t
      \right)^{-1/(1-q)}.
\]
All coefficients are positive and fixed, and
\(r_t=\Theta(t^{1-\rho})\).
Thus, for \(t\ge1\), the response has order
\(t^{-(1-\rho)/(1-q)}\).
The cumulative bound depends on whether this exponent is above,
equal to, or below one.

To evaluate the sum, put \(p=(1-\rho)/(1-q)>0\).
Monotonicity of \(u^{-p}\) gives
\[
  \sum_{t=0}^{T-1}(t+1)^{-p}
    \le1+\int_1^T u^{-p}\,du
    =\begin{cases}
      1+\log T,&p=1,\\[1mm]
      1+\dfrac{T^{1-p}-1}{1-p},&p\ne1.
    \end{cases}
\]
For \(p>1\), the integral is bounded independently of \(T\);
for \(p=1\), it grows logarithmically; and for \(p<1\), it
has order \(T^{1-p}\).
These cases correspond to \(\rho<q\), \(\rho=q\), and
\(\rho>q\), respectively, with
\(1-p=(\rho-q)/(1-q)\).
Summing over the positive gaps and adding the bounded noise
contribution proves \eqref{eq:tsallis-polynomial-lag-rates}.

Finally, each gain bound is \(o(T)\).
In the growing Tsallis case this uses
\((\rho-q)/(1-q)<1\), which follows from \(\rho<1\).
For max--min play, \eqref{eq:normal-surplus-identity} identifies
\(\lagd\) with surplus.
For non-max--min play, substituting \(\lagd=o(T)\) into
\eqref{eq:fixed-surplus-decomposition-discrete} gives
\[
  \Surplus(\hat x_{\times T})
    =-(\val-u_o^\ast)T+o(T)=-\Theta(T),
\]
since \(\val-u_o^\ast>0\) is fixed.
\end{proof}

\begin{remark}[Typical Tsallis tuning]
The negative \(1/2\)-Tsallis regularizer with a learning rate
of order \(t^{-1/2}\) is a standard bandit choice, as in
Tsallis-INF~\citep{abernethy2015fighting,zimmert2021tsallis}.
For \(q=\rho=1/2\) in \eqref{eq:polynomial-schedule},
\cref{cor:fixed-standard-schedule-criteria} gives
\(\lagd(\hat x_{\times T},T)=O_{\hat x}(\log T)\)
under its noise-tail assumption.
For max--min play, the same bound applies to
\(\Surplus((x^\star)_{\times T})\).
\end{remark}

\section{Proofs for Selection Exploitation: A Profitable Limiting Mean Response}
\label{subsec:tangential-location}

This appendix proves \cref{thm:main-selection-exploitation} and its
specializations to the three noise regimes and power-law schedules.
We use the setup of \cref{subsec:main-new-selection}, first deriving
the score and payoff identities needed in the proof, then establishing
uniform stability and the surplus bound for every horizon.
\appref{app:constant-rate-selection} gives a separate construction
for the noiseless, constant-rate setting.

\subsection{Score and payoff identities for the selection trajectory}
\label{app:effect-past-deviations}

Fix the max--min strategy \(x^\star\) from the main text, and retain
\(K=\br(x^\star)\), \(\delta=A^\top x^\star-\val\mathbf1\),
\(r_t=t\eta_t\), and \(\Xi_t=\eta_tZ_t\).
The reference response under repeated max--min play is
\(y_t^\star=Q_h(-r_t\delta+\Xi_t)\).
The learning-rate schedule is deterministic and known, and the
exogenous noise process has the same law under every optimizer
trajectory and at every horizon. The optimizer does not observe
its realizations. No temporal independence or moment assumption
is needed in the selection proofs below.

We record the mean-response hypothesis from
\cref{thm:main-selection-exploitation} for use in the intermediate
results.
\begin{assumption}[Profitable limiting mean response]
\label{ass:new-selection-mean}
There exists \(\mu\in\Delta_K\) such that
\begin{equation*}
  \mathbb E[y_t^\star]\longrightarrow\mu\in\Delta_K,
  \qquad
  g_{\rm mean}(x^\star;\mu)
    =\max_{x\in\Delta_n}x^\top A\mu-\val>0.
\end{equation*}
\end{assumption}
Write \(g_{\rm mean}\) for this positive gap, choose
\(x^\dagger\in\argmax_{x\in\Delta_n}x^\top A\mu\), and set
\(a=A^\top(x^\dagger-x^\star)\), as in the main text.
Since \(\mu\) is supported on \(K\),
\(x^{\star\top}A\mu=\val\), whereas
\(x^{\dagger\top}A\mu=\val+g_{\rm mean}\).
Consequently,
\begin{align}
  \langle a,\mu\rangle
    &=x^{\dagger\top}A\mu-x^{\star\top}A\mu
      =g_{\rm mean},\notag\\
  \mathbb E\langle a,y_t^\star\rangle
    &=\sum_{j=1}^m a_j\mathbb E[y_{t,j}^\star]
      \longrightarrow\sum_{j=1}^m a_j\mu_j
      =g_{\rm mean}.\label{eq:new-selection-direction:b}
\end{align}
This identifies the reference directional payoff that the actual
trajectory must preserve. Only the mean response must converge;
the response itself may remain random.

For deterministic weights \(0\le\varepsilon_t\le1\), consider
\(\widehat x_t=(1-\varepsilon_t)x^\star+\varepsilon_t x^\dagger\)
and let
\[
  b_t:=\sum_{\tau=0}^{t-1}\varepsilon_\tau,\qquad b_0=0.
\]
This cumulative weight measures the effect of all deviations
strictly before round \(t\).
Indeed, \(A^\top\widehat x_\tau=A^\top x^\star+\varepsilon_\tau a\),
so, using \(B=-A\),
\[
  \begin{aligned}
    S_t
      &=-\sum_{\tau=0}^{t-1}
          (A^\top x^\star+\varepsilon_\tau a)
        =-tA^\top x^\star-b_ta,\\
    \eta_tS_t+\Xi_t
      &=-r_t\val\mathbf1-r_t\delta-\eta_tb_ta+\Xi_t.
  \end{aligned}
\]
Removing the common score shift gives the exact response identity
\begin{equation}
  y_t=Q_h(-r_t\delta-\eta_tb_ta+\Xi_t).
  \label{eq:new-selection-score}
\end{equation}
Thus the score change relative to \(y_t^\star\) is governed by
\(\eta_tb_t\). Controlling the current weight alone would not
control the cumulative score change.

The corresponding round payoff separates into a nonnegative
max--min contribution and a directional contribution:
\begin{subequations}\label{eq:new-selection-payoff}
\begin{align}
  \widehat x_t^\top Ay_t-\val
    &=\langle A^\top x^\star+\varepsilon_ta
                -\val\mathbf1,y_t\rangle
      \\
    &=\langle\delta,y_t\rangle
        +\varepsilon_t\langle a,y_t\rangle
      \ge\varepsilon_t\langle a,y_t\rangle.
      \label{eq:new-selection-payoff:b}
\end{align}
\end{subequations}
Here \(\delta\ge0\), so the first term cannot reduce surplus.
It remains to show that a small accumulated score change preserves
a positive expected payoff in direction \(a\).

\subsection{Uniform stability of the response and directional payoff}
\label{app:uniform-stability}

The reference scores \(-r_t\delta+\Xi_t\) may be large or random.
The proof therefore needs one score-tilt threshold that works for
every reference score. Pointwise continuity alone does not supply
that threshold. The following argument obtains it from compactness
and strict convexity under \cref{def:regularizer}.

\begin{lemma}[Uniform stability of the choice map]
\label{lem:new-selection-uniform-stability}
Let \(h\) be a regularizer as in \cref{def:regularizer}.
For every fixed \(a\in\mathbb R^m\),
\begin{equation*}
  \lim_{c\downarrow0}
  \sup_{u\in\mathbb R^m}\sup_{0\le\lambda\le c}
    \|Q_h(u-\lambda a)-Q_h(u)\|_1=0.
\end{equation*}
\end{lemma}

\begin{proof}
The case \(m=1\) is immediate. Otherwise, fix
\(0<\epsilon\le2\).
To obtain a bound independent of the scores, measure the
strict-convexity gap between responses at distance at least
\(\epsilon\):
\[
  c_\epsilon
    :=\min_{\substack{p,q\in\Delta_m\\\|p-q\|_1\ge\epsilon}}
       \left\{\frac{h(p)+h(q)}2-h\!\left(\frac{p+q}2\right)\right\}.
\]
The constraint set is compact, the objective is continuous, and
strict convexity makes it positive at every feasible pair.
Hence \(c_\epsilon>0\).

Now take arbitrary \(u,v\in\mathbb R^m\) and put
\(p=Q_h(u)\), \(q=Q_h(v)\).
Each optimizer can be compared with their feasible midpoint:
\[
  \begin{aligned}
    \langle u,p\rangle-h(p)
      &\ge\left\langle u,\frac{p+q}2\right\rangle
          -h\!\left(\frac{p+q}2\right),\\
    \langle v,q\rangle-h(q)
      &\ge\left\langle v,\frac{p+q}2\right\rangle
          -h\!\left(\frac{p+q}2\right).
  \end{aligned}
\]
Adding cancels the common score contribution and gives
\[
  \langle u-v,p-q\rangle
    \ge4\left\{\frac{h(p)+h(q)}2
                    -h\!\left(\frac{p+q}2\right)\right\}.
\]
If \(\|p-q\|_1\ge\epsilon\), the right-hand side is at least
\(4c_\epsilon\). The left-hand side, however, depends only on
the relative score difference. Subtracting the midpoint of the
largest and smallest coordinates of \(u-v\) leaves it unchanged,
because \(\sum_j(p_j-q_j)=0\). Thus
\[
  \langle u-v,p-q\rangle
    \le\tfrac12\spnorm{u-v}\|p-q\|_1
    \le\spnorm{u-v}.
\]
Responses at distance at least \(\epsilon\) therefore require
\(\spnorm{u-v}\ge4c_\epsilon\).

For \(v=u-\lambda a\), this span equals
\(\lambda\spnorm{a}\).
If \(\spnorm{a}>0\), choose
\(0<c\le2c_\epsilon/\spnorm{a}\).
Then every \(u\) and \(0\le\lambda\le c\) satisfy
\(\spnorm{u-v}\le2c_\epsilon<4c_\epsilon\), forcing
\(\|Q_h(u-\lambda a)-Q_h(u)\|_1<\epsilon\).
If \(\spnorm{a}=0\), shift invariance makes the responses
identical. Since \(\epsilon\) is arbitrary, the uniform limit
follows.
\end{proof}

For \eqref{eq:new-selection-payoff}, the relevant consequence is
stability of the directional payoff. We choose the response
tolerance so that the induced payoff error uses at most one
quarter of the profitable gap.

\begin{corollary}[Uniform stability of the directional payoff]
\label{cor:new-selection-payoff-stability}
Suppose \cref{ass:new-selection-mean} holds.
Let \(x^\dagger\in\argmax_{x\in\Delta_n}x^\top A\mu\) and
\(a=A^\top(x^\dagger-x^\star)\).
There exists \(c_h>0\) such that, for every \(u\in\mathbb R^m\)
and \(0\le\lambda\le c_h\),
\begin{equation}
  \bigl|\langle a,Q_h(u-\lambda a)-Q_h(u)\rangle\bigr|
    \le\frac{g_{\rm mean}}4.
  \label{eq:new-selection-payoff-stability}
\end{equation}
\end{corollary}

\begin{proof}
Since \(\langle a,\mu\rangle=g_{\rm mean}>0\), we have
\(\|a\|_\infty>0\).
By \cref{lem:new-selection-uniform-stability}, choose \(c_h>0\)
so that
\[
  \|Q_h(u-\lambda a)-Q_h(u)\|_1
    \le\frac{g_{\rm mean}}{4\|a\|_\infty}
\]
for every \(u\) and \(0\le\lambda\le c_h\).
The estimate
\(|\langle a,p-q\rangle|\le\|a\|_\infty\|p-q\|_1\)
then gives \eqref{eq:new-selection-payoff-stability}.
\end{proof}

This choice of \(c_h\) depends on \(h\), \(a\), and
\(g_{\rm mean}\), but not on \(u\).
Substituting \(u=-r_t\delta+\Xi_t\) and
\(\lambda=\eta_tb_t\) therefore gives the same payoff bound
at every round and for every noise realization whenever
\(\eta_tb_t\le c_h\).

\subsection{Proof of selection surplus at every horizon}
\label{app:selection-schedule}

We now prove \cref{thm:main-selection-exploitation}.
The stability bound permits a cumulative deviation of size
\(c/\eta_t\) for any fixed \(c\le c_h\).
To enforce \(b_0=0\), the main-text construction uses
\(b_t=c(1/\eta_t-1/\eta_0)\). Its successive differences
give the weights below. Nonincreasing rates make these weights
nonnegative, and bounded reciprocal increments allow one fixed
\(c>0\) to keep them at most one.

\begin{theorem}[Formal statement of \cref{thm:main-selection-exploitation}]
\label{thm:new-selection-every-horizon}
Suppose \cref{ass:new-selection-mean} holds.
Assume \(\eta_t>0\) is nonincreasing and
\(D=\sup_{t\ge0}(1/\eta_{t+1}-1/\eta_t)<\infty\).
There exists \(c_0>0\) such that, for every fixed
\(c\in(0,c_0]\), the deterministic sequence
\begin{equation}
  \varepsilon_t=c\left(\frac1{\eta_{t+1}}-\frac1{\eta_t}\right),
  \qquad
  \widehat x_t=(1-\varepsilon_t)x^\star
                       +\varepsilon_t x^\dagger
  \label{eq:new-selection-strategy}
\end{equation}
satisfies \(0\le\varepsilon_t\le1\) and
\(\widehat x_t\in\Delta_n\) for every \(t\ge0\).
There is a finite constant \(C_0\ge0\), independent of \(T\),
such that
\begin{equation}
  \Surplus(\widehat x_{<T})
    \ge\frac{c\,g_{\rm mean}}2
       \left(\frac1{\eta_T}-\frac1{\eta_0}\right)-C_0
    \qquad\text{for all }T\ge1.
  \label{eq:new-selection-prefix-bound}
\end{equation}
If \(\eta_T\to0\), this one sequence earns
\(\Omega(1/\eta_T)\) expected surplus for all sufficiently
large \(T\).
\end{theorem}

\begin{proof}
We first bound the accumulated score change, then use payoff
stability to sum the gain after the reference mean is sufficiently
close to its limit.

\paragraph{Feasibility and control of the accumulated score change.}
The scale \(c\) must satisfy two requirements: the weights must
lie in \([0,1]\), and the score tilt must remain within the
payoff-stability range. Choose
\[
  c_0=\min\left\{c_h,\frac1{1+D}\right\},
\]
where \(c_h\) is from \cref{cor:new-selection-payoff-stability}.
For \(0<c\le c_0\),
\[
  0\le\varepsilon_t
    \le cD\le\frac{D}{1+D}\le1.
\]
Thus each \(\widehat x_t\) is a feasible convex combination.
The cumulative weight telescopes:
\[
  \begin{aligned}
    b_t
      =c\sum_{\tau=0}^{t-1}
        \left(\frac1{\eta_{\tau+1}}-\frac1{\eta_\tau}\right)=c\left(\frac1{\eta_t}-\frac1{\eta_0}\right).
  \end{aligned}
\]
The same identity holds at \(t=0\), when both sides are zero.
Since \(0<\eta_t\le\eta_0\),
\begin{equation}
  \eta_tb_t=c\left(1-\frac{\eta_t}{\eta_0}\right),
  \qquad 0\le\eta_tb_t\le c\le c_h.
  \label{eq:new-selection-budget}
\end{equation}
Hence the entire past intervention remains within the stability
range at every round.

\paragraph{Preserving a positive directional payoff.}
Compare \(y_t\) and \(y_t^\star\) using the same noise vector
\(\Xi_t\). This gives the correct marginal law for each response
because the noise is exogenous and its law is fixed across
strategies.
In \eqref{eq:new-selection-payoff-stability}, substitute
\(u=-r_t\delta+\Xi_t\) and \(\lambda=\eta_tb_t\).
By \eqref{eq:new-selection-score} and
\eqref{eq:new-selection-budget},
\[
  |\langle a,y_t-y_t^\star\rangle|\le\frac{g_{\rm mean}}4
\]
for every realization. The directional payoffs are bounded by
\(\|a\|_\infty\), so taking expectations gives
\[
  \left|\mathbb E\langle a,y_t\rangle
         -\mathbb E\langle a,y_t^\star\rangle\right|
    \le\mathbb E|\langle a,y_t-y_t^\star\rangle|
    \le\frac{g_{\rm mean}}4.
\]

By \eqref{eq:new-selection-direction:b}, the reference directional
payoff converges to \(g_{\rm mean}\).
Choose a finite integer \(t_0\ge1\) such that its distance from
\(g_{\rm mean}\) is at most \(g_{\rm mean}/4\) for every
\(t\ge t_0\). Then
\begin{align*}
  \mathbb E\langle a,y_t^\star\rangle
    &\ge g_{\rm mean}-\frac{g_{\rm mean}}4,
      \\
  \mathbb E\langle a,y_t\rangle
    &\ge\mathbb E\langle a,y_t^\star\rangle
            -\frac{g_{\rm mean}}4
      \ge\frac{g_{\rm mean}}2
      \qquad(t\ge t_0).
\end{align*}
One allowance of \(g_{\rm mean}/4\) covers convergence of the
reference mean; the other covers the score tilt.
No convergence rate is required: slower convergence only changes
the finite index \(t_0\).

\paragraph{Summing the gain, including the initial rounds.}
For \(t\ge t_0\), determinism and nonnegativity of
\(\varepsilon_t\) allow us to combine
\eqref{eq:new-selection-payoff:b} with the preceding bound:
\[
  \mathbb E[\widehat x_t^\top Ay_t-\val]
    \ge\varepsilon_t\mathbb E\langle a,y_t\rangle
    \ge\frac{g_{\rm mean}}2\varepsilon_t.
\]
For the earlier rounds, \(\langle a,y_t\rangle\ge-\|a\|_\infty\)
for every simplex response, so their expected surplus is at least
\(-\varepsilon_t\|a\|_\infty\).
For \(T\ge t_0\), split the sum at \(t_0\):
\begin{align*}
  \Surplus(\widehat x_{<T})
    &=\sum_{t=0}^{t_0-1}\mathbb E[\widehat x_t^\top Ay_t-\val]
      +\sum_{t=t_0}^{T-1}\mathbb E[\widehat x_t^\top Ay_t-\val]
      \\
    &\ge-\|a\|_\infty\sum_{t=0}^{t_0-1}\varepsilon_t
      +\frac{g_{\rm mean}}2\sum_{t=t_0}^{T-1}\varepsilon_t
      \\
    &=-\|a\|_\infty b_{t_0}
      +\frac{g_{\rm mean}}2(b_T-b_{t_0})
      \\
    &=\frac{g_{\rm mean}}2b_T
      -\left(\frac{g_{\rm mean}}2+\|a\|_\infty\right)b_{t_0}.
\end{align*}
The initial rounds cost at most
\[
  C_0:=\left(\frac{g_{\rm mean}}2+\|a\|_\infty\right)b_{t_0},
\]
which is finite and independent of \(T\).
This choice also covers shorter prefixes. If \(T<t_0\),
then \(b_T\le b_{t_0}\) and
\[
  \begin{aligned}
    \Surplus(\widehat x_{<T})
      &\ge-\|a\|_\infty b_T=\frac{g_{\rm mean}}2b_T
        -\left(\frac{g_{\rm mean}}2+\|a\|_\infty\right)b_T\\
      &\ge\frac{g_{\rm mean}}2b_T-C_0.
  \end{aligned}
\]
Thus the same lower bound holds for every prefix.
Substituting the formula for \(b_T\) proves
\eqref{eq:new-selection-prefix-bound}.

Finally, if \(\eta_T\to0\), then for all sufficiently large \(T\),
\[
  \frac{c\,g_{\rm mean}}{4\eta_T}
    \ge C_0+\frac{c\,g_{\rm mean}}{2\eta_0}.
\]
The prefix bound is then at least
\(c\,g_{\rm mean}/(4\eta_T)\), proving the claimed
\(\Omega(1/\eta_T)\) guarantee along the same trajectory.
\end{proof}

Taking \(C=C_0+c\,g_{\rm mean}/(2\eta_0)\) in
\eqref{eq:new-selection-prefix-bound} gives exactly
\cref{thm:main-selection-exploitation}.
The constants may depend on the game, regularizer, schedule, and
reference noise law, but not on the terminal horizon.
The strategy uses the known \(\eta_t,\eta_{t+1}\), a fixed
profitable direction, and a fixed \(c\).
The index \(t_0\) is only used to bound the initial cost; this
construction does not wait for \(t_0\) or observe the noise.

\paragraph{Applying the three response limits.}
The following corollary verifies the profitable-mean hypothesis
using the limits proved in \appref{app:normal}.

\begin{corollary}[Unknown-horizon selection in the three noise regimes]
\label{cor:new-selection-regimes}
Under the assumptions on \(h\) and the learning rates in
\cref{thm:new-selection-every-horizon}, fix a max--min strategy
\(x^\star\). Consider the following possibilities for the
reference response:
\begin{enumerate}
  \item \textbf{\regime{I}.} Under the hypotheses of
  \cref{prop:vanishing}, let
  \(\mu=y_K=\argmin_{y\in\Delta_K}h(y)\).
  \item \textbf{\regime{II}.} Under the hypotheses of
  \cref{prop:finite}, let
  \(\mu=\mathbb E[Q_{h,K}(\Xi_{\infty,K})]\).
  \item \textbf{\regime{III}.} Under the hypotheses of
  \cref{thm:diverging}, let
  \(\mu=\sum_{j\in K}\Pr(J_\infty=j)e_j\).
\end{enumerate}
If \(g_{\rm mean}>0\) for the corresponding \(\mu\), there
exists \(c_0>0\) such that, for every fixed \(c\in(0,c_0]\),
the sequence in \eqref{eq:new-selection-strategy} is feasible and
\[
  \Surplus(\widehat x_{<T})
    \ge\frac{c\,g_{\rm mean}}2
      \left(\frac1{\eta_T}-\frac1{\eta_0}\right)-C_0
    \qquad\text{for all }T\ge1,
\]
where \(C_0\ge0\) is finite and independent of \(T\).
If also \(\eta_T\to0\), this single sequence earns
\(\Omega(1/\eta_T)\) expected surplus.
\end{corollary}

\begin{proof}
In \regime{I}, \cref{prop:vanishing} gives
\(\mathbb E[y_t^\star]\to y_K\).
In \Regimes~\regimeref{II} and~\regimeref{III},
\cref{prop:finite} and \cref{thm:diverging} give convergence
in distribution to \(Q_{h,K}(\Xi_{\infty,K})\) and
\(e_{J_\infty}\), respectively.
Each coordinate is bounded and continuous on the simplex, so its
expectation converges to the corresponding coordinate of the
stated \(\mu\).
In all three cases, \(\mathbb E[y_t^\star]\to\mu\in\Delta_K\).
Together with \(g_{\rm mean}>0\), this verifies
\cref{ass:new-selection-mean};
\cref{thm:new-selection-every-horizon} gives the conclusion.
\end{proof}

\paragraph{Evaluating the power-law schedule.}
For the schedule used in \cref{sec:noise-examples}, we verify
the bounded-increment condition and compute the weights and
surplus explicitly.

\begin{corollary}[Selection surplus under power-law learning rates]
\label{cor:new-selection-power-law}
Under the assumptions of \cref{thm:new-selection-every-horizon},
let \(\eta_t=\eta_0(t+1)^{-\rho}\), with \(\eta_0>0\) and
\(0<\rho<1\).
There exists \(c_0>0\) such that, for every fixed \(c\in(0,c_0]\),
the sequence
\[
  \varepsilon_t=\frac c{\eta_0}
    \bigl((t+2)^\rho-(t+1)^\rho\bigr),
  \qquad
  \widehat x_t=(1-\varepsilon_t)x^\star+\varepsilon_t x^\dagger
\]
satisfies \(0\le\varepsilon_t\le1\) and
\(\widehat x_t\in\Delta_n\) for every \(t\ge0\).
There is a finite \(C_0\ge0\), independent of \(T\), such that
\begin{equation}
  \Surplus(\widehat x_{<T})
    \ge\frac{c\,g_{\rm mean}}{2\eta_0}
      \bigl((T+1)^\rho-1\bigr)-C_0
    \qquad\text{for all }T\ge1.
  \label{eq:new-selection-power-law-surplus}
\end{equation}
This single sequence, chosen independently of the horizon,
earns \(\Omega(T^\rho)\) expected surplus, with weights
\(\varepsilon_t=\Theta((t+1)^{\rho-1})\).
\end{corollary}

\begin{proof}
The schedule is positive and nonincreasing.
Since \(s^\rho\) is increasing and concave for \(0<\rho<1\),
\[
  0\le\frac1{\eta_{t+1}}-\frac1{\eta_t}
    =\frac{(t+2)^\rho-(t+1)^\rho}{\eta_0}
    \le\frac{2^\rho-1}{\eta_0}.
\]
Thus \cref{thm:new-selection-every-horizon} applies, and its
feasible weights are exactly those in the statement.
To evaluate their size, write
\[
  \varepsilon_t=\frac{c\rho}{\eta_0}
    \int_{t+1}^{t+2}s^{\rho-1}\,ds.
\]
The integrand is decreasing, so
\[
  \frac{c\rho}{\eta_0}(t+2)^{\rho-1}
    \le\varepsilon_t
    \le\frac{c\rho}{\eta_0}(t+1)^{\rho-1}.
\]
Because \(t+1\le t+2\le2(t+1)\), these bounds establish
\(\varepsilon_t=\Theta((t+1)^{\rho-1})\).

The cumulative weight telescopes to
\[
  b_T=\frac c{\eta_0}\bigl((T+1)^\rho-1\bigr).
\]
Substitution into the bound
\(\Surplus(\widehat x_{<T})\ge g_{\rm mean}b_T/2-C_0\)
from \cref{thm:new-selection-every-horizon} gives
\eqref{eq:new-selection-power-law-surplus}.
Since \(c\,g_{\rm mean}>0\), \(C_0\) is fixed, and
\((T+1)^\rho-1\sim T^\rho\), the surplus is
\(\Omega(T^\rho)\).
\end{proof}

\subsection{Selection exploitation with a constant learning rate}
\label{app:constant-rate-selection}

We now use the noiseless setting of \citet{su2026nocoin}:
for each horizon \(T\),
\[
  Z_t=0,\qquad \eta(T)>0,
  \qquad \eta(T)\to0,\qquad T\eta(T)\to\infty.
\]
For each fixed \(T\), the learning rate is \(\eta_t\equiv\eta(T)\)
throughout the interaction; its value may depend on \(T\).
Although the difference weights in \eqref{eq:new-selection-strategy}
are zero, selection exploitation remains possible: wait for a
profitable reference response, then spend a finite intervention
budget allowed by uniform stability.

\begin{theorem}[Selection exploitation in the noiseless constant-rate setting]
\label{thm:constant-rate-selection}
Under this setting, let \(h\) satisfy \cref{def:regularizer},
fix a max--min strategy \(x^\star\), and suppose
\[
  y_K=\argmin_{y\in\Delta_K}h(y),\qquad
  g_{\rm mean}:=g_{\rm mean}(x^\star;y_K)>0.
\]
There exists \(c_h>0\) such that, for every fixed
\(c\in(0,c_h]\) and all sufficiently large \(T\), a
deterministic optimizer trajectory satisfies
\begin{equation*}
  \Surplus(\widehat x_{<T})
    \ge\frac{c\,g_{\rm mean}}{2\eta(T)}
    =\Omega\!\left(\frac1{\eta(T)}\right).
\end{equation*}
If \(K=[m]\), the same construction satisfies
\begin{equation}
  \Surplus(\widehat x_{<T})
    =\Theta\!\left(\frac1{\eta(T)}\right).
  \label{eq:constant-rate-selection-full-face}
\end{equation}
The constants are independent of \(T\).
\end{theorem}

\begin{proof}
Write \(\eta=\eta(T)\), choose
\(x^\dagger\in\argmax_{x\in\Delta_n}x^\top Ay_K\), and set
\(a=A^\top(x^\dagger-x^\star)\).
By the noiseless limit in \cref{prop:vanishing},
\(Q_h(-s\delta)\to y_K\) as \(s\to\infty\).
Since \(\langle a,y_K\rangle=g_{\rm mean}\), there is a
constant \(R\ge0\), independent of \(T\) and \(\eta\), such that
\[
  \langle a,Q_h(-s\delta)\rangle
    \ge\frac{3g_{\rm mean}}4\qquad(s\ge R).
\]
Thus waiting until \(t_0=\lceil R/\eta\rceil\) makes the
reference response profitable at every later round.

Take \(c_h\) from \cref{cor:new-selection-payoff-stability};
it depends only on \(h,a,g_{\rm mean}\).
For fixed \(0<c\le c_h\), play
\(\widehat x_t=(1-\varepsilon_t)x^\star+\varepsilon_t x^\dagger\), where
\begin{equation*}
  \varepsilon_t=
  \begin{cases}
    0,&t<t_0,\\
    \min\{1,[c/\eta-(t-t_0)]_+\},&t\ge t_0,
  \end{cases}
  \qquad [z]_+:=\max\{z,0\}.
\end{equation*}
These weights lie in \([0,1]\), spend total mass \(c/\eta\)
in \(N=\lceil c/\eta\rceil\) rounds, and are then zero.
The final weight is \(c/\eta-(N-1)\in(0,1]\), so fractional
budgets are included. In particular, \(\eta b_t\le c\).
The score identity \eqref{eq:new-selection-score} and payoff
stability therefore give
\[
  \langle a,y_t\rangle
    \ge\langle a,Q_h(-t\eta\delta)\rangle-\frac{g_{\rm mean}}4
    \ge\frac{g_{\rm mean}}2\qquad(t\ge t_0).
\]
Every intervention round earns at least
\(\varepsilon_tg_{\rm mean}/2\), while all other rounds use
\(x^\star\) and have nonnegative surplus by
\eqref{eq:new-selection-payoff}.

It remains to check that the budget is spent within the horizon.
The assumption \(T\eta(T)\to\infty\) gives
\[
  t_0+N\le\frac{R+c}{\eta(T)}+2=o(T).
\]
Hence, for all sufficiently large \(T\), summing the completed
intervention yields
\[
  \Surplus(\widehat x_{<T})
    \ge\frac{g_{\rm mean}}2\sum_{t<T}\varepsilon_t
    =\frac{c\,g_{\rm mean}}{2\eta(T)}.
\]
If \(K=[m]\), then \(\delta=0\) and \(Q_h(0)=y_K\), so
we may take \(R=t_0=0\).
The payoff identity becomes
\(\widehat x_t^\top Ay_t-\val=\varepsilon_t\langle a,y_t\rangle\),
which also bounds total surplus by \(c\|a\|_\infty/\eta(T)\).
Together with the lower bound, this proves
\eqref{eq:constant-rate-selection-full-face}.
\end{proof}

For \(\eta(T)=\Theta(T^{-\rho})\), \(0<\rho<1\), the guarantee
is \(\Omega(T^\rho)\).
When \(K=[m]\), selection earns \(\Theta(T^\rho)\) even though
repeated max--min play has zero normal surplus
by \cref{prop:normal-surplus-bounds}.

\section{Proofs for Movement Exploitation: Alternating Play}
\label{sec:active}

This appendix proves \cref{thm:main-alternating-exploitation} from
\cref{sec:movement-exploitation}.
We verify the alternating construction, prove the response derivative
used in the main text, and derive the one-pair and cumulative surplus
bounds. The final subsection gives sufficient conditions for checking
the theorem's responsiveness and error hypotheses.
All notation follows \cref{sec:movement-exploitation}.

\subsection{Verification of the alternating construction}
\label{app:new-movement-construction}

We verify the pair asserted after \cref{ass:main-new-geometry}.
The same geometric condition appears in
\citet[Condition~1]{assos2024maximizing} and
\citet[Assumption~1]{su2026nocoin}.

\begin{assumption}[Restatement of the geometry in \cref{ass:main-new-geometry}]
\label{ass:new-movement-geometry}
There exist a max--min strategy \(x^\star\), distinct learner best
responses \(i,j\in\br(x^\star)\), and an action
\(\ell\in\supp(x^\star)\) such that \(A_{\ell i}\ne A_{\ell j}\).
\end{assumption}

Fix \(x^\star,i,j,\ell\) satisfying this condition and relabel
\(i,j\), if necessary, so that \(A_{\ell i}>A_{\ell j}\).
Following Equation~(40) of \citet{su2026nocoin}, take
\begin{equation}
  x'=(1-x_\ell^\star)x^\star+x_\ell^\star e_\ell,
  \qquad
  x''=(1+x_\ell^\star)x^\star-x_\ell^\star e_\ell.
  \label{eq:new-movement-strategies}
\end{equation}
The first vector is a convex combination of feasible strategies.
For the second, its \(\ell\)-th coordinate is
\((1+x_\ell^\star)x_\ell^\star-x_\ell^\star
=(x_\ell^\star)^2\ge0\), and every other coordinate is
\((1+x_\ell^\star)x_a^\star\ge0\).
Its coordinates sum to \((1+x_\ell^\star)-x_\ell^\star=1\).
Thus both vectors belong to \(\Delta_n\), and direct addition gives
\((x'+x'')/2=x^\star\).
These are the feasibility and averaging properties recorded in
Lemma~F.6 of \citet{su2026nocoin}.

Write \(v'=A^\top x'\), \(v''=A^\top x''\), and
\(v=A^\top x^\star=\val\mathbf1+\delta\), as in the main text.
Since \(i,j\in K=\br(x^\star)\), we have \(v_i=v_j=\val\).
The common \(x^\star\) component therefore cancels in the
coordinate difference:
\[
  v'_i-v'_j=x_\ell^\star(A_{\ell i}-A_{\ell j})>0,
  \qquad
  v''_i-v''_j=-(v'_i-v'_j)<0.
\]
The strict inequalities use \(x_\ell^\star>0\).
Thus the two plays distinguish the tied best responses in opposite
directions while preserving their max--min average.

\paragraph{Scores along the trajectory.}
Use the finite-horizon trajectory from
\cref{sec:movement-exploitation}: play \(x^\star\) for the first
\(\min\{2s_0,T\}\) rounds, then complete pairs \((x',x'')\),
and use \(x^\star\) for any unmatched final round.
Each complete pair adds
\(B^\top x'+B^\top x''=-2v\) to the cumulative scores,
exactly as two baseline plays would.
Consequently, before any complete pair \(s\ge s_0\) and after
its first play,
\begin{subequations}\label{eq:new-movement-scores}
\begin{align}
  S_{2s}&=-2s_0v-2(s-s_0)v=-2sv,
    \\
  S_{2s+1}&=S_{2s}-v'=-2sv-v'.
\end{align}
\end{subequations}
Substituting these scores into noisy FTRL gives
\begin{subequations}\label{eq:new-movement-responses}
\begin{align}
  y_{2s}&=Q_h(-2s\eta_{2s}v+\Xi_{2s}),
    \\
  y_{2s+1}&=Q_h(-2s\eta_{2s+1}v-\eta_{2s+1}v'+\Xi_{2s+1}).
\end{align}
\end{subequations}
The starting time changes the number of alternating pairs but
does not change these formulas for the pairs that are played.

\paragraph{Payoff of one pair.}
The averaging identity gives \(v''=2v-v'\).
Using \(\langle\mathbf1,y_{2s+1}\rangle=1\), we obtain
\begin{align}
  x'^\top Ay_{2s}+x''^\top Ay_{2s+1}-2\val
    &=\langle v',y_{2s}\rangle
       +\langle2v-v',y_{2s+1}\rangle-2\val
       \notag\\
    &=\langle v',y_{2s}-y_{2s+1}\rangle
       +2\langle v-\val\mathbf1,y_{2s+1}\rangle
       \notag\\
    &=\langle v',y_{2s}-y_{2s+1}\rangle
       +2\langle\delta,y_{2s+1}\rangle.
       \label{eq:new-movement-pair-payoff:c}
\end{align}
The off-face term is nonnegative because \(\delta\ge0\).
It therefore suffices to lower-bound the movement term
\(\langle v',y_{2s}-y_{2s+1}\rangle\), as claimed in the main text.

\subsection{Proof of the response derivative and payoff stability}
\label{app:new-movement-derivative}

The interpolation in the main text requires differentiating the
response even when its support changes.
The inverse-curvature hypothesis gives absolute continuity of the
response path, so the derivative can be integrated across those
changes.

\begin{assumption}[Inverse-curvature bound]
\label{ass:new-movement-regularity}
For a separable regularizer in \cref{def:separable-regularizer}, assume
\(C_\theta:=\sup_{0<u\le1}1/\theta''(u)<\infty\).
\end{assumption}

\begin{remark}
The formulas in \appref{app:preliminary-kernel-examples} give
\(C_\theta=1\) for negative entropy and Euclidean regularization,
and \(C_\theta=1/q\) for negative Tsallis regularization with
\(q\in(0,1)\). Thus all three regularizers satisfy this hypothesis.
\end{remark}

For an arbitrary response path \(y(r)=Q_h(u(r))\), use the
susceptibility weights from \cref{sec:movement-exploitation}:
\begin{equation}
  w_r(a)=
  \begin{cases}
    1/\theta''(y_a(r)),&y_a(r)>0,\\
    0,&y_a(r)=0,
  \end{cases}
  \qquad W_r=\sum_a w_r(a),\qquad \pi_r(a)=\frac{w_r(a)}{W_r}.
  \label{eq:new-movement-weights}
\end{equation}
At least one coordinate is positive, so \(W_r>0\).
We write \(I\sim\pi_r\) when forming the weighted covariance.
The normalization by \(W_r\) expresses the common adjustment
needed to keep the response coordinates summing to one.

\begin{lemma}[Response derivative and payoff stability]
\label{lem:new-movement-derivative}
Under \cref{ass:new-movement-regularity}, an absolutely continuous
score path \(u:[0,1]\to\mathbb R^m\) has an absolutely continuous
response \(y=Q_h\circ u\).
For every fixed \(q\in\mathbb R^m\), at almost every \(r\),
\begin{equation}
  \frac{d}{dr}\langle q,y(r)\rangle
    =W_r\Cov_{\pi_r}(q_I,\dot u_I(r)).
  \label{eq:new-movement-covariance}
\end{equation}
Moreover, for any scores \(u,\widetilde u\),
\begin{equation}
  |\langle q,Q_h(u)-Q_h(\widetilde u)\rangle|
    \le\frac{mC_\theta}{2}\spnorm{q}\spnorm{u-\widetilde u}.
  \label{eq:new-movement-lipschitz}
\end{equation}
\end{lemma}

\begin{proof}
We first justify differentiating the response, then eliminate the
KKT multiplier using the simplex constraint.
Integrating the resulting identity along a score segment will
give the finite-change bound.

\paragraph{Absolute continuity of the response.}
The bound \(\theta''\ge1/C_\theta\) makes \(h\)
\(1/C_\theta\)-strongly convex in the Euclidean norm on
\(\Delta_m\), including its boundary by continuity.
For \(y=Q_h(u)\) and \(\widetilde y=Q_h(\widetilde u)\), strong
convexity and optimality give
\[
  \begin{aligned}
    h(\widetilde y)-\langle u,\widetilde y\rangle
      &\ge h(y)-\langle u,y\rangle
          +\frac1{2C_\theta}\|y-\widetilde y\|_2^2,\\
    h(y)-\langle\widetilde u,y\rangle
      &\ge h(\widetilde y)-\langle\widetilde u,\widetilde y\rangle
          +\frac1{2C_\theta}\|y-\widetilde y\|_2^2.
  \end{aligned}
\]
Adding cancels the regularizer values and yields
\[
  C_\theta^{-1}\|y-\widetilde y\|_2^2
    \le\langle u-\widetilde u,y-\widetilde y\rangle
    \le\|u-\widetilde u\|_2\|y-\widetilde y\|_2.
\]
Thus \(Q_h\) is \(C_\theta\)-Lipschitz, which is the specialization
of \cref{prop:preliminary-strong-stability} used here.
Its composition with an absolutely continuous score path is
therefore absolutely continuous.

\paragraph{Differentiating through changes of support.}
Choose the KKT multiplier with equality on positive coordinates,
using the endpoint convention in \appref{app:preliminary-kkt}.
Near each parameter value, some coordinate remains positive by
continuity. On such a neighborhood,
\(\lambda(r)=\theta'(y_a(r))-u_a(r)\) for that coordinate.
Since its probability stays away from zero locally and \(\theta'\)
is continuously differentiable there, \(\lambda\) is locally
absolutely continuous.
At almost every parameter value, differentiation on each positive
coordinate gives
\[
  \dot y_a(r)=w_r(a)(\dot u_a(r)+\dot\lambda(r)).
\]
For an inactive coordinate, its derivative is zero almost everywhere
on its zero set: this holds for any nonnegative absolutely
continuous function. With the zero weights in
\eqref{eq:new-movement-weights}, the same identity consequently
holds almost everywhere for all coordinates, even when the support
changes.

Differentiate \(\sum_a y_a(r)=1\) and sum these identities:
\[
  0=\sum_a\dot y_a(r)
    =\sum_a w_r(a)\dot u_a(r)+W_r\dot\lambda(r),
  \qquad
  \dot\lambda(r)=-\sum_a\pi_r(a)\dot u_a(r).
\]
The common multiplier subtracts the susceptibility-weighted average
score change. Multiplying each coordinate derivative by \(q_a\)
and summing now gives
\[
  \begin{aligned}
    \frac{d}{dr}\langle q,y(r)\rangle
      &=\sum_a w_r(a)q_a\dot u_a(r)
        -\frac{\bigl(\sum_a w_r(a)q_a\bigr)
                    \bigl(\sum_b w_r(b)\dot u_b(r)\bigr)}{W_r}\\
      &=W_r\Cov_{\pi_r}(q_I,\dot u_I(r)).
  \end{aligned}
\]
This proves \eqref{eq:new-movement-covariance}.

\paragraph{Integrating a finite score change.}
Take the segment \(u(r)=u+r(\widetilde u-u)\), with constant
derivative \(d=\widetilde u-u\).
To bound its covariance by relative score differences, expand it as
\[
  W_r\Cov_{\pi_r}(q_I,d_I)
    =\frac1{2W_r}\sum_{a,b}w_r(a)w_r(b)
       (q_a-q_b)(d_a-d_b).
\]
Each difference is bounded by the span of its vector.
Since \(\sum_{a,b}w_r(a)w_r(b)=W_r^2\) and
\(W_r\le mC_\theta\),
\[
  |W_r\Cov_{\pi_r}(q_I,d_I)|
    \le\frac{W_r}{2}\spnorm{q}\spnorm{d}
    \le\frac{mC_\theta}{2}\spnorm{q}\spnorm{d}.
\]
Integrating the derivative over \([0,1]\) proves
\eqref{eq:new-movement-lipschitz}.
\end{proof}

\subsection{Derivation of the one-pair surplus bound}
\label{app:new-movement-pair}

We now establish the movement inequality used in
\cref{sec:movement-exploitation}.
Fix a coupling
\(\widehat\Xi_s^0\stackrel d=\Xi_{2s}\) and
\(\widehat\Xi_s^1\stackrel d=\Xi_{2s+1}\).
The optimizer's trajectory and hence its cumulative scores are
deterministic, so each expected round payoff depends only on the
corresponding noise marginal. The coupling need not reproduce the
actual temporal dependence.

To apply \cref{lem:new-movement-derivative}, change the deterministic
score while holding noise fixed, then replace the noise at the odd
endpoint. For fixed \(\xi\), the interpolation from the main text is
\begin{subequations}\label{eq:new-movement-interpolation}
\begin{align}
  \bar\eta_s(r)&=(1-r)\eta_{2s}+r\eta_{2s+1},
    \\
  u_s(r;\xi)&=-2s\bar\eta_s(r)v-r\eta_{2s+1}v'+\xi,
    \\
  y_s(r;\xi)&=Q_h(u_s(r;\xi)),\qquad r\in[0,1].
\end{align}
\end{subequations}
Its endpoints agree with \eqref{eq:new-movement-responses} when
the appropriate noise marginal is inserted.
Following the interpolation argument of \citet{su2026nocoin},
absolute continuity of the response gives
\begin{equation}
  \langle v',y_s(0;\xi)\rangle-\langle v',y_s(1;\xi)\rangle
    =-\int_0^1\frac{d}{dr}\langle v',y_s(r;\xi)\rangle\,dr.
  \label{eq:new-movement-integral}
\end{equation}
This represents the movement payoff with the noise held fixed.

Along the path with \(\xi=\widehat\Xi_s^0\), use the weights
\(w_{s,r},W_{s,r},\pi_{s,r}\) from \eqref{eq:new-movement-weights}.
The score derivative is
\[
  \frac{d}{dr}u_s(r;\widehat\Xi_s^0)
    =2s(\eta_{2s}-\eta_{2s+1})v-\eta_{2s+1}v'.
\]
The first term rescales the preceding \(2s\) rounds' payoffs;
the second is the score change caused by the new play \(x'\).
Applying \eqref{eq:new-movement-covariance} with \(q=v'\) gives
\begin{align}
  \frac{d}{dr}\langle v',y_s(r;\widehat\Xi_s^0)\rangle
    &=W_{s,r}\Cov_{\pi_{s,r}}\!\left(
        v'_I,\frac{d}{dr}u_{s,I}(r;\widehat\Xi_s^0)\right)
        \notag\\
    &=W_{s,r}\Cov_{\pi_{s,r}}\!\left(
        v'_I,2s(\eta_{2s}-\eta_{2s+1})v_I-\eta_{2s+1}v'_I\right)
        \notag\\
    &=2s(\eta_{2s}-\eta_{2s+1})W_{s,r}\Cov_{\pi_{s,r}}(v'_I,v_I)
       -\eta_{2s+1}W_{s,r}\Cov_{\pi_{s,r}}(v'_I,v'_I)
        \notag\\
    &=2s(\eta_{2s}-\eta_{2s+1})W_{s,r}\Cov_{\pi_{s,r}}(v'_I,v_I)
       -\eta_{2s+1}W_{s,r}\Var_{\pi_{s,r}}(v'_I).
        \label{eq:new-movement-path-derivative:d}
\end{align}
The minus sign in \eqref{eq:new-movement-integral} turns the variance
term into positive movement gain. The covariance can have either
sign, so its absolute value bounds the possible schedule loss.
The main-text quantities are therefore
\begin{align*}
  \mathcal K_s&=\mathbb E\int_0^1
    W_{s,r}\Var_{\pi_{s,r}}(v'_I)\,dr,
    \\
  \mathcal D_s&=\mathbb E\int_0^1
    W_{s,r}|\Cov_{\pi_{s,r}}(v'_I,v_I)|\,dr.
\end{align*}
Replacing the noise at the odd endpoint contributes
\begin{equation*}
  \mathcal N_s=\mathbb E\left|
    \left\langle v',Q_h(u_s(1;\widehat\Xi_s^0))
                    -Q_h(u_s(1;\widehat\Xi_s^1))\right\rangle
    \right|.
\end{equation*}
All three quantities are finite: \(W_{s,r}\le mC_\theta\),
the payoff vectors are fixed, and the absolute payoff difference in
\(\mathcal N_s\) is at most \(\spnorm{v'}\).
Thus no moment assumption on the noise is needed for this bound.
The choice of coupling may affect \(\mathcal N_s\).

\begin{theorem}[Expected movement from one pair]
\label{thm:new-movement-pair}
Suppose \cref{ass:new-movement-geometry,ass:new-movement-regularity}
hold and use the alternating trajectory \(x^{\rm alt}_{<T}\).
For any positive learning rates and any coupling of the adjacent
noise marginals, every complete pair satisfies
\begin{equation}
  \mathbb E\langle v',y_{2s}-y_{2s+1}\rangle
    \ge\eta_{2s+1}\mathcal K_s
      -2s|\eta_{2s}-\eta_{2s+1}|\mathcal D_s-\mathcal N_s.
  \label{eq:new-movement-pair-bound}
\end{equation}
Consequently, for every horizon \(T\ge1\),
\begin{equation}
  \Surplus(x^{\rm alt}_{<T})
    \ge\sum_{s=s_0}^{\lfloor T/2\rfloor-1}
      \left[\eta_{2s+1}\mathcal K_s
       -2s|\eta_{2s}-\eta_{2s+1}|\mathcal D_s-\mathcal N_s\right].
  \label{eq:new-movement-total-bound}
\end{equation}
\end{theorem}

\begin{proof}
We first integrate the deterministic score change under a common
noise realization, then account for the endpoint noise replacement.
Substitute \eqref{eq:new-movement-path-derivative:d} into
\eqref{eq:new-movement-integral} with \(\xi=\widehat\Xi_s^0\).
Taking expectations gives
\[
  \begin{aligned}
    &\mathbb E\langle v',y_s(0;\widehat\Xi_s^0)
                           -y_s(1;\widehat\Xi_s^0)\rangle\\
    &\quad=\eta_{2s+1}\mathcal K_s
       -2s(\eta_{2s}-\eta_{2s+1})
          \mathbb E\int_0^1W_{s,r}\Cov_{\pi_{s,r}}(v'_I,v_I)\,dr\\
    &\quad\ge\eta_{2s+1}\mathcal K_s
       -2s|\eta_{2s}-\eta_{2s+1}|\mathcal D_s.
  \end{aligned}
\]
Bounded inverse curvatures and fixed payoff vectors justify exchanging
expectations and integrals. The absolute covariance bound allows
arbitrary positive schedules, including schedules that increase.

Under the chosen coupling, add and subtract the common-noise odd
endpoint to write the actual expected movement as
\[
  \begin{aligned}
    \mathbb E\langle v',y_{2s}-y_{2s+1}\rangle
      &=\mathbb E\langle v',y_s(0;\widehat\Xi_s^0)
                              -y_s(1;\widehat\Xi_s^0)\rangle\\
      &\quad+\mathbb E\langle v',y_s(1;\widehat\Xi_s^0)
                              -y_s(1;\widehat\Xi_s^1)\rangle.
  \end{aligned}
\]
The prescribed marginals justify this equality. The second expectation
is bounded below by \(-\mathcal N_s\), proving
\eqref{eq:new-movement-pair-bound} without any temporal-independence
assumption.

Now sum over the complete pairs and use
\eqref{eq:new-movement-pair-payoff:c}.
The off-face terms are nonnegative.
Every initial or unmatched round plays \(x^\star\), whose payoff
is at least \(\val\) against every learner response, so those
rounds also have nonnegative surplus.
This proves \eqref{eq:new-movement-total-bound}.
If there are no complete pairs, the sum is empty and the same
baseline argument gives the bound.
\end{proof}

\subsection{Proof of the alternating-surplus theorem}
\label{app:new-movement-surplus}

The remaining step in \cref{thm:main-alternating-exploitation} is
to choose a finite starting time after which the movement gain
dominates both errors. The quantities below are exactly those in
the main-text theorem.

\begin{theorem}[Formal statement of \cref{thm:main-alternating-exploitation}]
\label{thm:new-movement-surplus}
Under the setup of \cref{thm:new-movement-pair}, suppose there are
\(\kappa>0\) and adjacent marginal couplings such that
\begin{equation}
  \liminf_{s\to\infty}\mathcal K_s\ge\kappa,
  \qquad
  2s|\eta_{2s}-\eta_{2s+1}|\mathcal D_s+\mathcal N_s
    =o(\eta_{2s+1}).
  \label{eq:new-movement-positive-condition}
\end{equation}
There is a finite \(s_0\), independent of \(T\), such that
\begin{equation}
  \Surplus(x^{\rm alt}_{<T})
    \ge\frac\kappa2
      \sum_{s=s_0}^{\lfloor T/2\rfloor-1}\eta_{2s+1}
    \qquad\text{for every }T\ge1.
  \label{eq:new-movement-surplus}
\end{equation}
\end{theorem}

\begin{proof}
To leave a margin for the errors, take a late enough starting time
that \(\mathcal K_s\ge3\kappa/4\) and the combined error is at
most \(\kappa\eta_{2s+1}/4\) whenever \(s\ge s_0\).
The liminf and little-oh conditions in
\eqref{eq:new-movement-positive-condition} make such a finite
\(s_0\) possible.
Then each complete pair after the starting time has lower bound
\[
  \eta_{2s+1}\mathcal K_s
    -2s|\eta_{2s}-\eta_{2s+1}|\mathcal D_s-\mathcal N_s
    \ge\frac{3\kappa}{4}\eta_{2s+1}
       -\frac\kappa4\eta_{2s+1}
    =\frac\kappa2\eta_{2s+1}.
\]
This explains the fractions: the lower bound on responsiveness and
the reserved error budget leave half of \(\kappa\) as net gain.
Substitution into \eqref{eq:new-movement-total-bound} proves the
claimed inequality for every horizon, including horizons with no
complete pair.

Finally, delaying alternation does not alter the response formulas
for the later pairs. Each complete pair and two baseline plays
produce the same cumulative score increment \(-2v\), so
\eqref{eq:new-movement-scores} holds with the chosen \(s_0\).
Thus the starting time is independent of \(T\), completing the
proof of \cref{thm:main-alternating-exploitation}.
\end{proof}

\begin{remark}[Noiseless learning with a constant step size]
\label{rem:new-movement-noiseless-constant}
For \(\eta_t\equiv\eta>0\), the schedule-drift contribution is
identically zero, even if \(\mathcal D_s\) is not.
With no noise, \(\widehat\Xi_s^0=\widehat\Xi_s^1=0\), so
\(\mathcal N_s=0\) as well.
The error condition in \eqref{eq:new-movement-positive-condition}
therefore holds automatically; only a positive lower bound on
\(\mathcal K_s\) remains to be checked.

\smallskip
Under the sufficiently small step-size condition of
\citet{su2026nocoin}, their Lemma~F.9 and Corollary~F.10 keep
the distinguished coordinates uniformly away from zero and one
along the noiseless interpolation.
Their curvatures are consequently bounded above by a fixed
\(M<\infty\), uniformly in \(s\), \(r\in[0,1]\), and all
sufficiently small \(\eta\).
Their Lemma~F.8 gives
\[
  \mathcal K_s=\int_0^1W_{s,r}\Var_{\pi_{s,r}}(v'_I)\,dr
    \ge\frac{(v'_i-v'_j)^2}{2M}>0.
\]
The bound holds for every pair. Taking \(s_0=0\) directly in
\eqref{eq:new-movement-total-bound} yields
\[
  \Surplus(x^{\rm alt}_{<T})
    \ge\frac{\eta(v'_i-v'_j)^2}{2M}
         \left\lfloor\frac T2\right\rfloor
    =\Omega(\eta T).
\]
\smallskip
Fixed \(\eta\) gives \(\Omega(T)\); a constant rate \(\eta(T)\)
gives \(\Omega(\eta(T)T)\), as in \citet{su2026nocoin}.
\end{remark}

The theorem concerns expected payoffs; it does not require any
particular visit frequencies along individual response trajectories.
The power-law rate stated after the main-text theorem follows by
evaluating the odd-round learning-rate sum.

\begin{corollary}[Power-law alternating surplus]
\label{cor:new-movement-power-law}
If \cref{thm:new-movement-surplus} applies and
\(\eta_t=\eta_0(t+1)^{-\rho}\), with \(\eta_0>0\) and
\(0<\rho<1\), then
\(\Surplus(x^{\rm alt}_{<T})=\Omega(T^{1-\rho})\).
\end{corollary}

\begin{proof}
Put \(N=\lfloor T/2\rfloor\).
The odd-round rate is \(\eta_{2s+1}=\eta_0[2(s+1)]^{-\rho}\).
The decreasing function \(u^{-\rho}\) satisfies
\[
  \int_1^{N+1}u^{-\rho}\,du
    \le\sum_{j=1}^N j^{-\rho}
    \le1+\int_1^N u^{-\rho}\,du.
\]
Evaluating the integrals gives
\(\sum_{j=1}^N j^{-\rho}=N^{1-\rho}/(1-\rho)+O(1)\).
Removing the fixed initial segment changes only the additive
constant. Hence
\begin{equation}
  \begin{aligned}
    \sum_{s=s_0}^{N-1}\eta_{2s+1}
      &=\eta_0 2^{-\rho}\sum_{j=s_0+1}^N j^{-\rho}=\frac{\eta_0 2^{-\rho}}{1-\rho}N^{1-\rho}+O(1)\\
      &=\frac{\eta_0}{2(1-\rho)}T^{1-\rho}+O(1).
  \end{aligned}
  \label{eq:new-movement-odd-sum}
\end{equation}
Applying \eqref{eq:new-movement-surplus} proves the corollary.
\end{proof}

\subsection{Sufficient conditions for positive movement}
\label{app:new-movement-sufficient}

To apply \cref{thm:new-movement-surplus}, we need a positive
lower bound on \(\mathcal K_s\) and errors smaller than
\(\eta_{2s+1}\).
The variance contribution is positive whenever the two useful
coordinates \(i,j\) have inverse curvatures bounded away from zero.
We first establish this for the restricted response near equal
scores. The proposition below then transfers the bound to the
full response on an event of positive probability.

We use the restricted choice map \(Q_{h,K}\) with regularizer
\(\sum_{a\in K}\theta(y_a)\).
For finite \(\theta(0)\), this agrees with
\eqref{eq:restricted-choice}: all omitted terms outside \(K\)
are constant.

\begin{lemma}[Restricted response near equal scores]
\label{lem:new-movement-radius}
Let \(h\) be a separable regularizer as in
\cref{def:separable-regularizer}, satisfying
\cref{ass:new-movement-regularity}.
There exists \(a_{\rm alt}>0\) such that, for every score vector
\(u_K\),
\begin{equation}
  \spnorm{u_K}\le2a_{\rm alt}
  \quad\Longrightarrow\quad
  Q_{h,K}(u_K)_i,\ Q_{h,K}(u_K)_j
       \in\left[\frac{3}{4k},1-\frac{3}{4k}\right].
  \label{eq:new-movement-radius}
\end{equation}
\end{lemma}

\begin{proof}
At zero scores, the restricted response minimizes
\(\sum_{a\in K}\theta(y_a)\).
Jensen's inequality gives
\[
  \sum_{a\in K}\theta(y_a)\ge k\theta(1/k),
\]
with equality only at the uniform vector by strict convexity.
Thus \(Q_{h,K}(0)_a=1/k\) for every \(a\in K\).

We choose the score radius so that the response cannot move
more than \(1/(4k)\) from this vector.
Apply \eqref{eq:new-movement-lipschitz} on the simplex over
\(K\), with \(q=e_a\) and comparison scores zero. For
\(a\in\{i,j\}\),
\[
  \left|Q_{h,K}(u_K)_a-\frac1k\right|
       \le\frac{kC_\theta}{2}\spnorm{u_K}.
\]
Taking \(a_{\rm alt}=1/(4k^2C_\theta)\) makes the right-hand
side at most \(1/(4k)\) whenever
\(\spnorm{u_K}\le2a_{\rm alt}\). Therefore
\[
  Q_{h,K}(u_K)_i,\ Q_{h,K}(u_K)_j
       \in\left[\frac{3}{4k},\frac{5}{4k}\right]
       \subseteq\left[\frac{3}{4k},1-\frac{3}{4k}\right],
\]
where the inclusion uses \(k\ge2\).
\end{proof}

The interval in \eqref{eq:new-movement-radius} leaves room
for the error incurred when passing to the full response.
Define
\[
  \underline y=\frac1{2k},\qquad
  M_{\rm alt}
       =\sup_{u\in[\underline y,1-\underline y]}\theta''(u)
       <\infty.
\]
The curvature bound is finite because the interval is compact
and lies inside \((0,1)\).
If the full and restricted responses differ by at most
\(1/(4k)\) in each useful coordinate, both full coordinates
lie in \([\underline y,1-\underline y]\).
Their inverse curvatures are then at least \(1/M_{\rm alt}\).

Controlling the error from a changing learning rate requires
an additional boundary property. Off-face probabilities
must have small inverse curvatures as they approach zero.
To express this property, define
\[
  \omega(z)=\sup_{0<u\le z}\frac1{\theta''(u)},
  \qquad 0<z\le1.
\]
For a steep regularizer, we require \(\omega(z)\to0\) as
\(z\downarrow0\).
For a non-steep regularizer, the proof instead makes the
off-face coordinates exactly zero.

\begin{proposition}[Conditions ensuring the alternating rate]
\label{prop:new-movement-sufficient}
Suppose \cref{ass:new-movement-geometry,ass:new-movement-regularity}
hold, and fix the alternating trajectory \(x^{\rm alt}_{<T}\).
Assume the following conditions.
\begin{enumerate}
  \item If the regularizer is steep, \(\omega(z)\to0\) as
  \(z\downarrow0\).
  \item The positive learning rates are nonincreasing and
  satisfy
  \begin{equation}
    \eta_t\to0,\qquad t\eta_t\to\infty,\qquad
    t(\eta_t-\eta_{t+1})=O(\eta_{t+1}).
    \label{eq:new-movement-schedule}
  \end{equation}
  \item The even-round relative noise is tight and has a
  uniformly positive probability of lying in the neighborhood
  from \cref{lem:new-movement-radius}:
    \begin{align*}
      &\lim_{R\to\infty}\limsup_{s\to\infty}
           \Pr(\spnorm{\Xi_{2s}}>R)=0,
        \\
      &\liminf_{s\to\infty}
           \Pr(\spnorm{\Xi_{2s}}\le a_{\rm alt})
                 \ge p_{\rm alt}>0.
    \end{align*}
  \item There are couplings of the adjacent marginals such that
  \begin{equation}
    \mathbb E\spnorm{\widehat\Xi_s^0-\widehat\Xi_s^1}
                     =o(\eta_{2s+1}).
    \label{eq:new-movement-coupling}
  \end{equation}
\end{enumerate}
Then
\begin{equation}
  \begin{aligned}
    \liminf_{s\to\infty}\mathcal K_s
       &\ge\frac{p_{\rm alt}(x_\ell^\star(A_{\ell i}-A_{\ell j}))^2}
                    {2M_{\rm alt}},\\
    \mathcal D_s&=o(1),\qquad
    \mathcal N_s=o(\eta_{2s+1}).
  \end{aligned}
  \label{eq:new-movement-three-bounds}
\end{equation}
For some finite \(s_0\), independent of \(T\),
\begin{equation}
  \Surplus(x^{\rm alt}_{<T})
     \ge\frac{p_{\rm alt}(x_\ell^\star(A_{\ell i}-A_{\ell j}))^2}
                   {8M_{\rm alt}}
             \sum_{s=s_0}^{\lfloor T/2\rfloor-1}\eta_{2s+1}
     \qquad\text{for every }T\ge1.
  \label{eq:new-movement-sufficient-surplus}
\end{equation}
\end{proposition}

\begin{proof}
We first show that accumulating payoff gaps remove probability
and inverse curvature outside \(K\). This makes the schedule
drift negligible.
On the small-noise event, we then compare the full response
with the restricted response and obtain a positive variance
contribution from \(i,j\).
The coupling condition controls the remaining noise error.

\paragraph{The deterministic gap removes off-face probability.}
Fix \(R<\infty\) and work on the event
\(\spnorm{\widehat\Xi_s^0}\le R\).
For \(a\in K\), \(b\notin K\), and \(r\in[0,1]\),
the scores in \eqref{eq:new-movement-interpolation} satisfy
\[
  \begin{aligned}
    u_{s,a}(r;\widehat\Xi_s^0)-u_{s,b}(r;\widehat\Xi_s^0)
       &=2s\bar\eta_s(r)\delta_b
          -r\eta_{2s+1}(v'_a-v'_b)
          +\widehat\Xi_{s,a}^0-\widehat\Xi_{s,b}^0\\
       &\ge2s\eta_{2s+1}\delta_b
          -\eta_{2s+1}\spnorm{v'}-R.
  \end{aligned}
\]
The inequality uses monotonicity of the learning rates,
which gives \(\bar\eta_s(r)\ge\eta_{2s+1}\).
If \(K\ne[m]\), set
\[
  L_s(R)=2s\eta_{2s+1}\delta_{\min}
                 -\eta_{2s+1}\spnorm{v'}-R.
\]
This common lower bound tends to infinity by
\eqref{eq:new-movement-schedule}.
The separation is uniform in \(r\) and in the noise vector
on the bounded-noise event.

Write \(y=y_s(r;\widehat\Xi_s^0)\) for this calculation.
At least one coordinate \(a\in K\) satisfies \(y_a\le1/k\),
since \(\sum_{a\in K}y_a\le1\).
The optimality conditions from \appref{app:preliminary-kkt}
give
\(u_{s,a}+\lambda\le\theta'(1/k)\): at a positive coordinate
this follows from \(u_{s,a}+\lambda=\theta'(y_a)\), and at
an inactive coordinate it follows from the inequality
\(u_{s,a}+\lambda\le\theta'(0^+)\).
Thus, for each \(b\notin K\),
\[
  u_{s,b}+\lambda\le\theta'(1/k)-L_s(R),
  \qquad
  y_b\le\phi\bigl(\theta'(1/k)-L_s(R)\bigr)\longrightarrow0.
\]
This is the same score-gap estimate used in
\cref{lem:pointwise-noisy-ftrl-bound}.
For a non-steep kernel, the argument eventually lies below
\(\theta'(0^+)\), so every off-face coordinate is exactly zero.
For a steep kernel, its weight is at most \(\omega(y_b)\),
which tends to zero by condition (i).
If \(K=[m]\), there are no off-face coordinates; the
probability and weight conclusions hold without a gap estimate.

\paragraph{Vanishing off-face weights remove schedule drift.}
For each fixed \(R\), the preceding calculation makes
\(\sum_{b\notin K}w_{s,r}(b)\) tend to zero uniformly in
\(r\) on \(\{\spnorm{\widehat\Xi_s^0}\le R\}\).
On the complementary event, it is at most \(mC_\theta\).
Splitting the expectation over these two events therefore gives
\[
  \limsup_{s\to\infty}\mathbb E\int_0^1
               \sum_{b\notin K}w_{s,r}(b)\,dr
     \le mC_\theta\limsup_{s\to\infty}
                  \Pr(\spnorm{\Xi_{2s}}>R).
\]
The probability on the right uses the even-round marginal,
which is the law of \(\widehat\Xi_s^0\).
Letting \(R\to\infty\) and applying tightness proves
\begin{equation*}
  \mathbb E\int_0^1\sum_{b\notin K}w_{s,r}(b)\,dr
       \longrightarrow0.
\end{equation*}
This controls the schedule term because \(v\) is constant
on \(K\). Indeed, expanding the covariance gives
\[
  W_{s,r}\Cov_{\pi_{s,r}}(v'_I,v_I)
    =\frac1{2W_{s,r}}\sum_{a,b}w_{s,r}(a)w_{s,r}(b)
                          (v'_a-v'_b)(v_a-v_b).
\]
Every summand with \(a,b\in K\) vanishes.
In the remaining terms, at least one index lies outside
\(K\); counting both possibilities bounds their total weight
by \(2W_{s,r}\sum_{b\notin K}w_{s,r}(b)\).
Consequently,
\[
  W_{s,r}|\Cov_{\pi_{s,r}}(v'_I,v_I)|
      \le\spnorm{v'}\spnorm{v}
                          \sum_{b\notin K}w_{s,r}(b).
\]
Taking expectations and integrating gives \(\mathcal D_s=o(1)\).
Condition (ii) also gives
\(2s(\eta_{2s}-\eta_{2s+1})=O(\eta_{2s+1})\).
Thus the schedule error in
\eqref{eq:new-movement-pair-bound} is \(o(\eta_{2s+1})\).

\paragraph{Small relative noise keeps two coordinates responsive.}
Consider the event
\(G_s=\{\spnorm{\widehat\Xi_s^0}\le a_{\rm alt}\}\).
Because \(v\) is constant on \(K\), only the alternating
perturbation and the noise contribute to the span of the
restricted scores:
\[
  \spnorm{u_s(r;\widehat\Xi_s^0)_K}
       \le a_{\rm alt}+\eta_{2s+1}\spnorm{v'_K}
       \le2a_{\rm alt}
\]
for all sufficiently large \(s\), uniformly in \(r\).
The restricted response therefore satisfies
\eqref{eq:new-movement-radius}.
It remains to justify that the full response is close enough
to inherit its coordinate bounds.

Let \(\beta=\sum_{a\in K}y_{s,a}\).
The off-face probability estimate with \(R=a_{\rm alt}\)
shows that \(\beta\to1\) uniformly on \(G_s\) and along
the interpolation.
By separability, holding the other coordinates fixed makes
\(y_{s,K}\) a minimizer of
\[
  \sum_{a\in K}\theta(z_a)-\langle u_{s,K},z\rangle,
  \qquad z_a\ge0,\quad\sum_{a\in K}z_a=\beta.
\]
Subtracting a common score on \(K\) leaves this minimizer
unchanged, since every feasible vector has mass \(\beta\).
After centering, the restricted scores lie in a compact set.

To compare the problems at masses \(\beta\) and one, take
any sequence with \(\beta\to1\) and centered scores tending
to \(\bar u_K\).
Pass to a subsequence on which \(y_{s,K}\to p\); then
\(p\in\Delta_K\).
For any \(z\in\Delta_K\) positive on \(K\), the vector
\(\beta z\) is feasible for the problem at mass \(\beta\).
Optimality against this vector, followed by continuity of
\(\theta\) on \([0,1]\), gives
\[
  \sum_{a\in K}\theta(p_a)-\langle\bar u_K,p\rangle
      \le\sum_{a\in K}\theta(z_a)-\langle\bar u_K,z\rangle.
\]
For a boundary competitor, take mixtures with the uniform
vector and let the mixing weight go to zero.
Continuity extends the same inequality to every
\(z\in\Delta_K\).
Hence \(p=Q_{h,K}(\bar u_K)\) by uniqueness of the restricted
minimizer.

This comparison is uniform over the bounded set of relative
scores. Otherwise, a sequence staying a fixed positive
distance from its restricted response would have convergent
subsequences as above.
Continuity of \(Q_{h,K}\), as in \cref{lem:continuity},
would then force both responses to the same limit, a
contradiction.
For all large \(s\), the full and restricted coordinates
are therefore within \(1/(4k)\), uniformly on \(G_s\)
and in \(r\).
The margin chosen after \cref{lem:new-movement-radius} now gives
\[
  y_{s,i},y_{s,j}\in[\underline y,1-\underline y],
  \qquad w_{s,r}(i),w_{s,r}(j)\ge1/M_{\rm alt}.
\]

To turn these coordinate bounds into a variance bound, use
the variational formula
\[
  W_{s,r}\Var_{\pi_{s,r}}(v'_I)
       =\min_{c\in\mathbb R}\sum_a w_{s,r}(a)(v'_a-c)^2.
\]
Deleting every term except \(i,j\) can only decrease this
minimum. For those two terms, the minimizing center is
\[
  c=\frac{w_{s,r}(i)v'_i+w_{s,r}(j)v'_j}
          {w_{s,r}(i)+w_{s,r}(j)}.
\]
Substitution yields
\[
  \begin{aligned}
    W_{s,r}\Var_{\pi_{s,r}}(v'_I)
       &\ge\frac{w_{s,r}(i)w_{s,r}(j)}{w_{s,r}(i)+w_{s,r}(j)}
                     (v'_i-v'_j)^2\\
       &=\frac{(v'_i-v'_j)^2}
                   {\theta''(y_{s,i})+\theta''(y_{s,j})}\\
       &\ge\frac{(x_\ell^\star(A_{\ell i}-A_{\ell j}))^2}
                    {2M_{\rm alt}}.
  \end{aligned}
\]
This lower bound holds for every \(r\) on \(G_s\), and
the variance is nonnegative elsewhere. Therefore
\[
  \mathcal K_s
       \ge\Pr(G_s)
          \frac{(x_\ell^\star(A_{\ell i}-A_{\ell j}))^2}
               {2M_{\rm alt}}.
\]
Since \(\Pr(G_s)=\Pr(\spnorm{\Xi_{2s}}\le a_{\rm alt})\),
condition (iii) gives the lower bound on
\(\liminf_s\mathcal K_s\) in
\eqref{eq:new-movement-three-bounds}.

\paragraph{The coupling controls the adjacent noise error.}
At the odd endpoint, the two responses used in
\(\mathcal N_s\) have score difference exactly
\(\widehat\Xi_s^0-\widehat\Xi_s^1\).
Applying \eqref{eq:new-movement-lipschitz} with \(q=v'\) gives
\[
  \mathcal N_s
       \le\frac{mC_\theta}{2}\spnorm{v'}\,
                \mathbb E\spnorm{\widehat\Xi_s^0-\widehat\Xi_s^1}
       =o(\eta_{2s+1})
\]
by condition (iv).
Together with the schedule estimate, this verifies that
both errors are smaller than the movement term.
Set
\[
  \kappa=\frac{p_{\rm alt}(x_\ell^\star(A_{\ell i}-A_{\ell j}))^2}
                   {4M_{\rm alt}}.
\]
This is half the lower bound on \(\liminf_s\mathcal K_s\),
so \cref{thm:new-movement-surplus} applies.
Its coefficient \(\kappa/2\) is exactly the coefficient in
\eqref{eq:new-movement-sufficient-surplus}.
\end{proof}

\begin{remark}[Scope of the assumptions]
For entropy, \(\omega(z)=z\); for negative Tsallis with
parameter \(0<q<1\), \(\omega(z)=z^{2-q}/q\).
Both satisfy the boundary condition in (i).
Euclidean regularization is non-steep, so that condition
places no additional restriction on it.

\smallskip
The schedules \(\eta_t=\eta_0(t+1)^{-\rho}\) satisfy (ii)
for every \(0<\rho<1\).
The first two limits make individual perturbations small
while the cumulative payoff gaps grow; the last bound
controls the change between adjacent rates.
The requirement \(\eta_t\to0\) belongs to this sufficient
criterion. For a constant rate \(\eta_t=\eta\), the schedule
term in \eqref{eq:new-movement-pair-bound} is zero.
With no noise, both errors are zero, and a positive lower
bound on \(\mathcal K_s\) gives \(\Omega(T)\) surplus
by \cref{thm:new-movement-pair}.

\smallskip
The two conditions in (iii) serve different purposes:
tightness controls the tails, while the probability bound
supplies an event on which the useful coordinates remain
responsive. Tightness alone need not supply this event.
The noise need not vanish; a fixed centered Gaussian law
satisfies both requirements.

\smallskip
Condition (iv) requires an adjacent-marginal coupling whose
expected span difference is \(o(\eta_{2s+1})\).
Convergence of the marginal laws alone does not ensure this
rate. Identical adjacent marginals can be coupled equally,
giving zero error regardless of their actual temporal
dependence.
For cumulative Gaussian noise, \appref{app:movement-noise-column}
verifies conditions (ii)--(iv) for power-law schedules with
\(1/2\le\rho<1\).
\end{remark}

\section{Noise Models and Exploitation Rates}
\label{app:noise-models-rates}

This appendix supplies the noise-specific verifications used in
\cref{sec:noise-examples}. We apply the response results of
\cref{sec:learner-dynamics-maxmin} and the surplus bounds of
\cref{subsec:normal-exploitation,subsec:main-new-selection,sec:movement-exploitation}
to obtain each entry of \cref{tab:main-noise-examples}. The proofs
below identify the hypothesis needed for each column, verify it for
the indicated noise model, and evaluate the resulting rate.

We retain the setup of \cref{sec:noise-examples}:
\(h\) is negative entropy, Euclidean, or negative Tsallis with
\(q\in(0,1)\), and
\[
  \eta_t=\eta_0(t+1)^{-\rho},
  \qquad \eta_0>0,\quad 0<\rho<1.
\]
Condition~\eqref{eq:finite-variance-noise} assumes
\[
  z\sim\mathcal D,
  \qquad \mathbb E[z]=0,
  \qquad \mathbb E\lVert z\rVert_2^2<\infty.
\]
Fresh noise has \(Z_t=z_t\), with common marginal law
\(\mathcal D\) and arbitrary temporal dependence. Cumulative noise
has \(Z_t=\sum_{\tau<t}z_\tau\), with i.i.d.\ increments of law
\(\mathcal D\). The game, regularizer, noise law, and chosen
max--min strategy \(x^\star\) are fixed independently of \(T\);
the constants below may depend on them.

\begin{remark}[Examples of finite-variance noise]
\label{rem:finite-variance-examples}
Condition~\eqref{eq:finite-variance-noise} covers the following
increments, in the order of heavier tails:
\begin{itemize}
  \item sub-Gaussian increments, which include centered Gaussian,
  Rademacher, and centered bounded distributions;
  \item Laplace increments, which have an exponential tail and
  variance \(2b^2\) for scale \(b\);
  \item Student-\(t\) increments with \(\nu>2\) degrees of freedom,
  which have a polynomial tail and variance \(\nu/(\nu-2)\).
\end{itemize}
Each of them can be used in the fresh model and in the cumulative
i.i.d.\ model.
\end{remark}

For the cumulative rows, let \(\Sigma:=\Cov(z)\succ0\) and
\(G\sim\mathcal N(0,\Sigma)\), as in \cref{sec:noise-examples}.
The positive-definiteness assumption ensures that the limiting
Gaussian scores have no ties.

The normal column uses the following common notation for the
regularizer-dependent bound. Its four cases come from summing the
response tails in \appref{app:normal-noise-column}:
\begin{equation}
  \mathcal N_{h,\rho}(T):=
  \begin{cases}
    1,&h\text{ is entropy or Euclidean},\\
    1,&h\text{ is Tsallis-}q,\ \rho<q,\\
    \log(T+1),&h\text{ is Tsallis-}q,\ \rho=q,\\
    T^{(\rho-q)/(1-q)},&h\text{ is Tsallis-}q,\ \rho>q.
  \end{cases}
  \label{eq:main-normal-rate}
\end{equation}

The summary table in \cref{tab:main-noise-examples} presents the
limiting reference responses and surplus bounds under this setup.

For a direct check of any entry, the limiting-response calculation
is in \appref{app:cumulative-finite-variance}; the normal bounds,
including each exponent in \eqref{eq:main-normal-rate}, are in
\appref{app:normal-noise-column}; the selection sum is in
\appref{app:selection-noise-column}; and the movement sum and its
additional noise condition are in \appref{app:movement-noise-column}.
The single-star entry
(\hyperref[tab:main-noise-examples]{\(\star\)}) is proved in
\appref{app:movement-below-half}.

\paragraph{Fresh noise and temporal dependence.}
\label{app:fresh-finite-variance}
For a deterministic optimizer trajectory,
\(S_t=\sum_{\tau<t}B^\top x_\tau\) is deterministic.
The expected reward at round \(t\) therefore depends only on the
marginal law of \(Q_h(\eta_t(S_t+z_t))\). Changing the dependence
across rounds while keeping these marginals leaves every expected
cumulative reward unchanged. No temporal independence is needed
in the fresh model.

\subsection{Verification of the limiting-response column}
\label{app:cumulative-finite-variance}

Fix a max--min strategy \(x^\star\) and let \(K=\br(x^\star)\),
as in \cref{sec:learner-dynamics-maxmin}. We first verify
localization to \(\Delta_K\), then apply the corresponding response
regime to identify the mixed strategy on this face.

For the power-law schedule,
\(r_t=t\eta_t\sim\eta_0t^{1-\rho}\longrightarrow\infty\).
The normal noise comparison uses \(\Xi_t=\eta_tZ_t\). Since
\(\spnorm{u}\le2\lVert u\rVert_2\), for \(t\ge1\),
\[
  \mathbb E\!\left[
    \left(\frac{\spnorm{\Xi_t}}{r_t}\right)^2\right]
  =\mathbb E\!\left[
    \left(\frac{\spnorm{Z_t}}{t}\right)^2\right]
  \le\frac{4\mathbb E\lVert Z_t\rVert_2^2}{t^2}.
\]
Here \(\mathbb E\lVert Z_t\rVert_2^2\) is zero in the noiseless model and equals
\(\mathbb E\lVert z\rVert_2^2\) in the fresh model. In the cumulative
model, independence and \(\mathbb E[z]=0\) give
\(\mathbb E\lVert Z_t\rVert_2^2=t\mathbb E\lVert z\rVert_2^2\).
Thus the displayed bound tends to zero in every model, proving
\(\spnorm{\Xi_t}/r_t\pto0\). This verifies
\cref{def:step-size-noise-model} for every row. The three
regularizers satisfy \cref{def:regularizer}, as recorded in
\cref{rem:finite-regularizer-examples}.
\cref{lem:reduction} in \appref{app:face-reduction}
therefore gives
\(\lVert y_t-Q_{h,K}(\Xi_{t,K})\rVert_1\pto0\).
It remains to identify the limit of the restricted response in each row.

\emph{Noiseless and fresh noise, \(0<\rho<1\).}
In the noiseless model \(\Xi_t=0\). In the fresh model,
\[
  \mathbb E\lVert\Xi_t\rVert_2^2
    =\eta_t^2\mathbb E\lVert z\rVert_2^2\longrightarrow0.
\]
In both cases \(\spnorm{\Xi_{t,K}}\pto0\).
\cref{prop:vanishing} therefore gives convergence to
\(y_K\), and \cref{cor:uniform-separable-selection}
identifies this minimizer for the three separable regularizers:
\(y_t\pto y_K=\frac1{|K|}\sum_{j\in K}e_j\).
These are the first two entries of the limiting-response column;
the cited results appear in \appref{app:vanishing-noise}.

\emph{Cumulative noise, \(\rho>1/2\).}
The second-moment calculation now gives
\[
  \mathbb E\lVert\Xi_t\rVert_2^2
    =t\eta_t^2\mathbb E\lVert z\rVert_2^2
    =O(t^{1-2\rho})\longrightarrow0.
\]
Thus \(\spnorm{\Xi_{t,K}}\pto0\), and the same proposition and
corollary give
\(y_t\pto |K|^{-1}\sum_{j\in K}e_j\), proving the third entry.

\emph{Cumulative noise, \(\rho=1/2\).}
The multivariate central limit theorem and
\(\eta_t\sqrt t=\eta_0\sqrt{t/(t+1)}\to\eta_0\) give, by
Slutsky's theorem,
\[
  \frac{Z_t}{\sqrt t}\dto G,
  \qquad
  \Xi_{t,K}
    =(\eta_t\sqrt t)\frac{Z_{t,K}}{\sqrt t}
      \dto\eta_0G_K.
\]
This is the distributional convergence required by
\cref{prop:finite} in \appref{app:finite-noise},
with \(c_t=0\) and \(\Xi_{\infty,K}=\eta_0G_K\). It yields
\(y_t\dto Q_{h,K}(\eta_0G_K)\),
which is the fourth entry.

\emph{Cumulative noise, \(0<\rho<1/2\).}
To apply the diverging-noise theorem, normalize by the cumulative
noise scale \(a_t=\eta_t\sqrt t\). Then
\[
  a_t\sim\eta_0t^{1/2-\rho}\longrightarrow\infty,
  \qquad
  \frac{\Xi_{t,K}}{a_t}
    =\frac{Z_{t,K}}{\sqrt t}\dto G_K.
\]
For distinct \(i,j\in K\), positive definiteness of \(\Sigma\)
gives
\(\Var(G_i-G_j) =(e_i-e_j)^\top\Sigma(e_i-e_j)>0\).
Each difference is therefore a nondegenerate Gaussian, so
\(\Pr(G_i=G_j)=0\). Since \(K\) is finite, the largest coordinate
of \(G_K\) is unique almost surely. \cref{thm:diverging}
in \appref{app:growing-noise} now applies with \(c_t=0\)
and \(\Xi_{\infty,K}=G_K\), giving the final entry:
\(y_t\dto e_{\argmax_{j\in K}G_j}\).

Thus the reference response under cumulative noise satisfies
\begin{equation}
  y_t^\star\dto
  \begin{cases}
    y_K,&\rho>1/2,\\
    Q_{h,K}(\eta_0G_K),&\rho=1/2,\\
    e_{\argmax_{j\in K}G_j},&0<\rho<1/2.
  \end{cases}
  \label{eq:main-finite-variance-response-limits}
\end{equation}
The first convergence also holds in probability.

The limiting mean responses in the three cumulative regimes are,
respectively,
\[
  \mu=y_K,\qquad
  \mu=\mathbb E[Q_{h,K}(\eta_0G_K)],\qquad
  \mu=\sum_{j\in K}\Pr(J=j)e_j,
  \quad J=\argmax_{j\in K}G_j.
\]
The last two means need not be uniform: centering the increments
does not imply exchangeability of the Gaussian scores.
Selection depends on the profitability of these means, while
movement depends on the susceptibility of the realized responses.

\subsection{Derivation of the normal-surplus column}
\label{app:normal-noise-column}
\label{app:finite-variance-rate-conditions}

To obtain the normal column, it suffices to verify the summable
tail condition and evaluate the response series of
\appref{sec:passive}. For the fixed max--min strategy
\(x^\star\), retain \(K=\br(x^\star)\) and
\(\delta_j=(A^\top x^\star)_j-\val\). The identity
\eqref{eq:fixed-surplus-decomposition-discrete} gives
\[
  \Surplus((x^\star)_{\times T})
    =\lagd((x^\star)_{\times T},T)
    =\sum_{t=0}^{T-1}\sum_{j\notin K}
       \delta_j\mathbb E[y_{t,j}].
\]
Thus only the coordinates with a positive payoff gap contribute.
When \(K=[m]\), this sum is empty and normal surplus is zero.
For the other cases, \cref{cor:fixed-standard-schedule-criteria}
bounds this sum under the power-law schedule
\eqref{eq:polynomial-schedule} and
\cref{ass:passive-linear-tail}. The schedule is the one
used in the table. We verify the noise-tail assumption for every
\(c>0\).

\emph{Noiseless and fresh noise.}
In the noiseless model,
\(\Pr(\spnorm{Z_t}>ct)=0\). In the fresh model,
\(Z_t=z_t\) has the same marginal law as \(z\). The bound
\(\spnorm{z_t}\le2\lVert z_t\rVert_2\) and Markov's inequality
give
\[
  \sum_{t=1}^{\infty}\Pr(\spnorm{Z_t}>ct) \le\sum_{t=1}^{\infty} \Pr\!\left(\lVert z_t\rVert_2>\frac{ct}{2}\right) \le\frac{4\mathbb E\lVert z\rVert_2^2}{c^2} \sum_{t=1}^{\infty}\frac1{t^2}<\infty.
\]
Thus both models satisfy \cref{ass:passive-linear-tail}.
The fresh-noise bound requires no independence across rounds and
also applies to varying marginals with uniformly bounded second moments.

\emph{Cumulative i.i.d.\ noise.}
For every coordinate \(j\), the increments \(z_{\tau,j}\) are
i.i.d., have mean zero, and satisfy
\(\mathbb E[z_{\tau,j}^2]\le\mathbb E\lVert z\rVert_2^2<\infty\).
The complete-convergence theorem of
\citet[Theorem~1]{hsu1947complete}, applied with threshold \(c/2\),
therefore gives
\[
  \sum_{t=1}^{\infty}
    \Pr\!\left(\left|\sum_{\tau=0}^{t-1}z_{\tau,j}\right|
                      >\frac{ct}{2}\right)<\infty.
\]
Since \(\spnorm{Z_t}\le2\max_{j\in[m]}|Z_{t,j}|\), a union
bound over the finitely many coordinates yields
\[
  \sum_{t=1}^{\infty}\Pr(\spnorm{Z_t}>ct)
    \le\sum_{j=1}^m\sum_{t=1}^{\infty}
       \Pr\!\left(|Z_{t,j}|>\frac{ct}{2}\right)<\infty.
\]
This verifies \cref{ass:passive-linear-tail} for the
cumulative model in all three ranges of \(\rho\).

The preceding bounds make the noise contribution finite in every
row, independently of \(\rho\). To see the source of the four rates
in \eqref{eq:main-normal-rate}, apply
\eqref{eq:fixed-lag-master-upper} with \(\alpha=1/2\). For each
\(j\notin K\), it remains to sum
\[
  \phi\!\left(\theta'(1/|K|)-\delta_jr_t/2\right),
  \qquad r_t=t\eta_t=\Theta(t^{1-\rho}).
\]
The scalar inversions established in
\appref{app:normal-standard-regularizer-example} give, respectively,
\[
  \frac1{|K|}e^{-\delta_jr_t/2},
  \qquad
  \left[\frac1{|K|}-\frac{\delta_jr_t}{2}\right]_+,
  \qquad
  \frac1{|K|}\left(
    1+\frac{1-q}{2q}|K|^{q-1}\delta_jr_t
  \right)^{-1/(1-q)}.
\]
The entropy terms have a summable stretched-exponential tail.
The Euclidean terms become zero once
\(\delta_jr_t/2\ge1/|K|\), so only finitely many contribute.
For Tsallis, the displayed term has order
\(t^{-(1-\rho)/(1-q)}\). The integral comparison therefore yields
\[
  \sum_{t=1}^{T}t^{-(1-\rho)/(1-q)}
   =\begin{cases}
      O(1),&(1-\rho)/(1-q)>1,\\
      O(\log(T+1)),&(1-\rho)/(1-q)=1,\\
      O\!\left(T^{1-(1-\rho)/(1-q)}\right),
        &(1-\rho)/(1-q)<1.
     \end{cases}
\]
These three conditions are exactly \(\rho<q\), \(\rho=q\),
and \(\rho>q\). In the last case the exponent is
\(1-(1-\rho)/(1-q)=(\rho-q)/(1-q)\). Adding the bounded noise
contribution and the finitely many positive-gap coordinates gives
\cref{cor:fixed-standard-schedule-criteria} with
\(\hat x=x^\star\):
\[
  \Surplus((x^\star)_{\times T})
    =\begin{cases}
       O(1),&h\text{ is negative entropy or Euclidean},\\
       O(1),&h\text{ is negative Tsallis-}q,\ \rho<q,\\
       O(\log(T+1)),&h\text{ is negative Tsallis-}q,\ \rho=q,\\
       O(T^{(\rho-q)/(1-q)}),
          &h\text{ is negative Tsallis-}q,\ \rho>q.
      \end{cases}
\]
These are the \(O(\mathcal N_{h,\rho}(T))\) bounds in every
row of the normal column.

\subsection{Derivation of the selection-surplus column}
\label{app:selection-noise-column}
The selection rates follow by applying
\cref{thm:new-selection-every-horizon} to the limiting mean response
in each row. Let \(y_t^\star\) denote the reference response under
repeated play of \(x^\star\). The convergence established above
also gives convergence of its mean: each coordinate is a bounded
continuous function on the simplex, so convergence in distribution
implies convergence of coordinate expectations. Thus
\[
  \mathbb E[y_t^\star]\longrightarrow\mu\in\Delta_K,
\]
where \(\mu=y_K\) in the noiseless and fresh-noise models and in
the cumulative model with \(\rho>1/2\). In the remaining
cumulative regimes,
\[
  \mu=
  \begin{cases}
    \mathbb E[Q_{h,K}(\eta_0G_K)],&\rho=1/2,\\[2mm]
    \displaystyle\sum_{j\in K}\Pr(J=j)e_j,
      \quad J=\argmax_{j\in K}G_j,&0<\rho<1/2.
  \end{cases}
\]
These means determine the available selection gain. Fix a row for
which
\[
  g_{\rm mean}(x^\star;\mu)
    =\max_{x\in\Delta_n}x^\top A\mu-\val>0,
\]
and choose
\(x^\dagger\in\argmax_{x\in\Delta_n}x^\top A\mu\).
Write \(g_{\rm mean}\) for this positive gap.

All three regularizers are finite, continuous, and strictly convex
on the simplex. The power-law learning rates are positive and
nonincreasing, and concavity of \(s^\rho\) gives
\[
  0\le\frac1{\eta_{t+1}}-\frac1{\eta_t}
    =\frac{(t+2)^\rho-(t+1)^\rho}{\eta_0}
    \le\frac{2^\rho-1}{\eta_0}.
\]
Thus the reciprocal increments are uniformly bounded, verifying
the remaining hypotheses of
\cref{thm:new-selection-every-horizon}.

The resulting strategy moves toward \(x^\dagger\) with weights
\[
  \widehat x_t=(1-\varepsilon_t)x^\star
                      +\varepsilon_t x^\dagger,
  \qquad
  \varepsilon_t=\frac c{\eta_0}
    \bigl((t+2)^\rho-(t+1)^\rho\bigr),
\]
where \(c>0\) is a sufficiently small fixed constant. The increment
bound ensures that these weights lie in \([0,1]\). Their cumulative
effect satisfies
\[
  b_t:=\sum_{\tau<t}\varepsilon_\tau
    =c\left(\frac1{\eta_t}-\frac1{\eta_0}\right),
  \qquad
  \eta_tb_t=c\left(1-\frac{\eta_t}{\eta_0}\right)\le c.
\]
The total deviation therefore grows at rate \(t^\rho\), while its
effect on the learner's scaled score remains uniformly small.

To see how this preserves the gain, write
\(a=A^\top(x^\dagger-x^\star)\). By
\eqref{eq:new-selection-score}, the scaled score differs from its
reference value by \(-\eta_tb_ta\).
\cref{cor:new-selection-payoff-stability} therefore bounds the
resulting change in expected directional payoff by
\(g_{\rm mean}/4\) when \(c\) is sufficiently small. Since
\(\mathbb E\langle a,y_t^\star\rangle\to
\langle a,\mu\rangle=g_{\rm mean}\), the actual response satisfies
\[
  \mathbb E\langle a,y_t\rangle\ge\frac{g_{\rm mean}}2
\]
for all sufficiently large \(t\). Max--min optimality of
\(x^\star\) then gives
\[
  \mathbb E[\widehat x_t^\top Ay_t-\val]
    =\mathbb E[(x^\star)^\top Ay_t-\val]
       +\varepsilon_t\mathbb E\langle a,y_t\rangle
    \ge\frac{\varepsilon_tg_{\rm mean}}2
\]
on each such round. Summing and absorbing the finite initial period
into a constant gives the bound of
\cref{cor:new-selection-power-law}:
\[
  \Surplus(\widehat x_{<T})
    \ge\frac{c\,g_{\rm mean}}{2\eta_0}
       \bigl((T+1)^\rho-1\bigr)-C_0,
  \qquad T\ge1,
\]
where \(C_0\ge0\) is independent of \(T\). This proves the
\(\Omega(T^\rho)\) entries throughout the selection column,
including \(\Omega(\sqrt T)\) at \(\rho=1/2\). The same
deterministic strategy sequence applies at every horizon.

\subsection{Derivation of the movement-surplus column}
\label{app:movement-noise-column}
The movement column requires a noise-specific check of the
alternating-surplus theorem. Use the pair \((x',x'')\) from
\eqref{eq:new-movement-strategies}, with
\(x^\star,i,j,\ell\) satisfying
\cref{ass:new-movement-geometry}. In this construction,
\(v'=A^\top x'\) distinguishes the two best responses:
\(v'_i-v'_j=x_\ell^\star(A_{\ell i}-A_{\ell j})>0\).
The bound already proved in \cref{thm:new-movement-pair} is
\[
  \mathbb E[x'^\top Ay_{2s}+x''^\top Ay_{2s+1}-2\val]
    \ge\eta_{2s+1}\mathcal K_s
       -2s(\eta_{2s}-\eta_{2s+1})\mathcal D_s-\mathcal N_s.
\]
For the desired rate, \(\mathcal K_s\) must remain positive and
both error terms must be \(o(\eta_{2s+1})\). We verify these
conditions using \cref{prop:new-movement-sufficient}, then sum the
learning rates. The cumulative rows need the additional
adjacent-law assumption stated explicitly in the corollary.

\begin{corollary}[Movement rates for the noise models]
\label{cor:movement-noise-models}
Under the setup of \cref{tab:main-noise-examples} and
\cref{ass:new-movement-geometry}, a finite initial period of
max--min play, chosen independently of \(T\), gives
\(\Surplus(x^{\rm alt}_{<T})=\Omega(T^{1-\rho})\) in either of
the following cases:
\begin{enumerate}
  \item noiseless or fresh finite-variance noise, with
  \(0<\rho<1\);
  \item cumulative i.i.d.\ finite-variance noise, with
  \(1/2\le\rho<1\), whose adjacent scaled marginals admit
  couplings satisfying \eqref{eq:new-movement-coupling}.
\end{enumerate}
\end{corollary}

The mark \(\dagger\) denotes the additional coupling condition in
the second case. The finite-variance assumption verifies the other
noise conditions but does not, by this argument, verify that
coupling. After the proof we give the exact Gaussian, symmetric
Laplace, and Student-\(t\) models for which it is established.

\begin{proof}
We apply \cref{prop:new-movement-sufficient}. Small relative noise
keeps the two distinguished best responses in the interior, where
their inverse curvatures give a positive movement gain. Localization
removes the schedule error, and a coupling of adjacent noise laws
controls the remaining error.

The regularizer and schedule conditions hold throughout the table.
For negative entropy, Euclidean, and negative Tsallis, respectively,
the inverse curvatures are \(p\), \(1\), and \(p^{2-q}/q\),
so they are bounded. The first and third vanish as \(p\to0\),
which verifies the boundary condition for the two steep regularizers.
Euclidean regularization is non-steep. Also,
\(\eta_t\to0\), \(t\eta_t\to\infty\), and the mean value
theorem gives
\[
  0\le\eta_t-\eta_{t+1}
       \le\rho\eta_0(t+1)^{-\rho-1},
  \qquad
  t(\eta_t-\eta_{t+1})=O(\eta_{t+1}).
\]
Thus the schedule satisfies \eqref{eq:new-movement-schedule}.
It remains to verify the noise conditions in each row.

\emph{Noiseless learning.}
Here \(\Xi_t=0\). The small-noise event has probability one,
and coupling both adjacent marginals at zero gives
\(\mathcal N_s=0\).

\emph{Fresh finite-variance noise.}
The bound
\(\mathbb E\lVert\Xi_t\rVert_2^2
  =\eta_t^2\mathbb E\lVert z\rVert_2^2\to0\)
shows that the scaled noise is tight and that
\(\Pr(\spnorm{\Xi_{2s}}\le a_{\rm alt})\to1\).
For adjacent marginals, use a common \(z\sim\mathcal D\):
\[
  \widehat\Xi_s^0=\eta_{2s}z,
  \qquad
  \widehat\Xi_s^1=\eta_{2s+1}z.
\]
Finite variance gives \(\mathbb E\spnorm{z}<\infty\), and hence
\[
  \mathbb E\spnorm{\widehat\Xi_s^0-\widehat\Xi_s^1}
    =(\eta_{2s}-\eta_{2s+1})\mathbb E\spnorm{z}
    =O(s^{-\rho-1})=o(\eta_{2s+1}).
\]
This comparison uses only the common marginal law of the fresh
perturbations, so it also applies when they are dependent across rounds.

\emph{Cumulative noise with \(\rho>1/2\).}
The second-moment bound above gives
\(\mathbb E\lVert\Xi_t\rVert_2^2=O(t^{1-2\rho})\to0\).
Thus tightness and the small-noise probability condition hold as
in the fresh model. The adjacent-law condition is the hypothesis
marked by \(\dagger\).

\emph{Cumulative noise with \(\rho=1/2\).}
The central limit theorem gives \(\Xi_t\dto\eta_0G\).
This implies tightness. Applying the Portmanteau theorem to the
open neighborhood \(\{u:\spnorm{u}<a_{\rm alt}\}\) also gives
\[
  \liminf_{s\to\infty}
    \Pr(\spnorm{\Xi_{2s}}\le a_{\rm alt})
    \ge\Pr(\spnorm{\eta_0G}<a_{\rm alt})>0.
\]
The last inequality holds because a centered Gaussian has positive
probability in every neighborhood of zero. Thus the response remains
sufficiently responsive at the threshold as well; the additional
hypothesis again controls the adjacent noise laws.

In each case, \cref{prop:new-movement-sufficient} now gives
\(\liminf_s\mathcal K_s>0\), \(\mathcal D_s=o(1)\), and
\(\mathcal N_s=o(\eta_{2s+1})\). Together with the schedule bound,
these are the conditions of \cref{thm:new-movement-surplus}.
For some \(c>0\) and finite \(s_0\), its conclusion is
\[
  \Surplus(x^{\rm alt}_{<T})
    \ge c\sum_{s=s_0}^{\lfloor T/2\rfloor-1}\eta_{2s+1}.
\]
The sum is the source of the exponent in the movement column:
\[
  \sum_{s=s_0}^{\lfloor T/2\rfloor-1}\eta_{2s+1}
    =\eta_0 2^{-\rho}
      \sum_{s=s_0}^{\lfloor T/2\rfloor-1}(s+1)^{-\rho}
    =\frac{\eta_0}{2(1-\rho)}T^{1-\rho}+O(1),
\]
by the integral comparison in \eqref{eq:new-movement-odd-sum}.
This gives \(\Omega(T^{1-\rho})\) in the first three rows
and \(\Omega(\sqrt T)\) in the fourth.
\end{proof}

\begin{remark}[Why finite variance does not supply the coupling]
  A centered finite-variance law may retain an arithmetic obstruction
between adjacent sums. For example, let
\(z=(R,2R')\), where \(R,R'\) are independent Rademacher variables.
Then \(\Sigma=\operatorname{diag}(1,4)\succ0\). Each increment
of the score difference is an odd integer. Consequently,
\(D_t:=Z_{t,1}-Z_{t,2}\) has the parity of \(t\), so every
coupling of \(D_t,D_{t+1}\) satisfies
\(|D_t-D_{t+1}|\ge1\). By the reverse triangle inequality,
\[
\begin{aligned}
  \mathbb E\bigl|\eta_tD_t-\eta_{t+1}D_{t+1}\bigr|
   &\ge\eta_t-(\eta_t-\eta_{t+1})\mathbb E|D_{t+1}|\\
   &\ge\eta_t-(\eta_t-\eta_{t+1})\sqrt{5(t+1)}
    =\eta_t\bigl(1-O(t^{-1/2})\bigr).
\end{aligned}
\]
In two dimensions, the span of a vector is the absolute difference
of its coordinates. Thus no adjacent coupling for this noise law
satisfies \eqref{eq:new-movement-coupling}. This explains why the
\(\dagger\) condition must be checked separately; failure of this
sufficient condition by itself does not determine the surplus rate.
\end{remark}

\subsubsection{The noise models covered by \texorpdfstring{\(\dagger\)}{the dagger}}
\label{app:dagger-noise-models}
The verification of \(\dagger\) concerns the following specified
vector distributions: centered Gaussian noise; symmetric Laplace
noise with a common exponential Gaussian mixing variable or with
independent Laplace coordinates; and Student-\(t\) noise with
\(\nu>2\) with a common inverse-chi-square Gaussian mixing variable
or with independent Student-\(t\) coordinates. Fixed invertible
linear transformations of the independent-coordinate models are
also allowed. Merely specifying finite variance or the univariate
coordinate marginals does not impose these joint-law assumptions.

The common-mixing models admit one construction: conditional on a random
variance, each increment is Gaussian. Conditional variances add under
cumulative noise. We can therefore couple two adjacent sums by using
the same Gaussian vector with their respective standard deviations.
The following lemma quantifies how little those standard deviations
change when one more increment is added.

\begin{lemma}[Adjacent cumulative laws for Gaussian scale mixtures]
\label{lem:movement-scale-mixture-coupling}
Suppose the increments are i.i.d.\ copies of
\[
  z=\sqrt V\,H,
  \qquad H\sim\mathcal N(0,\Sigma_0),
  \qquad V>0,
\]
where \(H\) and \(V\) are independent, \(\Sigma_0\succ0\),
\(\mathbb E[V]<\infty\), and
\(\mathbb E[V^{-1/2}]<\infty\).
For \(\eta_t=\eta_0(t+1)^{-\rho}\) with \(0<\rho<1\),
the adjacent scaled cumulative marginals have couplings satisfying
\begin{equation}
  \mathbb E\spnorm{\widehat\Xi_s^0-\widehat\Xi_s^1}
      =O(s^{-\rho-1/2})=o(\eta_{2s+1}).
  \label{eq:movement-scale-mixture-coupling}
\end{equation}
\end{lemma}

\begin{proof}
Let \(V_0,V_1,\ldots\) be i.i.d.\ copies of \(V\), and write
\(R_t=\sum_{\tau<t}V_\tau\) for the accumulated variance factor.
Conditional on these factors, \(Z_t\) is Gaussian with covariance
\(R_t\Sigma_0\). Using one \(H\sim\mathcal N(0,\Sigma_0)\)
independent of all the factors therefore gives the coupling
\[
  \widetilde Z_t=\sqrt{R_t}\,H,
  \qquad
  \widetilde Z_{t+1}=\sqrt{R_t+V_t}\,H.
\]
Each variable has the required cumulative marginal law.

We first bound the change in standard deviation. Since
\(u\mapsto u^{-1/2}\) is convex, Jensen's inequality gives
\[
  \mathbb E[R_t^{-1/2}]
    =\frac1{\sqrt t}\,
      \mathbb E\left[\left(\frac1t\sum_{\tau<t}V_\tau
                        \right)^{-1/2}\right]
    \le\frac{\mathbb E[V^{-1/2}]}{\sqrt t}.
\]
Rationalizing the difference of square roots and using independence
of \(V_t\) and \(R_t\) now yields
\[
\begin{aligned}
  \mathbb E[\sqrt{R_t+V_t}-\sqrt{R_t}]
    &=\mathbb E\left[\frac{V_t}{\sqrt{R_t+V_t}+\sqrt{R_t}}\right]\\
    &\le\frac12\mathbb E[V]\mathbb E[R_t^{-1/2}]
      =O(t^{-1/2}).
\end{aligned}
\]
Multiplying by \(\mathbb E\spnorm{H}<\infty\) gives
\(\mathbb E\spnorm{\widetilde Z_{t+1}-\widetilde Z_t}
  =O(t^{-1/2})\).
Also, \(\mathbb E[\sqrt{R_t}]\le\sqrt{t\mathbb E[V]}\), so
\(\mathbb E\spnorm{\widetilde Z_t}=O(\sqrt t)\).

The scaled comparison has one contribution from this change in
variance and one from the learning rate:
\[
\begin{aligned}
  \mathbb E\spnorm{\eta_t\widetilde Z_t
                       -\eta_{t+1}\widetilde Z_{t+1}}
  &\le(\eta_t-\eta_{t+1})
         \mathbb E\spnorm{\widetilde Z_t}
       +\eta_{t+1}
         \mathbb E\spnorm{\widetilde Z_{t+1}-\widetilde Z_t}\\
  &=O(t^{-\rho-1}\sqrt t)
       +O(t^{-\rho}t^{-1/2})
    =O(t^{-\rho-1/2}).
\end{aligned}
\]
Set \(t=2s\),
\(\widehat\Xi_s^0=\eta_{2s}\widetilde Z_{2s}\), and
\(\widehat\Xi_s^1=\eta_{2s+1}\widetilde Z_{2s+1}\).
This proves \eqref{eq:movement-scale-mixture-coupling}.
\end{proof}

To apply the lemma to the named examples, take
\[
  V=
  \begin{cases}
    1,&\text{Gaussian},\\
    E,\quad E\sim\operatorname{Exp}(1),
      &\text{symmetric multivariate Laplace},\\
    \nu/U,\quad U\sim\chi^2_\nu,
      &\text{multivariate Student-\(t\), }\nu>2.
  \end{cases}
\]
These are the standard Gaussian scale-mixture representations
\citep{andrews1974scale,hintz2020normal}.
For the Laplace choice, coordinate \(j\) has scale
\(\sqrt{(\Sigma_0)_{jj}/2}\). The covariance of the increment
vector is \(\Sigma=\mathbb E[V]\Sigma_0\) in every case.
For Gaussian and Laplace mixing, \(\mathbb E[V]=1\), while
\(\mathbb E[V^{-1/2}]\) equals \(1\) and \(\sqrt\pi\),
respectively. For Student-\(t\) mixing,
\[
  \mathbb E[V]=\frac{\nu}{\nu-2}<\infty,
  \qquad
  \mathbb E[V^{-1/2}]=\mathbb E\sqrt{U/\nu}\le1.
\]
Thus all three examples satisfy the lemma's moment conditions.

The conclusion also holds for independent Laplace coordinates or
independent Student-\(t\) coordinates with \(\nu>2\), and for fixed
invertible linear transforms of these vectors. Indeed, apply the same
scalar coupling independently to each
coordinate. The proof bounds its expected absolute displacement by
\(O(t^{-\rho-1/2})\). For a fixed matrix \(L\),
\(\spnorm{Lu}\le2\lVert L\rVert_{2\to2}\sum_j|u_j|\), so summing
the coordinate bounds preserves this order after the transformation.
Consequently, \cref{cor:movement-noise-models} gives both
\(\dagger\) entries for these distributions: \(\Omega(T^{1-\rho})\)
when \(\rho>1/2\), and \(\Omega(\sqrt T)\) when \(\rho=1/2\).

For cumulative noise with \(0<\rho<1/2\), the coupling estimate
above still holds for these models, but it no longer ensures a
positive susceptibility. The calculation for
(\hyperref[tab:main-noise-examples]{\(\star\)}) below proves a
smaller rate in matching pennies, including an upper bound for
every trajectory independent of the noise.

\subsection{Movement exploitation below \texorpdfstring{\(\rho=1/2\)}{rho=1/2}: the entry \texorpdfstring{(\hyperref[tab:main-noise-examples]{\(\star\)})}{(star)}}
\label{app:movement-below-half}

We now justify the entry
(\hyperref[tab:main-noise-examples]{\(\star\)}) by proving that
\(\Omega(T^{1-\rho})\) cannot hold throughout the final row.
The missing hypothesis is the positive-susceptibility bound
\eqref{eq:new-movement-positive-condition}: when
\(\eta_t\sqrt t\to\infty\), increasingly large relative noise can
make the response insensitive to the alternating score shift.
In matching pennies we compute the alternating surplus as
\(\sqrt T/(\sigma\sqrt{2\pi})+o(\sqrt T)\) and prove an
\(O(\sqrt T)\) upper bound for every trajectory fixed independently
of the noise. Since \(\sqrt T=o(T^{1-\rho})\) for
\(0<\rho<1/2\), this example rules out the nominal movement rate
under the final row's assumptions alone.

\paragraph{The game and the noise.}
Consider matching pennies,
\[
  A=\begin{pmatrix}1&-1\\-1&1\end{pmatrix},\qquad
  \val=0,\qquad x^\star=(\tfrac12,\tfrac12).
\]
Here \(v=A^\top x^\star=0\), so \(K=\{1,2\}\) and \(\delta=0\).
\cref{ass:new-movement-geometry} holds with \(i=1\),
\(j=2\), and \(\ell=1\), and \eqref{eq:new-movement-strategies} gives
\[
  x'=(\tfrac34,\tfrac14),\qquad x''=(\tfrac14,\tfrac34),\qquad
  v'=A^\top x'=(\tfrac12,-\tfrac12).
\]
The trajectory \(x^{\rm alt}_{<T}\) plays \(x'\) on even rounds and
\(x''\) on odd rounds, starting at round zero, and plays \(x^\star\)
on an unmatched final round.

The noise is cumulative, \(Z_t=\sum_{\tau<t}z_\tau\), with i.i.d.\
increments that satisfy \eqref{eq:finite-variance-noise} and
\(\Sigma\succ0\). With two actions only the difference of the two
noise coordinates affects the response. We write
\[
  \xi_\tau:=z_{\tau,1}-z_{\tau,2},\qquad
  W_t:=Z_{t,1}-Z_{t,2}=\sum_{\tau<t}\xi_\tau,\qquad
  \sigma^2:=\Var(\xi_\tau).
\]
The increments \(\xi_\tau\) are i.i.d.\ with mean zero, and
\(0<\sigma^2<\infty\) because \(\Sigma\succ0\).

The limiting response of this row is \(e_{\argmax_jG_j}\), whose
mean is \(\mu=(\tfrac12,\tfrac12)\) because \(G_1-G_2\) is a
centered Gaussian. Hence \(g_{\rm mean}(x^\star;\mu)=0\), so the
positive-gap selection theorem supplies no
\(\Omega(T^\rho)\) guarantee in this game. For asymmetric
increments the finite-time mean need not equal \(\mu\); the proof
below explicitly controls the resulting adjacent-round bias.

\begin{proposition}[Matching pennies with \(0<\rho<1/2\)]
\label{prop:movement-below-half}
In matching pennies, let \(h\) be negative entropy, Euclidean, or
negative Tsallis with \(q\in(0,1)\), and let
\(\eta_t=\eta_0(t+1)^{-\rho}\) with \(\eta_0>0\) and
\(0<\rho<1/2\). Let the noise be cumulative with i.i.d.\ increments
that satisfy \eqref{eq:finite-variance-noise} and \(\Sigma\succ0\).
Then
\begin{equation}
  \Surplus(x^{\rm alt}_{<T})
     =\frac{\sqrt T}{\sigma\sqrt{2\pi}}+o(\sqrt T).
  \label{eq:movement-below-half-alt}
\end{equation}
Moreover, there is a constant \(C<\infty\) such that every trajectory
\(x_{<T}\) fixed independently of the realized noise satisfies
\begin{equation}
  \Surplus(x_{<T})\le C\sqrt T
  \qquad\text{for every }T\ge1.
  \label{eq:movement-below-half-upper}
\end{equation}
\end{proposition}

The first statement is the exact first-order surplus of the
alternating strategy. The constant depends on the noise only through
\(\sigma\). It does not depend on \(\rho\), on \(\eta_0\), or on the
regularizer. The second statement says that \(\sqrt T\) is the best
order in this game. Together they give
\[
  \sup_{x_{<T}}\Surplus(x_{<T})=\Theta(\sqrt T),
  \qquad
  \Surplus(x^{\rm alt}_{<T})=o(T^{1-\rho}).
\]
Here the supremum is over trajectories fixed independently of the noise.

\paragraph{Proof of \cref{prop:movement-below-half}: mechanism and required estimates.}
With two actions the response is a function of one number, the
difference of the two scores. Playing \(x'\) shifts this difference
by one unit, and the pair gains from the change of the response under
this shift. The response changes only when the noisy score difference
is within distance of order \(1/\eta_t\) of a tie. The noise \(W_t\)
is spread over a range of order \(\sqrt t\), so this happens with
probability of order \(1/(\eta_t\sqrt t)\). On this event a unit
shift moves the response by order \(\eta_t\). The product is of order
\(1/\sqrt t\), and the learning rate cancels. Summing \(1/\sqrt t\)
over the pairs gives \(\sqrt T\).

Two steps need an argument under \eqref{eq:finite-variance-noise}.
First, the noise need not have a density, so we cannot integrate
against one. We use a local limit theorem to count the probability
that \(W_t\) falls in a window of width \(1/\eta_t\). Second, the
noise need not be symmetric, so the expected payoff on an even round
is not zero. We show that this payoff is almost the same on two
adjacent rounds, so it cancels inside each pair. This step uses only
\(\mathbb E\xi_\tau=0\). The upper bound
\eqref{eq:movement-below-half-upper} uses a potential function.

\paragraph{Step 1: the response is a scalar function.}
The response to a score \((u_1,u_2)\) does not change when the same
constant is added to both coordinates. So the difference of its two
coordinates is a function of \(u_1-u_2\). We write
\(\psi(u_1-u_2):=Q_h(u_1,u_2)_1-Q_h(u_1,u_2)_2 \).
When both coordinates of \(y=Q_h(u_1,u_2)\) are positive, the KKT
conditions of \appref{app:preliminary-kkt} give
\(\theta'(y_1)-\theta'(y_2)=u_1-u_2\). With \(y_1=(1+\psi)/2\) and
\(y_2=(1-\psi)/2\), this equation determines \(\psi\). For negative
entropy and for Euclidean,
\[
  \psi(u)=\tanh(u/2)
  \qquad\text{and}\qquad
  \psi(u)=\max\{-1,\min\{u,1\}\},
\]
respectively. For negative Tsallis, \(\psi(u)\) is the solution of
\(\theta'\bigl(\tfrac{1+\psi}2\bigr)
  -\theta'\bigl(\tfrac{1-\psi}2\bigr)=u\).

The gain of a pair is a change of \(\psi\) under a small shift of
its argument. So we need the derivative \(\psi'\). For Euclidean,
\(\psi\) has a kink at \(u=\pm1\), and we set
\(\psi'(u)=\mathbf 1\{|u|\le1\}\).

\begin{lemma}[The scalar response]
\label{lem:mp-scalar-response}
For each of the three regularizers, \(\psi\) is odd, nondecreasing,
and Lipschitz, and \(\psi(u)\to\pm1\) as \(u\to\pm\infty\). Its
derivative \(\psi'\) is even and nonincreasing on \([0,\infty)\), and
\begin{equation}
  0\le\psi'(u)\le\psi'(0)<\infty,\qquad
  \int_{\mathbb R}\psi'(u)\,du=2,\qquad
  |u|\,\psi'(u)\le1\quad\text{for every }u\in\mathbb R.
  \label{eq:mp-psi-properties}
\end{equation}
\end{lemma}

\begin{proof}
Exchanging the two coordinates of the score exchanges the two
coordinates of the response, so \(\psi\) is odd. When both
coordinates of the response are positive, differentiating the KKT
equation in \(u\) gives
\begin{equation}
  \psi'(u)=\frac{2}{\theta''(y_1)+\theta''(y_2)},
  \qquad y_1=\frac{1+\psi(u)}2,\quad y_2=\frac{1-\psi(u)}2 .
  \label{eq:mp-psi-derivative}
\end{equation}
For negative entropy this is
\(\psi'(u)=2y_1y_2=\tfrac12\operatorname{sech}^2(u/2)\). For
Euclidean it is \(1\) when \(|u|<1\), and \(\psi'(u)=0\) when
\(|u|>1\), where the response is a vertex. For negative Tsallis it is
\(2/\bigl(q(y_1^{q-2}+y_2^{q-2})\bigr)\). In each case \(\psi'\) is
nonnegative and bounded, so \(\psi\) is nondecreasing and Lipschitz.
The limits \(\pm1\) are clear for the first two formulas. For
negative Tsallis, \(\theta'(0^+)=-\infty\), so the left side of the
KKT equation tends to \(+\infty\) exactly when \(y_2\to0\).

Next we show that \(\psi'\) is nonincreasing on \([0,\infty)\). The
second derivative \(\theta''(p)\) equals \(1/p\), \(1\), and
\(qp^{q-2}\) for the three kernels, and each of these is convex in
\(p\). So \(p\mapsto\theta''(p)+\theta''(1-p)\) is convex and
symmetric about \(1/2\), hence nondecreasing on \([1/2,1)\). As
\(u\) increases from zero, \(y_1\) increases from \(1/2\). By
\eqref{eq:mp-psi-derivative}, \(\psi'\) is nonincreasing as long as
both coordinates are positive, and for Euclidean it drops to zero
once the response is a vertex. So \(\psi'\) is nonincreasing on
\([0,\infty)\). It is even because \(\psi\) is odd, so its maximum
is \(\psi'(0)\).

The integral of \(\psi'\) equals
\(\psi(+\infty)-\psi(-\infty)=2\). Finally, for \(u>0\), the
monotonicity of \(\psi'\) gives
\[
  u\,\psi'(u)\le\int_0^u\psi'(r)\,dr\le\int_0^\infty\psi'(r)\,dr=1 .
  \qedhere
\]
\end{proof}

Next we express the payoff through \(\psi\). For a trajectory
\(x_{<T}\), write
\[
  \alpha_t:=x_{t,1}-x_{t,2}\in[-1,1],\qquad
  d_t:=2\sum_{\tau<t}\alpha_\tau .
\]
In matching pennies \(A^\top x=(x_1-x_2)\,(1,-1)\). Since \(B=-A\),
the cumulative score \(S_t=-\sum_{\tau<t}A^\top x_\tau\) has
\(S_{t,1}-S_{t,2}=-d_t\). The two coordinates of the learner's score
\(\eta_t(S_t+Z_t)\) therefore differ by \(\eta_t(W_t-d_t)\). The
payoff is \(x^\top Ay=(x_1-x_2)(y_1-y_2)\). Hence
\begin{equation}
  x_t^\top Ay_t=\alpha_t\,\psi\bigl(\eta_t(W_t-d_t)\bigr).
  \label{eq:mp-payoff}
\end{equation}

The surplus is an expectation of \eqref{eq:mp-payoff}, so two
averages appear in every step below. The first is the expected
response when the deterministic part of the score difference is
\(v\). The second is its derivative, the expected sensitivity of the
response to a unit change of the score difference:
\begin{equation*}
  \bar\psi_t(v):=\mathbb E\,\psi\bigl(\eta_t(W_t+v)\bigr),
  \qquad
  \chi_t(v):=\eta_t\,\mathbb E\,\psi'\bigl(\eta_t(W_t+v)\bigr).
\end{equation*}
Since \(\psi\) is Lipschitz with derivative \(\psi'\),
\begin{equation}
  \bar\psi_t(b)-\bar\psi_t(a)=\int_a^b\chi_t(v)\,dv
  \qquad\text{for all }a,b\in\mathbb R.
  \label{eq:mp-ftc}
\end{equation}
With \(\val=0\), \eqref{eq:mp-payoff} gives
\begin{equation}
  \Surplus(x_{<T})
     =\sum_{t=0}^{T-1}\alpha_t\,\bar\psi_t(-d_t).
  \label{eq:mp-surplus}
\end{equation}

\paragraph{Step 2: the expected sensitivity.}
The gain of a pair will be an integral of \(\chi_t\) over a unit
interval. So we need the size of \(\chi_t\).

\begin{lemma}[Expected sensitivity]
\label{lem:mp-sensitivity}
There is a constant \(C_\chi<\infty\) such that
\begin{equation}
  \sup_{v\in\mathbb R}\chi_t(v)\le\frac{C_\chi}{\sqrt{t+1}}
  \qquad\text{for every }t\ge0.
  \label{eq:mp-sensitivity-bound}
\end{equation}
If \(0<\rho<1/2\), then for every \(R>0\),
\begin{equation}
  \sup_{|v|\le R}\left|\sqrt t\,\chi_t(v)
        -\frac{2}{\sigma\sqrt{2\pi}}\right|\longrightarrow0
  \qquad\text{as }t\to\infty.
  \label{eq:mp-sensitivity-limit}
\end{equation}
\end{lemma}

The constant in \eqref{eq:mp-sensitivity-limit} comes from the
probability of a score tie. The factor
\(\psi'(\eta_t(W_t+v))\) is non-negligible only when \(W_t+v\) is
within order \(1/\eta_t\) of zero. Near the origin \(W_t\) has about
the density \(1/(\sigma\sqrt{2\pi t})\) of its Gaussian
approximation. So
\[
  \chi_t(v)\approx\frac{1}{\sigma\sqrt{2\pi t}}
      \int_{\mathbb R}\psi'(u)\,du
      =\frac{2}{\sigma\sqrt{2\pi t}} .
\]
The two restrictions on \(\rho\) are both used in this calculation. Since \(\rho>0\),
the window \(1/\eta_t\) grows, so the count does not depend on the
fine structure of the law of \(\xi_\tau\), such as a lattice. Since
\(\rho<1/2\), the window is short compared with \(\sqrt t\), so the
Gaussian density is flat on it.

The proof uses two facts about the sum \(W_t\). The first is a bound
that holds for every \(t\): there is a constant \(C_0<\infty\) such
that, for all \(t\ge0\), \(a\in\mathbb R\), and \(\ell\ge0\),
\begin{equation}
  \Pr(a\le W_t\le a+\ell)\le\frac{C_0(1+\ell)}{\sqrt{t+1}} .
  \label{eq:mp-concentration}
\end{equation}
This is the Kolmogorov--Rogozin inequality for sums of i.i.d.\
nondegenerate random
variables~\citep{rogozin1961estimate,esseen1966kolmogorov}. The
second fact supplies the exact coefficient. Let \(\ell_t\to\infty\) with
\(\ell_t=o(\sqrt t)\). Then, for every \(R>0\),
\begin{equation}
  \sup_{|v|\le R}\left|\frac{\sqrt t}{\ell_t}
        \Pr\bigl(|W_t+v|\le\ell_t\bigr)
        -\frac{2}{\sigma\sqrt{2\pi}}\right|\longrightarrow0 .
  \label{eq:mp-local-limit}
\end{equation}
It follows from the local limit theorem for i.i.d.\ sums with finite
variance~\citep{shepp1964local,stone1965local}, which has a nonlattice
form and a lattice form.

We derive \eqref{eq:mp-local-limit} in the two cases. Suppose first
that the law of \(\xi_\tau\) is nonlattice. Then an interval of
length one whose distance from the origin is \(o(\sqrt t)\) has
probability \((1+o(1))/(\sigma\sqrt{2\pi t})\), uniformly over such
intervals~\citep[Theorem~3.5.4]{durrett2019probability}. Split the
window \(\{w:|w+v|\le\ell_t\}\) into \(\lfloor2\ell_t\rfloor\)
intervals of length one and a remainder of length less than one. All
of them are within distance \(R+\ell_t=o(\sqrt t)\) of the origin.
The remainder has probability at most \(2C_0/\sqrt{t+1}\) by
\eqref{eq:mp-concentration}. Summing gives
\[
  \Pr\bigl(|W_t+v|\le\ell_t\bigr)
    =\frac{2\ell_t}{\sigma\sqrt{2\pi t}}\,(1+o(1))
      +O\bigl(t^{-1/2}\bigr),
\]
which is \eqref{eq:mp-local-limit} because \(\ell_t\to\infty\). Now
suppose that the law of \(\xi_\tau\) is lattice with span
\(\delta_\xi\). Then each point of the lattice of \(W_t\) within
distance \(o(\sqrt t)\) of the origin has probability
\(\delta_\xi(1+o(1))/(\sigma\sqrt{2\pi t})\),
uniformly~\citep[Theorem~3.5.3]{durrett2019probability}. The window
contains \(2\ell_t/\delta_\xi+O(1)\) such points, and the same
display follows.

\begin{proof}[Proof of \cref{lem:mp-sensitivity}]
We turn the expectation of \(\psi'\) into probabilities of windows.
By \cref{lem:mp-scalar-response}, \(\psi'\) is even and
nonincreasing on \([0,\infty)\). So for each level
\(y\in(0,\psi'(0))\) the set \(\{u:\psi'(u)>y\}\) is an interval
centered at zero. Let \(r(y)>0\) be its half-width. Then
\begin{equation}
  \psi'(u)=\int_0^{\psi'(0)}\mathbf 1\{\psi'(u)>y\}\,dy,
  \qquad
  \int_0^{\psi'(0)}2r(y)\,dy=\int_{\mathbb R}\psi'(u)\,du=2 .
  \label{eq:mp-layer}
\end{equation}
Taking expectations in the first identity,
\begin{equation}
  \chi_t(v)=\eta_t\int_0^{\psi'(0)}
      \Pr\bigl(\psi'(\eta_t(W_t+v))>y\bigr)\,dy ,
  \label{eq:mp-chi-layer}
\end{equation}
and the event inside contains
\(\{|W_t+v|<r(y)/\eta_t\}\) and is contained in
\(\{|W_t+v|\le r(y)/\eta_t\}\).

We first prove \eqref{eq:mp-sensitivity-bound}. The larger event is
a window of length \(2r(y)/\eta_t\). By \eqref{eq:mp-concentration}
and \eqref{eq:mp-layer},
\[
  \chi_t(v)\le\eta_t\int_0^{\psi'(0)}
       \frac{C_0\bigl(1+2r(y)/\eta_t\bigr)}{\sqrt{t+1}}\,dy
     =\frac{C_0\bigl(\eta_t\psi'(0)+2\bigr)}{\sqrt{t+1}} .
\]
Since \(\eta_t\le\eta_0\), the bound holds with
\(C_\chi:=C_0(\eta_0\psi'(0)+2)\).

We now prove \eqref{eq:mp-sensitivity-limit}. Fix
\(y\in(0,\psi'(0))\) and put \(\ell_t=r(y)/\eta_t\). Then
\(\ell_t\to\infty\) because \(\eta_t\to0\), and
\(\ell_t=o(\sqrt t)\) because \(\eta_t\sqrt t\to\infty\), which is
the condition \(\rho<1/2\). So \eqref{eq:mp-local-limit} gives
\[
  \eta_t\sqrt t\,\Pr\bigl(|W_t+v|\le r(y)/\eta_t\bigr)
     \longrightarrow\frac{2r(y)}{\sigma\sqrt{2\pi}}
\]
uniformly in \(|v|\le R\). The same limit holds with strict
inequality, because the two events differ by two points, each of
probability at most \(C_0/\sqrt{t+1}\). Hence the integrand of
\eqref{eq:mp-chi-layer}, multiplied by \(\sqrt t\), converges to
\(2r(y)/(\sigma\sqrt{2\pi})\) for every \(y\). By
\eqref{eq:mp-concentration} it is at most \(C_0(\eta_0+2r(y))\),
which is integrable over \((0,\psi'(0))\). Dominated convergence and
\eqref{eq:mp-layer} give
\[
  \sqrt t\,\chi_t(v)\longrightarrow
     \int_0^{\psi'(0)}\frac{2r(y)}{\sigma\sqrt{2\pi}}\,dy
     =\frac{2}{\sigma\sqrt{2\pi}} ,
\]
and the convergence is uniform in \(|v|\le R\) because the pointwise
limits and the dominating function are.
\end{proof}

\paragraph{Step 3: the expected response at a tie changes slowly.}
On an even round of the alternating strategy the deterministic score
difference is zero. The expected response there is
\(\beta_t:=\bar\psi_t(0)=\mathbb E\,\psi(\eta_tW_t)\).
If the law of \(\xi_\tau\) is symmetric, then \(\beta_t=0\) because
\(\psi\) is odd. Under \eqref{eq:finite-variance-noise} the law can
be skewed, and then \(\beta_t\) need not be zero. A pair compares
round \(2s\) with round \(2s+1\). So what we need is that
\(\beta_t\) is almost the same on two adjacent rounds.

\begin{lemma}[Adjacent rounds]
\label{lem:mp-adjacent}
If \(0<\rho<1/2\), then \(\beta_{t+1}-\beta_t=o(t^{-1/2})\).
\end{lemma}

\begin{proof}
Two things change from round \(t\) to round \(t+1\): the learning
rate, and the noise, which receives one more increment. We treat
them one at a time:
\[
  \beta_{t+1}-\beta_t
   =\underbrace{\mathbb E\bigl[\psi(\eta_{t+1}W_{t+1})
                    -\psi(\eta_tW_{t+1})\bigr]}_{\text{learning rate}}
    +\underbrace{\mathbb E\bigl[\psi(\eta_tW_{t+1})
                    -\psi(\eta_tW_t)\bigr]}_{\text{new increment}} .
\]

\emph{The learning rate.} Fix \(w\in\mathbb R\). The derivative of
\(\lambda\mapsto\psi(\lambda w)\) is \(w\,\psi'(\lambda w)\), and
\(|w\,\psi'(\lambda w)|\le1/\lambda\) by the last bound in
\eqref{eq:mp-psi-properties}. Integrating over
\(\lambda\in[\eta_{t+1},\eta_t]\) gives
\[
  \bigl|\psi(\eta_{t+1}w)-\psi(\eta_tw)\bigr|
     \le\log\frac{\eta_t}{\eta_{t+1}}
     =\rho\log\frac{t+2}{t+1}\le\frac{\rho}{t+1}.
\]
So the first term is \(O(1/t)\), which is \(o(t^{-1/2})\).

\emph{The new increment.} We have \(W_{t+1}=W_t+\xi_t\), and
\(\xi_t\) is independent of \(W_t\). Conditioning on \(\xi_t\) and
using \eqref{eq:mp-ftc},
\[
  \mathbb E\bigl[\psi(\eta_tW_{t+1})-\psi(\eta_tW_t)\bigr]
    =\mathbb E\bigl[\bar\psi_t(\xi_t)-\bar\psi_t(0)\bigr]
    =\mathbb E\int_0^{\xi_t}\chi_t(v)\,dv .
\]
The local sensitivity estimate makes \(\chi_t(0)\,\xi_t\) the
first-order term. Its mean is zero, so we subtract \(\chi_t(0)\,\mathbb E\xi_t=0\):
\[
  \mathbb E\int_0^{\xi_t}\chi_t(v)\,dv
    =\mathbb E\int_0^{\xi_t}\bigl(\chi_t(v)-\chi_t(0)\bigr)\,dv .
\]
For each fixed value of \(\xi_t\),
\[
  \sqrt t\,\left|\int_0^{\xi_t}
        \bigl(\chi_t(v)-\chi_t(0)\bigr)\,dv\right|
   \le|\xi_t|\sup_{|v|\le|\xi_t|}
        \sqrt t\,\bigl|\chi_t(v)-\chi_t(0)\bigr| .
\]
The right side tends to zero by \eqref{eq:mp-sensitivity-limit} with
\(R=|\xi_t|\), and it is at most \(2C_\chi|\xi_t|\) by
\eqref{eq:mp-sensitivity-bound}. Since \(\mathbb E|\xi_t|<\infty\)
and the law of \(\xi_t\) does not depend on \(t\), dominated
convergence shows that the second term is \(o(t^{-1/2})\).
\end{proof}

A new increment shifts the score difference by \(\xi_t\). To first
order the expected response changes by \(\chi_t(0)\,\xi_t\), which
averages to zero. The rest is small because the sensitivity
\(\chi_t(v)\) hardly depends on \(v\). Symmetry of the noise is not
used.

\paragraph{Step 4: the alternating strategy.}
For \(x^{\rm alt}_{<T}\) we have \(\alpha_{2s}=\tfrac12\) and
\(\alpha_{2s+1}=-\tfrac12\), so \(d_{2s}=0\) and \(d_{2s+1}=1\). By
\eqref{eq:mp-surplus}, pair \(s\) contributes
\[
  \mathcal G_s:=\tfrac12\bar\psi_{2s}(0)
                 -\tfrac12\bar\psi_{2s+1}(-1)
     =\tfrac12\beta_{2s}-\tfrac12\bar\psi_{2s+1}(-1).
\]
The two terms differ in two ways: the round, and the deterministic
shift. We separate them by adding and subtracting
\(\tfrac12\beta_{2s+1}\):
\begin{equation*}
  \mathcal G_s
    =\tfrac12\bigl(\beta_{2s}-\beta_{2s+1}\bigr)
     +\tfrac12\bigl(\bar\psi_{2s+1}(0)-\bar\psi_{2s+1}(-1)\bigr).
\end{equation*}
The first term is \(o(s^{-1/2})\) by \cref{lem:mp-adjacent}.
The second term is the effect of the unit shift that playing \(x'\)
creates. By \eqref{eq:mp-ftc} and \eqref{eq:mp-sensitivity-limit}
with \(R=1\),
\[
  \bar\psi_{2s+1}(0)-\bar\psi_{2s+1}(-1)
     =\int_{-1}^{0}\chi_{2s+1}(v)\,dv
     =\frac{2}{\sigma\sqrt{2\pi(2s+1)}}\,(1+o(1)).
\]
Therefore
\begin{equation}
  \mathcal G_s=\frac{1}{\sigma\sqrt{2\pi(2s+1)}}\,(1+o(1))
  \qquad\text{as }s\to\infty.
  \label{eq:mp-pair-gain}
\end{equation}
Thus each late pair has gain of order \(t^{-1/2}\), with coefficient
\(1/(\sigma\sqrt{2\pi})\) at \(t=2s+1\). The probability of lying
in the responsive window has canceled the learning-rate factor.

It remains to sum over the pairs. An unmatched final round plays
\(x^\star\), which has \(\alpha_t=0\) and earns zero by
\eqref{eq:mp-payoff}. Comparing with the integral of
\((2u+1)^{-1/2}\),
\[
  \sum_{s=0}^{\lfloor T/2\rfloor-1}(2s+1)^{-1/2}
     =\sqrt{2\lfloor T/2\rfloor}+O(1)=\sqrt T+O(1).
\]
The relative errors \(o(1)\) in \eqref{eq:mp-pair-gain} contribute
\(o(\sqrt T)\): for any tolerance, all late pairs are within the
tolerance, and the finitely many early pairs contribute a constant.
This proves \eqref{eq:movement-below-half-alt}.

\paragraph{Step 5: the upper bound.}
Now fix any trajectory \(x_{<T}\) that is independent of the noise.
If it is randomized, we condition on it, so we take it
deterministic. Such a trajectory has no pair structure. We use
instead a potential that measures how far the optimizer has moved
the score difference away from a tie. Define
\[
  \Psi_t(v):=\int_0^v\psi(\eta_tu)\,du,\qquad
  \Phi_t(v):=\mathbb E\,\Psi_t(v-W_t).
\]
The potential \(\Phi_t\) has four properties.
\begin{enumerate}
  \item \(\Phi_t\ge0\) and \(\Phi_0(0)=0\). Indeed, \(\psi\) is odd
  and nondecreasing, so
  \(\Psi_t(v)=\int_0^{|v|}|\psi(\eta_tu)|\,du\ge0\), and \(W_0=0\).
  \item \(\Psi_{t+1}\le\Psi_t\). Indeed, \(\eta_{t+1}\le\eta_t\),
  and \(|\psi(\eta u)|\) is nondecreasing in \(\eta\).
  \item \(\Phi_t'(v)=\mathbb E\,\psi(\eta_t(v-W_t))=-\bar\psi_t(-v)\),
  because \(\psi\) is odd. So by \eqref{eq:mp-surplus} the expected
  payoff of round \(t\) is
  \(\alpha_t\bar\psi_t(-d_t)=-\alpha_t\Phi_t'(d_t)\).
  \item \(\Phi_t'\) is Lipschitz with constant
  \(L_t:=\sup_v\chi_t(v)\le C_\chi/\sqrt{t+1}\), by
  \eqref{eq:mp-ftc} and \eqref{eq:mp-sensitivity-bound}. Hence, for
  all \(v,u\in\mathbb R\),
  \begin{equation}
    \Phi_t(v+u)\le\Phi_t(v)+u\,\Phi_t'(v)+\frac{L_t}{2}u^2 .
    \label{eq:mp-descent}
  \end{equation}
\end{enumerate}
Round \(t\) changes the pair \((\Phi_t,d_t)\) into
\((\Phi_{t+1},d_{t+1})\) in two moves: the optimizer moves \(d_t\),
and the noise receives a new increment.

\emph{The optimizer's move.} Apply \eqref{eq:mp-descent} with
\(v=d_t\) and \(u=d_{t+1}-d_t=2\alpha_t\):
\(\Phi_t(d_{t+1})\le\Phi_t(d_t)+2\alpha_t\Phi_t'(d_t)+2L_t\alpha_t^2 \).
Rearranging,
\begin{equation}
  -\alpha_t\Phi_t'(d_t)
    \le\tfrac12\bigl(\Phi_t(d_t)-\Phi_t(d_{t+1})\bigr)+L_t\alpha_t^2 .
  \label{eq:mp-optimizer-move}
\end{equation}
The left side is the expected payoff of round \(t\). So the
optimizer earns in a round only by moving the score difference down
the potential, up to a curvature term of order \(L_t\).

\emph{The noise's move.} By the second property and
\(W_{t+1}=W_t+\xi_t\), with \(\xi_t\) independent of \(W_t\),
\[
  \Phi_{t+1}(v)=\mathbb E\,\Psi_{t+1}(v-\xi_t-W_t)
     \le\mathbb E\,\Psi_t(v-\xi_t-W_t)
     =\mathbb E\,\Phi_t(v-\xi_t).
\]
Apply \eqref{eq:mp-descent} with \(u=-\xi_t\) and take the
expectation. The first-order term vanishes because
\(\mathbb E\xi_t=0\), and \(\mathbb E\xi_t^2=\sigma^2\). So
\begin{equation}
  \Phi_{t+1}(v)\le\Phi_t(v)+\frac{\sigma^2L_t}{2}
  \qquad\text{for every }v\in\mathbb R.
  \label{eq:mp-noise-move}
\end{equation}
A new increment raises the potential by at most \(\sigma^2L_t/2\).

\emph{Telescoping.} Use \eqref{eq:mp-noise-move} at \(v=d_{t+1}\)
in \eqref{eq:mp-optimizer-move}:
\[
  -\alpha_t\Phi_t'(d_t)
    \le\tfrac12\bigl(\Phi_t(d_t)-\Phi_{t+1}(d_{t+1})\bigr)
        +L_t\Bigl(\alpha_t^2+\frac{\sigma^2}{4}\Bigr).
\]
Sum over \(t<T\). The potential terms telescope,
\(\Phi_0(d_0)=\Phi_0(0)=0\), \(\Phi_T\ge0\), and
\(\alpha_t^2\le1\). Therefore
\[
  \Surplus(x_{<T})
    \le\Bigl(1+\frac{\sigma^2}{4}\Bigr)\sum_{t=0}^{T-1}L_t
    \le\Bigl(1+\frac{\sigma^2}{4}\Bigr)C_\chi
          \sum_{t=0}^{T-1}\frac1{\sqrt{t+1}}
    \le2\Bigl(1+\frac{\sigma^2}{4}\Bigr)C_\chi\sqrt T .
\]
The constant does not depend on the trajectory. This proves
\eqref{eq:movement-below-half-upper} and completes the proof of
\cref{prop:movement-below-half}.

\end{document}